\documentclass[11pt]{article}
\usepackage{eurosym}
\usepackage[english]{babel}
\usepackage[margin=1in]{geometry}
\usepackage{setspace}
\usepackage[authoryear,longnamesfirst]{natbib}
\usepackage{subcaption}
\usepackage{amsmath}
\usepackage{amsfonts}
\usepackage{amssymb}
\usepackage{mathrsfs}
\usepackage{graphicx}
\usepackage{epstopdf}
\usepackage{bm}
\usepackage{algorithm}
\usepackage{algpseudocode}
\usepackage{comment}
\usepackage{ifthen}
\usepackage{accents}
\usepackage{url}
\usepackage{enumitem}
\usepackage{multirow}
\usepackage{multicol}
\usepackage{array}
\usepackage{rotating}
\usepackage{longtable}
\usepackage{float}
\usepackage{booktabs}
\usepackage{lscape}
\usepackage[expansion=false]{microtype}
\usepackage{xcolor}
\definecolor{revcolor}{rgb}{1.0, 0.49, 0.0}
\usepackage{hyperref}

\renewcommand{\baselinestretch}{1.4}
\selectfont
\allowdisplaybreaks
\DeclareMathOperator*{\argmin}{argmin}

\DeclareMathOperator*{\plim}{plim}
\hypersetup{
colorlinks=true,   
linkcolor=red,     
citecolor=blue,    
urlcolor=blue      }
\newtheorem{theorem}{Theorem}

\newtheorem{lemma}{Lemma}

\newtheorem{corollary}[theorem]{Corollary}
\newenvironment{proof}[1][Proof]{\noindent\textbf{#1.} }{\ \rule{0.5em}{0.5em}}
\newtheorem{innercustomass}{Assumption}
\providecommand{\func}[1]{\mathop{\rm #1}\nolimits}
\newenvironment{namedassumption}[1]{
\begingroup 
\renewcommand\theinnercustomass{#1}
\innercustomass}{
\endinnercustomass
\endgroup }

\newcommand{\indep}{\perp \!\!\! \perp}
\begin{document}

\title{Difference-in-Differences Meets Synthetic Control: Doubly Robust
Identification and Estimation}
\author{Yixiao Sun\thanks{%
Department of Economics, University of California, San Diego, USA. Email: 
\href{mailto:yisun@ucsd.edu}{yisun@ucsd.edu}.} \and Haitian Xie\thanks{%
Corresponding author. Guanghua School of Management, Peking University,
China. Email: \href{mailto:xht@gsm.pku.edu.cn}{xht@gsm.pku.edu.cn}.} \and %
Yuhang Zhang\thanks{%
Department of Statistics and Data Science, National University of Singapore.
Email: \href{mailto:zhangyh3141@gmail.com}{zhangyh3141@gmail.com}.}}
\date{\today}
\maketitle

\begin{abstract}
\singlespacing Difference-in-Differences (DiD) and Synthetic Control (SC)
are widely used methods for causal inference in panel data, each with
distinct strengths and limitations. We propose a novel method for
short-panel causal inference that integrates the advantages of both
approaches. Our method delivers a doubly robust identification strategy for
the average treatment effect on the treated (ATT) under either of two
non-nested assumptions: parallel trends or a group-level SC condition.
Building on this identification result, we develop a unified semiparametric
framework for estimating the ATT. Notably, the identification-robust moment
function satisfies Neyman orthogonality under the parallel trends assumption
but not under the SC assumption, leading to different asymptotic variances
across the two identification strategies. To ensure valid inference, we
propose a multiplier bootstrap method that consistently approximates the
asymptotic distribution under either assumption. Furthermore, we extend our
methodology to accommodate repeated cross-sectional data and staggered
treatment designs. As an empirical application, we evaluate the impact of
the 2003 minimum-wage increase in Alaska on family income. Finally, in
simulation studies based on empirically calibrated data-generating
processes, we demonstrate that the proposed estimation and inference methods
perform well in finite samples under either identification assumption.

\bigskip \noindent \textbf{Keywords}: Difference-in-Differences, Doubly
Robust Identification, Multiplier Bootstrap, Semiparametric Estimation,
Synthetic Control.

\bigskip \noindent \textbf{JEL codes:} C14, C21, C23
\end{abstract}

\vspace{-2.5em}

\vspace{-3em}

\vspace{-2em}

\newpage

\section{Introduction}

In recent years, Difference-in-Differences (DiD) methods have gained
significant traction in economics and the social sciences, playing a central
role in the \textquotedblleft credibility revolution\textquotedblright\ %
\citep{AngristPischke2010_credibility}. Over 30\% of NBER applied
microeconomics working papers in 2024 mention DiD or related event study
methods, surpassing alternatives such as instrumental variables and
regression discontinuity designs \citep{goldsmith2024tracking}. Yet, DiD's
reliance on the parallel trends assumption, a requirement often criticized
as an oversimplification of real-world dynamics, leaves researchers
vulnerable to biased estimates when trends diverge.

The synthetic control method --- \textquotedblleft arguably the most
important innovation in the policy-evaluation literature in the last 15
years\textquotedblright\ \citep{athey2017state} --- constructs
counterfactuals by matching treated units to weighted combinations of
untreated units. It has surged in popularity in settings where the parallel
trends assumption may fail. However, its use remains largely confined to
studies with long panels of aggregate data, such as countries, states, or
regions 
\citep{abadie2003economic, abadie2010synthetic, hsiao2012panel,
abadie2015comparative}, rather than to studies with micro-level data, which
offer rich individual information but are typically observed over far
shorter horizons. The method has therefore not yet realized its full
potential in applied microeconomic research.

In this paper, we address these limitations by adopting a different
perspective on the synthetic control problem and integrating the strengths
of DiD and synthetic control methods. We focus on a typical microeconomic
panel data setting where individuals are observed repeatedly over time and
grouped into aggregate-level units (such as households within states), with
treatment assigned at the group level. We introduce a novel causal inference
method that unifies the identification strategies of DiD and synthetic
control in a doubly robust framework. Our approach nonparametrically
identifies the average treatment effect on the treated (ATT) under either
the DiD parallel trends assumption or the synthetic control assumption. This
robustness enables applied researchers to avoid the conventional trade-off
between DiD and synthetic control, thereby strengthening the credibility of
causal inference. For easy reference, we call our method the
Difference-in-Differences Synthetic Control (DiD-SC) method.

To operationalize our method, we propose a semiparametric estimation
procedure and a bootstrap inference approach.\footnote{%
We use \textquotedblleft semiparametric\textquotedblright\ to describe a
setting in which the parameters of interest are finite-dimensional, while
nuisance parameters, such as the propensity score, conditional outcome
expectations, and synthetic control weights, are specified nonparametrically.%
} Developing asymptotic theory within this framework presents an important
challenge because the two identifying assumptions imply fundamentally
different first-order behavior. Under the parallel trends assumption, the
proposed moment condition satisfies Neyman orthogonality, allowing flexible
estimation of the nuisance parameters \citep{chernozhukov2018DML}. Under the
synthetic control assumption, however, the moment condition is not Neyman
orthogonal, and estimation of the synthetic control weights contributes to
the first-order asymptotic distribution. We explicitly characterize these
additional first-stage effects and derive the corresponding adjustment to
the asymptotic variance. Despite this asymmetry, we develop a unified
semiparametric estimation procedure and multiplier bootstrap inference
method and establish their validity under either identifying assumption,
while allowing the technical conditions on the nuisance estimators to differ
across the two regimes.

We further extend the method to accommodate two additional settings:
repeated cross-sectional data and staggered treatment designs. For repeated
cross-sectional data, a different set of individuals is sampled in each time
period. We show that the doubly robust identification results from the panel
setting extend naturally to the repeated cross-sectional setting, provided
that the joint distribution of the variables is time-invariant, a standard
assumption in this context. For staggered treatment designs, we build on the
recent DiD literature \citep[e.g.,][]{callaway2021difference}, using
untreated groups as controls for the treatment group of interest. This
approach effectively reduces the problem to a setting with a single
treatment group and a single treatment period.

In our empirical study, we apply the proposed method to evaluate the effect
of Alaska's 2003 minimum-wage increase on family income, using the Current
Population Survey with states as groups and households as units. Consistent
with \citet{gunsilius2023distributional}, we find no statistically
significant immediate impact of the policy on family income.

In the simulation study, we calibrate the data-generating processes (DGPs)
to match the distribution observed in the empirical application. We consider
three DGPs: (i) only the synthetic control assumption holds, (ii) only the
parallel trends assumption holds, and (iii) both assumptions hold. The
proposed estimation and inference procedures demonstrate consistently good
performance across all three DGPs. We further conduct a broad set of
robustness and comparison exercises, including sensitivity analyses, local
misspecification designs, alternative inference procedures, and comparisons
with existing DiD and synthetic control methods.

\paragraph{Related literature}

Our paper contributes to three strands of literature. First, it contributes
to the modern semiparametric panel data and DiD literature 
\citep{heckman1997matching,abadie2005semiparametric, sant2020doubly,
callaway2021difference, chen2025efficient}. When the parallel trends
assumption holds, our estimator can be interpreted as a weighted average of
doubly robust DiD estimators across different groups, inheriting their
robustness to parametric misspecification of nuisance parameters and their
compatibility with double/debiased machine learning (DML) frameworks %
\citep{chernozhukov2018DML}. More importantly, our method introduces a novel
dimension of robustness: even when parallel trends fail due to unobserved
confounding, our estimator remains valid if the synthetic control structure
holds. This innovation addresses concerns about the fragility of parallel
trends in applied work, offering researchers a safeguard against violations
of this key assumption.\footnote{%
Relatedly, \citet{callaway2025beyond} develop an
identification-strategy-robust approach based on selecting close comparison
groups. Their approach differs from ours in relying on individual comparison
groups rather than a synthetic combination of multiple control groups.}

Second, our paper contributes to the growing literature on synthetic control
methods, particularly in the context of micro-level data with a short time
dimension. While classical synthetic control approaches have predominantly
focused on aggregate-level data, recent studies have begun to explore
applications at the individual level. For example, relying on factor models, 
\cite{robbins2017framework} examined the impact of region-level crime
interventions using individual-level data, whereas \cite%
{chen2020distributional} introduced a distributional synthetic control
method based on a quantile factor model. More recently, nonparametric
approaches, such as those proposed by \cite{gunsilius2023distributional} and 
\cite{chen2023group}, have been developed to study the synthetic matching of
the entire outcome distribution by drawing on the changes-in-changes method
from \cite{athey2006identification}. While \cite{gunsilius2023distributional}%
's method applies to short-panel settings, \cite{chen2023group} require a
large number of pre-treatment periods. Complementing these developments, our
approach integrates synthetic control methods within the modern causal
inference framework, requiring only a short time dimension. It achieves
nonparametric identification and semiparametric estimation results that
align with those in recent DiD and other causal methods, thereby broadening
the scope of synthetic control applications in empirical research.

Third, our work extends the literature on causal panel data methods, such as
the matrix completion method \citep{athey2021matrix}, the synthetic DiD
method \citep{arkhangelsky2021synthetic, arkhangelsky2024sequential}, the
augmented SC method \citep{ben2021augmented}, and work on doubly robust
methods in panel settings \citep{arkhangelsky2022doubly}.\footnote{\cite%
{arkhangelsky2022doubly} study doubly robust identification in causal panel
models under parametric or semiparametric structure. Our framework differs
in that it focuses on robustness across nonparametric identification
assumptions.} These papers operate on aggregate-level data and, to drive
bias to zero, typically assume a large number of control groups and/or a
large number of pre-treatment periods.\footnote{\cite%
{arkhangelsky2024sequential} adopt a fixed $T$ asymptotic framework with
large cross-sectional samples, but their method relies on aggregating
individual-level outcomes into cohort-level averages.}

In contrast, a first key difference of our approach is that we directly work
with short panels of individual-level data organized into heterogeneous
groups. A prototypical example of the group structure we consider is the
classification of individuals or households by their state of residence, a
common practice in state-level policy studies that rely on widely used
microeconomic datasets, such as the Current Population Survey, American
Community Survey, Panel Study of Income Dynamics, and National Longitudinal
Surveys. By leveraging this group structure with granular information, we
can identify the synthetic control weights without relying on either a large
number of pre-treatment periods or a large number of aggregate control
groups. Instead, we assume that the number of individuals in each group is
large, and thus our asymptotic framework is different. A more detailed
discussion is provided in Section \ref{sec:no-covariate}.

A second key difference of our framework is that it incorporates covariates
in a fully nonparametric manner by imposing the identification assumptions
conditional on covariates. This contrasts with existing causal panel
methods, which typically incorporate covariates through linear projection
residualization or other parametric structures.

\paragraph{Organization of the paper}

Section \ref{sec:setup} introduces the panel data model, the identification
assumptions, and the doubly robust identification result. Section \ref%
{sec:asymptotics} establishes the semiparametric estimation theory and the
multiplier bootstrap theory under the parallel trends and synthetic control
assumptions separately. Sections \ref{sec:rc} and \ref{sec:staggered} extend
the analysis to repeated cross-sectional data and staggered treatment
designs, respectively. Sections \ref{sec:empirical} and \ref{sec:simulation}
present the empirical and simulation studies. Section \ref{sec:conclusion}
concludes the paper. Proofs of the results in Sections \ref{sec:setup} and %
\ref{sec:asymptotics} are provided in the appendix; the remaining proofs and
additional materials are collected in the Supplementary Appendix.

\section{Setup and identification}

\label{sec:setup}

\subsection{Panel model with a single treated group}

We initiate our investigation with a panel data model, where we observe
repeated outcomes over time for the same individuals. Our dataset comprises
a collection of groups, which represent aggregate-level units, such as
countries, states, provinces, cities, or other similar entities. Within each
aggregate-level unit, there are many individual units.\footnote{%
The asymptotic framework we consider assumes that the number of individuals
in each group approaches infinity.} The categorical variable $G$ signifies
the specific group to which an individual unit belongs. The first group, $g=1
$, is the treated group, and the remaining groups, $g=2,\ldots ,N_{G}+1,$
are the control groups, so that $N_{G}$ is the total number of control
groups. Following the synthetic control literature, we use the terms
\textquotedblleft control\textquotedblright\ and \textquotedblleft
donor\textquotedblright\ interchangeably. Individual units are observed over 
$\mathcal{T}$ time periods denoted by $t\in \{1,\ldots ,\mathcal{T}\}$.%
\footnote{%
We do not use the usual notation $T$ here because it denotes the time stamp
for individual units in the repeated cross-sectional setting in Section \ref%
{sec:rc}.} For a group $g$, let $\mathcal{G}_{g}\equiv \mathbf{1}\{G=g\}$ be
a binary variable (0 or 1) that indicates whether an individual unit belongs
to group $g$. That is, $\mathcal{G}_{g}$ is a binary indicator of group $g$,
while $G$ is a categorical variable indexing the group.

A binary, irreversible treatment is assigned to the groups. In periods $%
t=1,\dots,\mathcal{T}-1$, all groups remain untreated. In the final period $%
t=\mathcal{T}$, the first group receives treatment, while the other groups
remain untreated. We focus on this scenario with a single treated group and
a single treatment period, as it allows us to more effectively illustrate
the methodological contribution in this simpler setting. An extension to
staggered treatment designs, involving multiple treatment groups and
variation in treatment timing, is presented in Section \ref{sec:staggered}.

The potential outcomes are defined over the entire treatment trajectory, as
is standard in multi-period treatment settings (e.g., \citealp{roth2023whats}%
). Let $\bm{0}_{s}$ denote a vector of zeros of dimension $s$. We define $%
\tilde{Y}_{t}(\bm{0}_{\mathcal{T}-1},1)$ as the potential outcome at period $%
t$ if the individual receives treatment in the final period and no treatment
in all preceding periods. Similarly, we define $\tilde{Y}_{t}(\bm{0}_{%
\mathcal{T}-1},0)\equiv \tilde{Y}_{t}(\bm{0}_{\mathcal{T}})$ as the
potential outcome if the individual never receives treatment. This
definition of potential outcomes will be useful when extending the framework
to account for staggered treatment adoption.

In the scenario of a single treatment time, for simplicity and with a slight
abuse of notation, we use $Y_{t}(1)$ and $Y_{t}(0)$ to represent $\tilde{Y}%
_{t}(\bm{0}_{\mathcal{T}-1},1)$ and $\tilde{Y}_{t}(\bm{0}_{\mathcal{T}-1},0)$%
, respectively. We make a standard no-anticipation assumption as follows:

\begin{namedassumption}	{NA}[No-Anticipation] \label{asm:NA}%
For any pre-treatment period $t\in \{1,\ldots ,\mathcal{T}-1\}$, $%
Y_{t}(1)=Y_{t}(0)$. 
\end{namedassumption}%

Under Assumption \ref{asm:NA}, an individual's observed pre-treatment
outcome is equal to the untreated potential outcome. Therefore, the observed
outcome is given by 
\begin{align*}
Y_{t}& =Y_{t}(0),t=1,\ldots ,\mathcal{T}-1, \\
Y_{\mathcal{T}}& =\mathcal{G}_{1}Y_{\mathcal{T}}(1)+(1-\mathcal{G}_{1})Y_{%
\mathcal{T}}(0).
\end{align*}%
The target parameter is 
\begin{equation*}
\theta \equiv \mathbb{E}[Y_{\mathcal{T}}(1)-Y_{\mathcal{T}}(0)|G=1],
\end{equation*}%
which represents the average treatment effect on the treated (ATT) in the
post-treatment period. For simplicity, we focus on identifying and
estimating this mean parameter, although this approach can be easily
extended to, for example, the distributional change from $Y_{\mathcal{T}}(0)$
to $Y_{\mathcal{T}}(1)$ in the treatment group. This can be achieved by
considering $\mathbb{E}\left[ h(Y_{\mathcal{T}}(1))|G=1\right] -\mathbb{E}%
\left[ h(Y_{\mathcal{T}}(0))|G=1\right] $ for an appropriately chosen
function $h$.

For each individual unit, let $X$ represent a set of time-invariant
covariates that describe the individual's characteristics. Denote the
support of $X$ by $\mathcal{X}$.

	\begin{namedassumption}{PDS}[Panel Data Sampling] \label{asm:panel}
We observe an independent and identically distributed (iid) sample $\{S_i
\equiv (Y_{1i},\ldots,Y_{\mathcal{T}i},G_i,X_i): 1 \leq i \leq n\}$ of $S
\equiv (Y_{1},\ldots,Y_{\mathcal{T}},G,X)$. 
\end{namedassumption}%

Under Assumption \ref{asm:panel}, there are $n$ individual-level units
indexed by $i=1,\ldots ,n$. The number of such units in group $g$ is the
cardinality of the set $\{i:G_{i}=g\}=\{i:\mathcal{G}_{gi}=1\}$.

\subsection{Identification assumptions}

We introduce the following three identification assumptions, with the
understanding that they need not all be imposed simultaneously.

\begin{namedassumption}{O}[Overlap] \label{asm:overlap}
For all groups $g=1,\ldots,N_G+1$, the propensity score $p_g(X) \equiv 
\mathbb{P}(G=g|X) > 0$ almost surely. 
\end{namedassumption}%

\begin{namedassumption}{PT}[Parallel Trends] \label{asm:PT}%
For each control group $g \geq 2$, $\mathbb{E}[Y_{\mathcal{T}}(0) - Y_{%
\mathcal{T}-1}(0)|G=1,X] = \mathbb{E}[Y_{\mathcal{T}}(0) - Y_{\mathcal{T}%
-1}(0)|G=g,X]$ almost surely. 
\end{namedassumption}%

\begin{namedassumption}{SC}[Synthetic Control] \label{asm:SC}%
There exists a vector of weight functions $w \equiv (w_2, \ldots, w_{N_G+1})$
for which the elements sum to one and satisfy that for almost all $x \in 
\mathcal{X}$, 
\begin{align}  \label{eqn:w-id}
\mathbb{E}[Y_{t}(0)|G=1,X=x] = \sum_{g=2}^{N_G+1} w_g(x) \mathbb{E}%
[Y_{t}(0)|G=g,X=x], t=1,\ldots,\mathcal{T}.
\end{align}
Furthermore, $\mathcal{T} \geq N_G$.\footnote{%
Theoretically, imposing nonnegativity on the weights is not necessary in our
framework, although such a restriction is commonly maintained in the
classical synthetic control literature \citep{abadie2021using}. If the true
weights are nonnegative, this structure can be imposed on the weight
estimation without altering the asymptotic theory for the proposed ATT
estimator.} 
\end{namedassumption}%

For each unit with covariate value $X$ in the treatment group, Assumption %
\ref{asm:overlap} ensures that, in large samples, there exist units with the
same covariate value in each control group. It also implies that $\mathbb{P}%
\left( G=g\right) >0$ for each $g=1,\ldots ,N_{G}+1$ so that under iid
sampling, the number of individuals in each group approaches infinity almost
surely as $n\rightarrow \infty .$

Assumption \ref{asm:PT} adopts a nonparametric parallel trends condition
that fits naturally within our framework.\footnote{%
At this stage, ``nonparametric'' means that the identification assumption
does not impose any parametric functional form restrictions, such as the
additive separability inherent in a two-way fixed-effects specification.}
First, it is common practice in empirical research to estimate panel
regressions on individual-level data with group-level (e.g., state- or
county-level) fixed effects and common time trends. Such models inherently
assume that groups share the same trend while allowing for group-specific
fixed effects. Second, concerns about the plausibility of parallel trends
often arise in settings with long pre-treatment windows, as untreated
outcome paths may gradually diverge over time. In contrast, our assumption
applies to a single period and requires a common trend only between the last
pre-treatment period and the treatment period. Compared with the two-way
fixed-effects model, which presumes a common trend over the entire horizon,
this is a markedly weaker temporal restriction. Third, when there are two or
more control groups, the parallel trends assumption allows for
overidentification tests, enabling researchers to empirically assess whether
the control groups follow a parallel trend. Control groups that deviate from
this trend can be excluded based on these tests. Of course, this approach
requires identifying at least one control group that satisfies the parallel
trends assumption, a choice that must rely on domain knowledge rather than
data. Finally, the assumption is made conditional on the covariate $X$.
Therefore, we only need to maintain the parallel trends assumption for the
subset of the treatment group and the corresponding subset of the control
group with the same covariate value. This is more plausible than an
unconditional parallel trends assumption, especially when the covariates
that may affect the outcome variable differ between the treatment and
control groups \citep{abadie2005semiparametric,ghanem2026pretrends}.

Assumption \ref{asm:SC} is the group-level representation of the synthetic
control structure. It requires that, for each covariate value, the
conditional mean of the outcome in the treated group can be represented as a
linear combination of the conditional means for the untreated groups. We
allow the weights to depend on covariates, which accommodates the
possibility that the weighting patterns vary across units with different
covariate values. In other words, the synthetic control structure is imposed
only on subsets of the treatment and control groups that share the same
covariate value, enhancing the assumption's plausibility. We further assume
that the weighting structure remains the same over time, allowing us to
identify the weighting scheme using the pre-treatment data and apply it in
the post-treatment period to identify the ATT. In our setting, both $%
\mathcal{T}$ and $N_{G}$ are treated as fixed, but for the weights $\left\{
w_{g}\left( \cdot \right) \right\} $ to be identified, it is necessary that $%
\mathcal{T}\geq N_{G}$.

One can also motivate Assumption \ref{asm:SC} with a factor model structure,
as in the conventional synthetic control literature:%
\begin{equation*}
Y_{i,t}(0)=\lambda (G_{i},X_{i})^{\prime }F_{t}+\varepsilon _{i,t},\qquad 
\mathbb{E}[\varepsilon _{i,t}|X_{i},G_{i}]=0,
\end{equation*}%
where $F_{t}$ is a vector of common factors, and $\lambda(G_i,X_i)$ is a
vector of factor loadings. Then Assumption \ref{asm:SC} holds if there exist
weighting functions $w_{g}\left( \cdot \right) $ summing to one such that 
\begin{equation*}
\lambda (1,x)=\sum_{g=2}^{N_{G}+1}w_{g}\left( x\right) \lambda \left(
g,x\right),
\end{equation*}%
for almost all $x\in \mathcal{X}.$ Because we impose no parametric form on
the weighting functions $w_{g}\left( \cdot \right) $, the factor loading $%
\lambda (G_{i},X_{i})$ may vary across individuals through their group
membership and observed covariates, and there is no parametric restriction
on $\lambda (\cdot ,\cdot ).$

The structure in Assumption \ref{asm:SC} takes a different approach from the
distributional synthetic control framework in \cite%
{gunsilius2023distributional}. Specifically, we focus on matching the
expectation of the outcome rather than the entire distribution, and we allow
the weights to depend on covariates, while \cite{gunsilius2023distributional}
does not. We require that the number of time periods ($\mathcal{T}$) be at
least as large as the number of control groups ($N_{G}$) to identify the
weights. This differs from the traditional synthetic control method, which
relies on an increasing number of pre-treatment periods for consistent
estimation. Here, we only require $\mathcal{T}\geq N_{G}$, without assuming
that $\mathcal{T}$ diverges to infinity.\footnote{%
If the number of time periods is smaller than the number of control groups,
i.e., $\mathcal{T}<N_{G}$, a natural and practical approach is to leverage
prior studies or domain knowledge to retain only the most relevant control
groups for comparison with the treatment group. An alternative approach is
to incorporate a penalty term in the weight estimation process, similar to
the synthetic DiD method \citep{arkhangelsky2021synthetic}, to enable
estimation even when the weights are not point-identified. However,
incorporating such penalization into semiparametric estimation adds further
complexity, which we leave for future research.} This condition is often met
in practice, such as in monthly datasets like the Current Population Survey,
where states serve as the groups. Unlike the canonical synthetic control
method, our framework does not require an increasing number of time periods
because it leverages a growing number of individual units within each group.

An alternative identification assumption commonly used in the synthetic
control literature is the conditional independence of post-treatment $%
Y_{t}(0)$ and treatment assignment, given the pre-treatment potential
outcomes %
\citep[e.g.,][]{robbins2017framework,ding2019bracketing,kellogg2021combining}%
. This condition, often referred to as ignorability conditional on lagged
outcomes, is relatively strong. In contrast, the parallel trends condition
is conceptually less stringent while still preserving the dynamic nature of
the model.

In summary, Assumptions \ref{asm:PT} and \ref{asm:SC} reflect two different
approaches to identification. Assumption \ref{asm:PT} relies on a stronger
inter-group restriction, although this restriction needs to hold only across
two periods. In contrast, Assumption \ref{asm:SC} allows for a weaker
inter-group restriction but requires it to persist over a longer duration.
Importantly, neither assumption implies the other; the two assumptions are
non-nested. In practice, it is challenging for researchers to decide which
assumption to rely on, as each leads to different estimation methods and
possibly very different estimates.

\subsection{Doubly robust identification}

We introduce the identification results for the ATT. Let $p\equiv
(p_{g}\left( \cdot \right) ,g=1,\ldots ,N_{G}+1)$ and $w\equiv (w_{g}\left(
\cdot \right) ,g=2,\ldots ,N_{G}+1)$. For simplicity, denote $\Delta Y\equiv
\Delta Y_{\mathcal{T}}\equiv Y_{\mathcal{T}}-Y_{\mathcal{T}-1}$. Let $%
m_{\Delta }(X)\equiv \mathbb{E}[Y_{\mathcal{T}}-Y_{\mathcal{T}-1}|G\neq 1,X]$%
, and $\pi _{1}\equiv \mathbb{P}(G=1)$. Note that $\pi _{1}$ is strictly
positive, as implied by Assumption \ref{asm:overlap}. Our identification
strategy is to construct a moment function whose mean is the target causal
parameter $\theta $. We propose the following moment function, where recall
that $S=(Y_{1},\ldots ,Y_{\mathcal{T}},G,X)$ denotes the observed variables: 
\begin{equation}
\phi (S;m_{\Delta },p,w;\pi _{1})\equiv \frac{\mathcal{G}_{1}}{\pi _{1}}%
(\Delta Y-m_{\Delta }(X))-\frac{1}{\pi _{1}}\sum_{g=2}^{N_{G}+1}w_{g}(X)%
\frac{p_{1}(X)}{p_{g}(X)}\mathcal{G}_{g}(\Delta Y-m_{\Delta }(X)).
\label{eqn:phi-expression1}
\end{equation}

Define the \emph{synthetic control contrast} 
\begin{equation*}
\Lambda \left( \mathcal{G},X\right) \equiv \frac{1}{\pi _{1}}\left( \mathcal{%
G}_{1}-\sum_{g=2}^{N_{G}+1}w_{g}(X)\frac{p_{1}(X)}{p_{g}(X)}\mathcal{G}%
_{g}\right) ,
\end{equation*}%
where $\mathcal{G}\equiv (\mathcal{G}_{1},\ldots ,\mathcal{G}_{N_{G}+1})$ and the dependence of $\Lambda $ on $(p,w,\pi _{1})$ is suppressed for
simplicity. Then $\phi (S;m_{\Delta },p,w;\pi _{1})$ can be equivalently
represented as%
\begin{equation}
\phi (S;m_{\Delta },p,w;\pi _{1})=\Lambda \left( \mathcal{G},X\right) \left[
\Delta Y-m_{\Delta }(X)\right] =\Lambda \left( \mathcal{G},X\right) Y_{%
\mathcal{T}}-\Lambda \left( \mathcal{G},X\right) \left[ Y_{\mathcal{T}%
-1}+m_{\Delta }(X)\right] .  \label{eqn:phi-expression2}
\end{equation}
The term $Y_{\mathcal{T}-1}+m_{\Delta }(X)$ is the \emph{DiD forecast} of
the individual's untreated outcome at $\mathcal{T}$: the level $Y_{\mathcal{T%
}-1}$ carried forward by the control groups' average change $m_{\Delta }(X)$%
. Thus $\phi $ applies the contrast $\Lambda \left( \mathcal{G},X\right) $
to the DiD forecast error.

Before showing that $\phi \left( \cdot \right) $ is unbiased for $\theta $
under either Assumption \ref{asm:PT} or Assumption \ref{asm:SC}, we provide
some intuition for its construction. Equations (\ref{eqn:phi-expression1})
and (\ref{eqn:phi-expression2}) are two representations of the same $\phi $.
Each decomposes it into a moment function that is standard under one of the
two identification assumptions, together with an adjustment that has mean
zero under that same assumption.

First, if we maintain only Assumption \ref{asm:PT}, the standard approach is
to estimate the ATT using the DiD moment function $\mathcal{G}_{1}(\Delta
Y-m_{\Delta }(X))/\pi _{1}$, whose mean is $\theta $. Equation (\ref%
{eqn:phi-expression1}) shows that $\phi $ is exactly this moment function
augmented with a synthetic control adjustment. Under Assumption \ref{asm:PT}%
, $\mathbb{E}[\mathcal{G}_{g}(\Delta Y-m_{\Delta }(X))|X]=0$ for each $g\geq
2$, and hence the adjustment has mean zero.

Conversely, suppose we maintain only Assumption \ref{asm:SC}. The
corresponding moment function is then $\Lambda \left( \mathcal{G},X\right)
Y_{\mathcal{T}}$, whose mean is again $\theta $. Equation (\ref%
{eqn:phi-expression2}) shows that $\phi $ is this moment function less the
synthetic control contrast applied to the DiD forecast, and that adjustment
has mean zero under Assumption \ref{asm:SC}. This is to be expected: the
forecast is a prediction of the untreated outcome $Y_{\mathcal{T}}(0)$,
formed from pre-treatment outcomes and covariates alone, and Assumption \ref%
{asm:SC} equates the conditional means of the treated group and of the
weighted control groups for all such quantities.

We can gain further insight into the moment function by considering a
special case. Specifically, in the absence of covariate $X,$ the mean
function $m_{\Delta }$ becomes%
\begin{equation*}
m_{\Delta }=\mathbb{E}\left[ Y_{\mathcal{T}}-Y_{\mathcal{T}-1}|G\neq 1\right]
,
\end{equation*}%
which is the expected change in the outcome of interest from time $\mathcal{T%
}-1$ to $\mathcal{T}$ for the control groups. Then 
\begin{equation*}
\mathbb{E}\left[ \frac{\mathcal{G}_{1}}{\pi _{1}}\left( \Delta Y-m_{\Delta
}\right) \right] =\mathbb{E}\left[ \Delta Y|G=1\right] -\mathbb{E}\left[
\Delta Y|G\neq 1\right] ,
\end{equation*}%
which takes the familiar DiD form. On the other hand, in the absence of $X,$ 
$w_{g}$ and $p_{g}$ become constants, and $p_{1}=\pi _{1}.$ Then 
\begin{align*}
\frac{1}{\pi _{1}}\mathbb{E}\left[ \left( \mathcal{G}_{1}-%
\sum_{g=2}^{N_{G}+1}w_{g}\frac{p_{1}}{p_{g}}\mathcal{G}_{g}\right) Y_{%
\mathcal{T}}\right] & =\frac{1}{\pi _{1}}\mathbb{E}\left[ \mathcal{G}_{1}Y_{%
\mathcal{T}}\right] -\frac{1}{\pi _{1}}\mathbb{E}{\left[
\sum_{g=2}^{N_{G}+1}w_{g}\frac{p_{1}}{p_{g}}\mathcal{G}_{g}Y_{\mathcal{T}}%
\right] } \\
& =\mathbb{E}\left[ Y_{\mathcal{T}}|G=1\right] -\sum_{g=2}^{N_{G}+1}w_{g}%
\mathbb{E}\left[ Y_{\mathcal{T}}|G=g\right] ,
\end{align*}%
which is the familiar synthetic control formula.

In general, $\phi $ integrates elements from both approaches and remains
unbiased under either of the identification assumptions, as established in
the following theorem.

\begin{theorem}
\label{thm:identification} Let Assumptions \ref{asm:NA}, \ref{asm:panel},
and \ref{asm:overlap} hold. \ 

\begin{enumerate}
\item If Assumption \ref{asm:PT} holds, then 
\begin{equation*}
\theta =\mathbb{E}[\phi (S;m_{\Delta },\tilde{p},\tilde{w};\pi _{1})]=%
\mathbb{E}[\phi (S;\tilde{m}_{\Delta },p,\tilde{w};\pi _{1})]=\mathbb{E}%
\left[ \phi \left( S;m_{\Delta },p,\tilde{w};\pi _{1}\right) \right] ,
\end{equation*}%
for any vector of weights $\tilde{w}$ that sum to one, any positive
functions $\tilde{p}$, and any functions $\tilde{m}_{\Delta }$.

\item If Assumption \ref{asm:SC} holds, then 
\begin{equation*}
\theta =\mathbb{E}[\phi (S;\tilde{m}_{\Delta },p,w;\pi _{1})],
\end{equation*}%
for any function $\tilde{m}_{\Delta }$, where $w$ is the vector of weights
satisfying (\ref{eqn:w-id}).
\end{enumerate}
\end{theorem}

Theorem \ref{thm:identification} establishes that the moment function $\phi $
is unbiased for the ATT across various scenarios, providing the key
sufficient condition for the consistency of the ATT estimator proposed in
the next section.

The unbiasedness of $\phi $ reflects two distinct notions of double
robustness. The first, which is the central contribution of this paper,
concerns identification. The moment function $\phi $ remains unbiased for
the ATT as long as either the parallel trends condition or the synthetic
control condition holds. This feature allows researchers to use a single
estimator based on $\phi $ without having to commit ex ante to one
identifying assumption over the other.

The second notion concerns robustness to parametric misspecification of
nuisance components. Although this is not the primary focus of the paper,
discussing it helps clarify the structure of the identification formula.
When the parallel trends condition holds, the moment function $\phi$ remains
unbiased as long as either the outcome model $m_{\Delta }$ or the propensity
score model $p$ is correctly specified. The double robustness between $p$
and $m_{\Delta }$ arises because $\phi $ closely resembles the doubly robust
DiD moment function in the classical setting \citep{sant2020doubly}.
Moreover, the moment function is unbiased even if the weights $w$ are
entirely misspecified. The weights are irrelevant for identification since
any control group can be used for comparison on its own, as can any
combination of control groups. On the other hand, under the synthetic
control condition, both the propensity scores $p$ and the weights $w$ must
be correctly specified, while the outcome model $m_{\Delta }$ can be
misspecified.\footnote{%
See \cite{lewbel2023over} for a recent study of double robustness in
parametric, overidentified GMM models. In contrast, our paper focuses on
double robustness with respect to identification assumptions and employs
nonparametric estimation methods.}

Theorem \ref{thm:identification} is also the key ingredient in Neyman
orthogonality, namely the insensitivity of the moment condition to the
nuisance functions. Theorem \ref{thm:asymp-dist-1} uses this property to
establish the asymptotics of the ATT estimator under Assumption \ref{asm:PT}%
. Orthogonality holds because two conditional moment equalities make the
first-order derivatives of the moment function with respect to the nuisance
functions vanish at the truth. The first, implied by Assumption \ref{asm:PT}, is $%
\mathbb{E}[ \mathcal{G}_g[\Delta Y-m_\Delta(X)] |X]=0$ for $g\geq 2$. The
second, $\mathbb{E}\big[
\mathcal{G}_1 - \sum_{g=2}^{N_G+1} w_g(X)\frac{p_1(X)}{p_g(X)}\mathcal{G}_g
|X \big]=0, $ requires only that the weights sum to one and therefore holds
under either identification assumption.

\paragraph{Alternative moment function}

For completeness, we note that an alternative formulation of the moment
function is possible: 
\begin{align*}
& \phi^{\#}(S;\{m_{g,\Delta }:2\leq g\leq N_{G}+1\},p,w;\pi _{1}) \\
=& \frac{\mathcal{G}_{1}}{\pi _{1}}\left( \Delta
Y-\sum_{g=2}^{N_{G}+1}w_{g}(X)m_{g,\Delta }(X)\right) -\frac{1}{\pi _{1}}%
\sum_{g=2}^{N_{G}+1}w_{g}(X)\frac{p_{1}(X)}{p_{g}(X)}\mathcal{G}_{g}\left(
\Delta Y-m_{g,\Delta }(X)\right) \\
=& \frac{1}{\pi _{1}}\sum_{g=2}^{N_{G}+1}w_{g}(X)\left( \mathcal{G}_{1}-%
\frac{p_{1}(X)}{p_{g}(X)}\mathcal{G}_{g}\right) (\Delta Y-m_{g,\Delta }(X)),
\end{align*}%
where $m_{g,\Delta }(x)\equiv \mathbb{E}[\Delta Y|G=g,X=x]$. This
formulation $\phi^{\#}$ offers double robustness between the outcome models $%
\bigl(m_{g,\Delta },\,g=2,\ldots,N_{G}+1\bigr)$ and the propensity score
model $p$, and it does so under both the parallel trends and synthetic
control assumptions. However, it requires correct specification of the
weights in the synthetic control setting and therefore does not exhibit
Neyman orthogonality with respect to all nuisance parameters. Moreover, this
approach is relatively non-parsimonious, as it adds another set of nuisance
parameters $\bigl(m_{g,\Delta },\,g=2,\ldots,N_{G}+1\bigr)$, which scales
with the number of groups, on top of the two existing sets of functions $p$
and $w$. Since our paper focuses on semiparametric estimation of the ATT
through nonparametric first-step estimation guided by identification
assumptions (rather than assuming parametric functional forms), we favor a
more parsimonious approach based on $\phi $.

\subsection{Comparison with existing methods in the case without covariates}

\label{sec:no-covariate}

Our framework incorporates covariates nonparametrically by imposing
identification conditional on covariates, in contrast to existing work that
relies on linear projection residualization or other parametric structures.
In the absence of covariates, our approach specializes to a formulation that
is directly comparable to leading alternatives, and we provide the
comparison in this subsection.

When there is no covariate information, we can simplify the nuisance
parameters into constants: $m_{\Delta }(X)\equiv m_{\Delta}$, $p_{1}\left(
X\right) /p_{g}(X)=p_{1}/p_{g}$, and $w_{g}(X)\equiv w_{g}$. For simplicity,
we assume that the number of individuals in group $g$ is $np_{g}.$ Noting
that, in the absence of covariate information, $\pi _{1}=p_{1},$ we have%
\begin{eqnarray*}
\phi (S_{i};m_{\Delta },p,w;\pi _{1}) &\equiv &\frac{\mathcal{G}_{1i}}{\pi
_{1}}(\Delta Y_{i}-m_{\Delta })-\frac{1}{\pi _{1}}\sum_{g=2}^{N_{G}+1}w_{g}%
\frac{p_{1}}{p_{g}}\mathcal{G}_{gi}(\Delta Y_{i}-m_{\Delta }) \\
&=&\frac{\mathcal{G}_{1i}}{p_{1}}(\Delta Y_{i}-m_{\Delta
})-\sum_{g=2}^{N_{G}+1}w_{g}\frac{\mathcal{G}_{gi}}{p_{g}}(\Delta
Y_{i}-m_{\Delta }).
\end{eqnarray*}%
The simple average estimator based on the above $\phi $ is then given by 
\begin{eqnarray}
&&\frac{1}{n}\sum_{i=1}^{n}\phi (S_{i};m_{\Delta },p,w;\pi _{1})  \notag \\
&=&\frac{1}{np_{1}}\sum_{i=1}^{n}\mathcal{G}_{1i}(\Delta Y_{i}-m_{\Delta
})-\sum_{g=2}^{N_{G}+1}w_{g}\frac{1}{np_{g}}\sum_{i=1}^{n}\mathcal{G}%
_{gi}(\Delta Y_{i}-m_{\Delta })  \notag \\
&=&(\Delta \bar{Y}_{1}-m_{\Delta })-\sum_{g=2}^{N_{G}+1}w_{g}(\Delta \bar{Y}%
_{g}-m_{\Delta })=\Delta \bar{Y}_{1}-\sum_{g=2}^{N_{G}+1}w_{g}\Delta \bar{Y}%
_{g}  \notag \\
&=&\bar{Y}_{1,\mathcal{T}}-\bar{Y}_{1,\mathcal{T}-1}-%
\sum_{g=2}^{N_{G}+1}w_{g}(\bar{Y}_{g,\mathcal{T}}-\bar{Y}_{g,\mathcal{T}-1}),
\label{eqn:degenerate-formula}
\end{eqnarray}%
where $\bar{Y}_{g,t}$ is the average of the outcomes for individuals in
group $g$ at time period $t.$

Building on \cite{shen2023same}'s formulation (specifically, their equation
(6)) of the synthetic DiD estimator studied by \cite%
{arkhangelsky2021synthetic}, the synthetic DiD (SDiD) estimator of the ATT
based on the group-level data $\{ \bar{Y}_{g,t}:g=1,\ldots,N_{G}+1,
t=1,\ldots,\mathcal{T}\} $ is%
\begin{equation}
\bar{Y}_{1,\mathcal{T}}-\sum_{g=2}^{N_{G}+1}w_{g}\bar{Y}_{g,\mathcal{T}%
}-\sum_{t=1}^{\mathcal{T}-1}\alpha _{t}\bar{Y}_{1,t}+\sum_{g=2}^{N_{G}+1}%
\sum_{t=1}^{\mathcal{T}-1}w_{g}\alpha _{t}\bar{Y}_{g,t},  \label{eqn:syn_did}
\end{equation}%
where $\alpha _{t}$ represents the coefficient from the horizontal
regression of $\bar{Y}_{\cdot, \mathcal{T}}$ on its lagged values. With a
suitably chosen time-series prediction for the outcomes at $t=\mathcal{T}$,
the augmented synthetic control (ASC) estimator of \cite{ben2021augmented}
also takes the above form. Note that if we set the temporal weights $\alpha
_{t}$ to be $\mathbf{1}\{t=\mathcal{T}-1\}$, then the estimator in (\ref%
{eqn:syn_did}) becomes our estimator in (\ref{eqn:degenerate-formula}).
Theorem 1 of \cite{shen2023same} shows that horizontal and vertical OLS
regressions with the minimum $\ell _{2}$-norm solution, when necessary,
produce a numerically identical estimate of the ATT, which is also
numerically identical to the corresponding SDiD and ASC estimates.\footnote{%
The numerical equivalence of the four methods also holds if the weights $%
\left\{ \alpha _{t}\right\} $ and $\left\{ w_{g}\right\} $ are estimated
from a ridge regression with an $\ell _{2}$ penalty. Their difference lies
in how the weights are estimated.} In contrast, DiD-SC integrates a
first-difference comparison with vertical regression, distinguishing it from
existing approaches, including horizontal and vertical regressions, SDiD,
and ASC. While we adopt the \textquotedblleft similar units behave
similarly\textquotedblright\ principle from vertical regression, we depart
from horizontal regression's reliance on the entire history to guide the
future. Instead, we employ the notion that \textquotedblleft comparable
trends evolve similarly,\textquotedblright\ focusing on trends derived from
only two time periods: the last pre-treatment period and the post-treatment period. In
addition, it is useful to reiterate that we accommodate covariate
information at the individual level, while existing SC and SDiD methods
typically operate on aggregate-level data and can only incorporate
aggregate-level covariates.

\section{Estimation and inference}

\label{sec:asymptotics}

The identification result in Theorem \ref{thm:identification} establishes
that estimators based on $\phi$ can be consistent for $\theta$ as long as
either Assumption \ref{asm:PT} or Assumption \ref{asm:SC} is satisfied.
However, the distribution of the estimator varies depending on the
underlying identification assumption. In this section, we impose suitable
technical conditions under each identification assumption, derive the
corresponding asymptotic distributions, and present a unified multiplier
bootstrap method for inference. When both assumptions hold, the two
asymptotic distributions coincide and reduce to the one under Assumption \ref%
{asm:PT}, but with the true weight functions.

We emphasize that Theorems \ref{thm:asymp-dist-1}, \ref{thm:asymp-dist-2},
and \ref{thm:bootstrap} below do not describe the asymptotics of different
estimators. Rather, they analyze the same estimator $\widehat{\theta}$ and
its bootstrap counterpart, as defined in \eqref{eqn:theta-hat-def} and %
\eqref{eqn:theta-hat-star-def}, respectively, under different identification
assumptions. The technical conditions on the nuisance estimators differ
across the two cases because nuisance estimation plays different first-order
roles. Under the synthetic control condition, the requirements are more
explicit and potentially stronger than the rate conditions under the
parallel trends assumption, because the moment function does not satisfy
Neyman orthogonality under synthetic control.

\subsection{Semiparametric estimation}

Suppose that we have decided on a nonparametric method for estimating each
of $m_{\Delta }$, $p$, and $w$ (we will specify each method in detail
later). We implement the following cross-fitting procedure: assign each
cross-sectional unit independently and uniformly at random to one of $L$
folds. The fold sizes are then random rather than exactly $n/L$, but each
fold contains a fraction $1/L$ of the sample in the limit, which is all that
the asymptotic theory requires. For $\ell =1,\ldots ,L$, let $I_{\ell }$
denote the index set of the cross-sectional units in the $\ell $th fold and $%
I_{\ell }^{c}=\bigcup_{\ell ^{\prime }\neq \ell }I_{\ell ^{\prime }}$ the
index set of the cross-sectional units not in the $\ell $th fold. For an
observation $X_{i}$ with index $i\in I_{\ell }$, we use the subsample with
indices in $I_{\ell }^{c}$ to construct the nonparametric estimates $\hat{m}%
_{\Delta }^{\ell }(X_{i})$, $\hat{p}^{\ell }(X_{i})$, and $\hat{w}^{\ell
}(X_{i})$, where the superscript $\ell $ signifies the fact that each of the
three nonparametric estimators is constructed using data in $I_{\ell }^{c}$.
The semiparametric estimator of $\theta $ is constructed as%
\begin{equation}
\hat{\theta}=\frac{1}{n}\sum_{\ell =1}^{L}\sum_{i\in I_{\ell }}\phi (S_{i};%
\hat{m}_{\Delta }^{\ell },\hat{p}^{\ell },\hat{w}^{\ell };\hat{\pi}_{1}),
\label{eqn:theta-hat-def}
\end{equation}
where $\hat{\pi}_{1}=\sum_{i=1}^{n}\mathcal{G}_{1i}/n$ is the sample average
estimator for $\pi _{1}$.

\paragraph{Asymptotic theory under parallel trends}

In the absence of Assumption \ref{asm:SC}, when only Assumption \ref{asm:PT}
holds, the synthetic weight $w$ may not be well-defined: there may not exist
any $w$ such that (\ref{eqn:w-id}) is satisfied. However, since researchers
do not know in practice whether Assumption \ref{asm:PT} or Assumption \ref%
{asm:SC} is the valid assumption, the proposed procedure uses the same
estimator $\hat w$ throughout. The specific construction of $\hat w$ is the
synthetic control projection introduced later in (\ref{eqn:hat_w-def}). For
deriving the asymptotic distribution of $\hat{\theta}$ under Assumption \ref%
{asm:PT}, it is necessary for the random quantity $\hat{w}$ to converge to a
probability limit $\omega \equiv \plim \hat{w}$.\footnote{%
Theorem \ref{thm:asymp-dist-2}(iii) later provides a primitive rank
condition on the pre-treatment conditional-mean system under which this
high-level convergence requirement holds.} This limit $\omega$ can be
interpreted as a pseudo-true vector of weights that minimizes the
discrepancy between the left-hand and right-hand sides of (\ref{eqn:w-id})
in the pre-treatment periods.

The following theorem establishes the asymptotic distribution of the
estimator under the parallel trends condition. Throughout this paper,
asymptotics are considered with a fixed number of control groups $N_{G}$, time
periods $\mathcal{T}$, and cross-fitting folds $L$, while the
cross-sectional sample size $n$ grows to infinity. This asymptotic framework
is motivated by micro panel settings in which, for example, groups are
states and units are households, with many households per state and a
limited number of states.

To\ simplify the presentation, we let $w_{1}\left( \cdot \right) $ and $\hat{%
w}_{1}\left( \cdot \right) $ be the constant function $\mathbf{1}(\cdot )$
with $\mathbf{1}(x)=1$ for all $x\in \mathcal{X}$, a convention that will be
used throughout the rest of the paper.

\begin{theorem}
\label{thm:asymp-dist-1} Let Assumptions \ref{asm:NA}, \ref{asm:panel}, \ref%
{asm:overlap}, \ref{asm:PT}, and the following conditions hold:

\begin{enumerate}
\item For each $g\geq 1$, $\mathbb{E}[(\Delta Y-m_{\Delta }(X))^{2}|X,G=g]$, 
$1/p_{g}$, and $\omega _{g}$ are bounded functions of $X$.\footnote{%
The boundedness of $1/p_g(X)$, equivalently, requiring $p_g(X)$ to be
uniformly bounded away from zero, is a uniform overlap condition stronger
than Assumption \ref{asm:overlap}. This stronger condition is imposed here
for the asymptotic theory, where the inverse propensity weights and their
estimation errors must be controlled uniformly.}

\item The first-stage estimators satisfy that (1) the estimated weights sum
to one and are bounded in probability, i.e., $\sum_{g\geq 2}\hat{w}%
_{g}\left( x\right) =1$ for all $x\in \mathcal{X}$ and $\lVert \hat{w}%
_{g}\rVert _{\infty }=O_{p}(1)$ for all $g$, (2) the estimators are $L_{2}$%
-consistent, i.e., 
\begin{align*}
\left\Vert \hat{m}_{\Delta }-m_{\Delta }\right\Vert _{L_{2}(F_{X})}&
=o_{p}(1), \\
\left\Vert \hat{p}_{1}/\hat{p}_{g}-p_{1}/p_{g}\right\Vert _{L_{2}(F_{X})}&
=o_{p}(1),\text{ }g\geq 2, \\
\left\Vert \hat{w}_{g}-\omega _{g}\right\Vert _{L_{2}(F_{X})}& =o_{p}(1),%
\text{ }g\geq 2,
\end{align*}%
and (3) their rates satisfy that 
\begin{equation*}
\left\Vert \hat{p}_{1}/\hat{p}_{g}-p_{1}/p_{g}\right\Vert
_{L_{2}(F_{X})}\left\Vert \hat{m}_{\Delta }-m_{\Delta }\right\Vert
_{L_{2}(F_{X})}=o_{p}(n^{-1/2}),\text{ }g\geq 2.
\end{equation*}
\end{enumerate}

Then $\sqrt{n}(\hat{\theta}-\theta )\overset{d}{\rightarrow }N(0,V_{\mathrm{%
PT}})$, where the asymptotic variance is 
\begin{align}
V_{\mathrm{PT}}& \equiv \frac{1}{\pi _{1}^{2}}\Bigg(\mathbb{E}\Bigg[\mathcal{%
G}_{1}(\Delta Y-m_{1,\Delta }(X))^{2}+\mathcal{G}_{1}(m_{1,\Delta
}(X)-m_{\Delta }(X)-\theta )^{2}  \notag  \label{eqn:V1} \\
& \quad +\sum_{g=2}^{N_{G}+1}\omega _{g}(X)^{2}\mathcal{G}_{g}\frac{%
p_{1}(X)^{2}}{p_{g}(X)^{2}}(\Delta Y-m_{\Delta }(X))^{2}\Bigg]\Bigg),
\end{align}%
with $m_{1,\Delta }\equiv \mathbb{E}[\Delta Y|X,G=1]$.
\end{theorem}

The asymptotic variance $V_{\mathrm{PT}}$ in (\ref{eqn:V1}) can be expressed
as the variance of the weighted average of the efficient influence functions
for the ATT in a canonical $2\times 2$ DiD design, as characterized in
Proposition 1 of \cite{sant2020doubly}. More specifically, in the $2\times 2$
design with a single control group $g$, the efficient influence function for
the ATT is 
\begin{align*}
\mathbb{IF}_{g}& \equiv \frac{1}{\pi _{1}}\Bigg(\mathcal{G}_{1}(\Delta
Y-m_{1,\Delta }(X))+\mathcal{G}_{1}(m_{1,\Delta }(X)-m_{\Delta }(X)-\theta )
-\mathcal{G}_{g}\frac{p_{1}(X)}{p_{g}(X)}(\Delta Y-m_{\Delta }(X))\Bigg).
\end{align*}%
Equation (\ref{eqn:V1-calculation}) in the proof shows that $V_{\mathrm{PT}}$
is the variance of the weighted average of all such $\mathbb{IF}_{g}$, that
is, $V_{\mathrm{PT}}=\mathbb{E}[(\sum_{g\geq 2}\omega _{g}(X)\mathbb{IF}%
_{g})^{2}]$. For illustration, in a stylized scenario with a single control
group ($g=2$) and a trivial weight $\omega _{g}\equiv 1$, $V_{\mathrm{PT}}$
reduces to the efficiency bound derived in \cite{sant2020doubly}.
Alternatively, if one disregards the finer group structure and instead pools
all control groups into a single aggregated control group, the asymptotic
variance of the doubly robust DiD estimator in \cite{sant2020doubly} becomes 
\begin{align}
V_{\mathrm{pool}}& =\frac{1}{\pi _{1}^{2}}\Bigg(\mathbb{E}\Bigg[\mathcal{G}%
_{1}(\Delta Y-m_{1,\Delta }(X))^{2}+\mathcal{G}_{1}(m_{1,\Delta
}(X)-m_{\Delta }(X)-\theta )^{2}  \notag  \label{eqn:drdid-var} \\
& \quad +\sum_{g=2}^{N_{G}+1}\left( \frac{p_{g}(X)}{p_{-1}(X)}\right) ^{2}%
\mathcal{G}_{g}\frac{p_{1}(X)^{2}}{p_{g}(X)^{2}}(\Delta Y-m_{\Delta }(X))^{2}%
\Bigg]\Bigg),
\end{align}%
where $p_{-1}(X)\equiv \mathbb{P}(G\neq 1|X)$ and the weight $%
p_{g}(X)/p_{-1}(X)$ reflects the group shares, i.e., the fraction of units
in control group $g$ relative to all control units. The variance $V_{\mathrm{%
pool}}$ in (\ref{eqn:drdid-var}) uses the weight $p_{g}(X)/p_{-1}(X) $,
which is distinct from the synthetic control weight $\omega _{g}\left(
X\right) $ in this paper.

In general, there is no dominance relation between $V_{\mathrm{PT}}$ and $V_{%
\mathrm{pool}}$. We briefly provide some guidance on how to assess which
variance is likely to be smaller, leaving the detailed calculation to
Supplementary Appendix \ref{sec:variance-comparison}. In the homoskedastic case, where the
conditional variance of $\Delta Y$ is the same across control groups, $V_{%
\mathrm{pool}}$ is weakly smaller because pooling according to group shares
is the most efficient way to aggregate equally noisy control groups. In the
heteroskedastic case, the variance-minimizing weights depend on both group
sizes and group-specific variances, placing more weight on larger and more
stable control groups. Therefore, $V_{\mathrm{PT}}$ can be smaller than $V_{%
\mathrm{pool}}$ when the probability limit of $\hat{w}$ is closer to these
variance-minimizing weights than the pooling weights are. In Supplementary Appendix \ref%
{sec:variance-comparison}, we provide a simple numerical illustration
showing that this case can arise naturally, especially when
heteroskedasticity across control groups is substantial.

Theorem \ref{thm:asymp-dist-1} aligns with the framework of double/debiased
machine learning (DML) as discussed in \citet{chernozhukov2018DML}. Under
Assumption \ref{asm:PT}, the moment function $\phi $ satisfies Neyman
orthogonality with respect to the nuisance parameters. This implies that the
asymptotic distribution of the estimator is the same as if the true values
of the nuisance parameters were used, provided that the product of the
estimation errors for $m_{\Delta }$ and $p$ converges at a rate faster than $%
1/\sqrt{n}$. This rate requirement allows for the use of a wide range of
estimators, including machine learning methods or traditional nonparametric
methods, such as kernel-based or sieve-based methods. Notably, there is no
rate requirement for the weight estimator $\hat{w}$, as long as it converges
in probability to some limit, sums to one, and is uniformly bounded in
probability. This is consistent with the result of Theorem \ref%
{thm:identification}(i), which shows that the nuisance weights play a
secondary role under Assumption \ref{asm:PT}.

\paragraph{Asymptotic theory under synthetic control}

Deriving the asymptotic distribution of $\hat{\theta}$ under Assumption \ref%
{asm:SC} is more challenging because $\phi $ does not satisfy Neyman
orthogonality with respect to $p$ and $w$. Consequently, when computing the
asymptotic variance, additional adjustment terms are needed to account for
the first-stage estimation error in estimating these nuisance parameters. To
derive the adjustment terms, we need to find the influence function
associated with their estimation. To keep our paper focused, we examine
kernel-based estimators for the nuisance parameters, although other
nonparametric methods, such as the method of sieves, could also be used.

For any function $f$, define the empirical average operator $P_{n}\left[
\cdot \right] $ as $P_{n}[f(S)]\equiv \frac{1}{n}\sum_{i=1}^{n}f(S_{i})$.
Note that the propensity scores enter the moment function $\phi $ only
through the ratios $r_{1,g}\equiv p_1/p_g$, $g=2,\ldots,N_G+1$. Therefore,
in the estimation and asymptotic analysis below, we work directly with these
ratios rather than with the raw propensity scores. With a slight abuse of
notation, we let $\phi $ take either $p$ or the corresponding ratios as its
propensity-score argument, whichever is more convenient.\footnote{%
An alternative approach involves estimating the propensity scores using
local polynomial regression and combining them into ratio estimates. Under
suitable regularity conditions, this estimator has the same Bahadur
representation as in Theorem \ref{thm:asymp-dist-2}(ii).} Similarly, $\hat{p}%
_{1}/\hat{p}_{g}$ in Theorem \ref{thm:asymp-dist-1} and elsewhere should be
read as $\hat{r}_{1,g}$ when the ratio is estimated directly, as in (\ref%
{eqn:hat_r-def}) below.

We first introduce compact notation for polynomial bases up to a generic
order $\bar{s}$. Denote $d$ as the dimension of $X$. For $%
x=(x_1,\ldots,x_d)\in\mathbb{R}^d$ and a $d$-dimensional vector $%
s=(s_1,\ldots,s_d)$ of nonnegative integers, denote $x^s\equiv
x_1^{s_1}\cdots x_d^{s_d}$ and $|s|\equiv s_1+\cdots+s_d$. Define the
polynomial basis vector $b_{\bar{s}}(x)\equiv (x^s:|s|\leq \bar{s})$,
ordered with the constant term first and then by increasing total degree.
The local polynomial coefficients are collected in a vector $\beta$ of the
same length as $b_{\bar{s}}(x)$.

Using the fact that the ratio $r_{1,g}$ minimizes the objective function $%
\mathbb{E}[\rho (r(X),G)]$ for 
\begin{equation*}
\rho (r(X),G)\equiv r(X)^{2}\mathcal{G}_{g}-2r(X)\mathcal{G}_{1},
\end{equation*}%
we construct the following local polynomial regression: 
\begin{equation}
\hat{r}_{1,g}(x)=\iota _{1}^{\prime }\argmin_{\beta}P_{n}\left[ K\left(
(X-x)/h\right) \rho \left( b_{\bar{s}}(X-x)^{\prime }\beta,G\right) \right] ,
\label{eqn:hat_r-def}
\end{equation}%
where $h$ is the bandwidth, $K$ is the multivariate kernel function, and $%
\iota_{1}\equiv (1,0,\ldots ,0)^{\prime }$ is a vector of the same length as 
$b_{\bar{s}}(x)$ with $1$ in the first position and $0$ elsewhere. To
simplify the presentation, we define $r_{1,g}(\cdot )$ and $\hat{r}%
_{1,g}(\cdot )$ to be the constant function $\mathbf{1}(\cdot )$ when $g=1$.

Second, the weights $w$ are determined by solving the system of
identification equations involving the outcome functions $m_{g,t}(x)\equiv 
\mathbb{E}[Y_{t}|G=g,X=x]$. Let $\bm{w}_{0}(x)\equiv (w_{2}(x),\dots
,w_{N_{G}}(x))^{\prime }$ denote the vector of weights excluding the entry
for the last group. Take the last donor group as the reference group. Since
the weights sum to one, the weight for the last group can be expressed as $%
w_{N_{G}+1}(x)=1-\mathbf{1}_{N_{G}-1}^{\prime }\bm{w}_{0}(x)$, where $%
\mathbf{1}_{N_{G}-1}$ is a vector of ones with length $N_{G}-1$. By
Assumption \ref{asm:SC}, the weights can be identified via the equation: 
\begin{equation*}
\begin{pmatrix}
m_{2,1}(x) & \cdots & m_{N_{G}+1,1}(x) \\ 
\vdots &  & \vdots \\ 
m_{2,\mathcal{T}-1}(x) & \cdots & m_{N_{G}+1,\mathcal{T}-1}(x)%
\end{pmatrix}%
\begin{pmatrix}
\bm{w}_{0}(x) \\ 
1-\mathbf{1}_{N_{G}-1}^{\prime }\bm{w}_{0}(x)%
\end{pmatrix}%
=%
\begin{pmatrix}
m_{1,1}(x) \\ 
\vdots \\ 
m_{1,\mathcal{T}-1}(x)%
\end{pmatrix}%
.
\end{equation*}%
After rearranging terms, we obtain that $M(x)\bm{w}_{0}(x)=m_{1}(x)$, where 
\begin{equation}
M(x)\equiv 
\begin{pmatrix}
m_{-1,1}{(x)}^{\prime } \\ 
m_{-1,2}{(x)}^{\prime } \\ 
\vdots \\ 
m_{-1,\mathcal{T}-1}{(x)}^{\prime }%
\end{pmatrix}%
,m_{1}(x)\equiv 
\begin{pmatrix}
m_{1,1}(x)-m_{N_{G}+1,1}(x) \\ 
\vdots \\ 
m_{1,\mathcal{T}-1}(x)-m_{N_{G}+1,\mathcal{T}-1}(x)%
\end{pmatrix}%
,  \label{eqn:M-m1-def}
\end{equation}%
with $m_{-1,t}(x)\equiv (m_{2,t}(x)-m_{N_{G}+1,t}(x),\ldots
,m_{N_{G},t}(x)-m_{N_{G}+1,t}(x))^{\prime }$. We refer to $M(x)$ as the
pre-treatment donor-difference matrix because its columns collect, across
the pre-treatment periods, the conditional-mean trajectories of donor groups 
$2,\ldots,N_G$ relative to the reference donor group $N_G+1$. Throughout, we
index the columns of $M(x)$ by donor group rather than by position: $M(x)$
has $N_{G}-1$ columns, one for each of the donor groups $g=2,\ldots ,N_{G}$,
and the column belonging to group $g$ collects $%
[M(x)]_{t,g}=m_{g,t}(x)-m_{N_{G}+1,t}(x)$, $t=1,\ldots ,\mathcal{T}-1$. The
same convention applies to every matrix below with one column per
non-reference donor. Provided that $\mathcal{T}\geq N_{G}$ and $M{(x)}%
^{\prime }M(x)$ is invertible, we can solve for $\bm{w}_{0}(x)$ as%
\begin{equation}
\bm{w}_{0}(x)=\argmin_{u}(M(x)u-m_{1}(x))^{\prime }(M(x)u-m_{1}(x))=(M{(x)}%
^{\prime }M(x))^{-1}M{(x)}^{\prime }m_{1}(x).  \label{eqn:hat_w-def}
\end{equation}
This is similar to solving an ordinary least squares problem. Let $\hat{M}(x)
$ and $\hat{m}_{1}(x)$ denote the respective estimators of $M(x)$ and $%
m_{1}(x)$, obtained by replacing each $m_{g,t}(x)$ with the corresponding $%
\hat{m}_{g,t}(x)$. The weight estimator is constructed as $\hat{\bm{w}}%
_{0}(x)\equiv (\hat{M}{(x)}^{\prime }\hat{M}{(x)})^{-1}\hat{M}{(x)}^{\prime }%
\hat{m}_{1}(x)$. Thus, the estimation of $w$ and $m_{\Delta }$ reduces to
estimating $m_{g,t}$, which can be accomplished using another local
polynomial regression:%
\begin{equation}
\hat{m}_{g,t}(x)=\iota _{1}^{\prime }\argmin_{\beta}P_{n}\left[ K\left( 
\frac{X-x}{h}\right) \left( Y_{t}-b_{\bar{s}}(X-x)^{\prime }\beta\right) ^{2}%
\mathcal{G}_{g}\right].  \label{eqn:hat_m_gt-def}
\end{equation}

\begin{theorem}
\label{thm:asymp-dist-2} Let Assumptions \ref{asm:NA}, \ref{asm:panel}, \ref%
{asm:overlap}, \ref{asm:SC}, conditions (i) and (ii) of Theorem \ref%
{thm:asymp-dist-1}, and the following conditions hold:

\begin{enumerate}
\item The convergence rates of the nuisance estimators satisfy 
\begin{equation*}
\left\Vert \hat{m}_{g,t}-m_{g,t}\right\Vert _{L_{2}(F_{X})},\left\Vert \hat{r%
}_{1,g}-r_{1,g}\right\Vert _{L_{2}(F_{X})},\left\Vert \hat{w}%
_{g}-w_{g}\right\Vert _{L_{2}(F_{X})}=o_{p}(n^{-1/4}),\text{ for all }g.
\end{equation*}

\item Let $h=o\big(n^{-1/\{2(\bar{s}+1)\}}\big)$ be an undersmoothing
bandwidth, and assume that $\mathcal{X}$ is compact. Uniformly over $x\in%
\mathcal{X}$, the nuisance estimators admit the following Bahadur
representations: for all $g$, 
\begin{align*}
\hat{r}_{1,g}(x)-r_{1,g}(x) &= P_{n}\left[ \frac{1}{h^d}K^{\text{eq}}_{\bar{s%
}}\left(\frac{X-x}{h}\right) \frac{\mathcal{G}_{1}-r_{1,g}(X)\mathcal{G}_{g}%
}{p_{g}(X)f_{X}(x)} \right] +o_{p}(n^{-1/2}), \\
\hat{m}_{g,t}(x)-m_{g,t}(x) &= P_{n}\left[ \frac{1}{h^d}K^{\text{eq}}_{\bar{s%
}}\left(\frac{X-x}{h}\right) \frac{\mathcal{G}_{g}(Y_{t}-m_{g,t}(X))}{%
p_{g}(X)f_{X}(x)} \right] +o_{p}(n^{-1/2}), 1\leq t\leq \mathcal{T}-1,
\end{align*}
where $f_X$ is the marginal density function of $X$, and $K_{\bar{s}}^{\text{%
eq}}$ is the equivalent kernel induced by the local polynomial estimator.
The kernel $K_{\bar{s}}^{\text{eq}}$ is bounded and compactly supported, and
satisfies $\int K_{\bar{s}}^{\text{eq}}(u)\,du=1$ and $\int u^s K_{\bar{s}}^{%
\text{eq}}(u)\,du=0$ for all $1\leq |s|\leq \bar{s}$.

\item The functions $w_{g}$, $m_{g,t}$, and $p_{g}$, $1\leq g\leq
N_{G}+1,1\leq t\leq \mathcal{T},$ are $(\bar{s}+1)$-times continuously
differentiable with $\bar{s}\geq d$. There exists a constant $c>0$ such that
the smallest eigenvalue of $M(x)^{\prime }M(x)$ is larger than $c$ for all $%
x\in \mathcal{X}$.
\end{enumerate}

Let $\psi _{\mathrm{SC}}\left( S\right) $ be the influence function defined
in (\ref{eqn:psi_SC}) in the proof, and assume that $V_{\mathrm{SC}}=\mathbb{%
E}\left[ \psi _{\mathrm{SC}}\left( S\right) ^{2}\right] $\ is finite. Then $%
\sqrt{n}(\hat{\theta}-\theta )\overset{d}{\rightarrow }N(0,V_{\mathrm{SC}})$.
\end{theorem}

We discuss the three conditions of Theorem \ref{thm:asymp-dist-2}. Condition
(i) of Theorem \ref{thm:asymp-dist-2} imposes a high-level rate condition on
the nuisance estimators. Since $\phi $ is not Neyman orthogonal under
Assumption \ref{asm:SC}, the required rate of convergence for each nuisance
function estimator has been strengthened to $o_{p}\left( n^{-1/4}\right)$
compared to Theorem \ref{thm:asymp-dist-1}. This rate is typical for
first-step nonparametric estimators %
\citep{newey1994asymptotic,chen2003estimation}. When the nuisance functions
are estimated by local-polynomial methods and $X\in\mathbb{R}^{d}$, the
corresponding uniform convergence rate is typically of order $O_p\big(
h^{\bar{s}+1} + \left(\log n/(n h^{d})\right)^{1/2} \big), $ provided that
the nuisance functions are sufficiently smooth, for example $(\bar{s}+1) $%
-times smooth \citep{masry1996multivariate}. Under undersmoothing, this rate
can be made faster than $n^{-1/4}$ whenever $\bar{s} \geq d$. This is a
standard requirement in classical semiparametric estimation with
undersmoothed nonparametric first steps \citep[e.g.,][]{newey1994kernel}. In
our empirical application, however, there is only one continuous covariate,
so the dimension of $X$ does not create a practical difficulty.

For practical implementation, the relevant quantity is the local
within-group sample size used to estimate the nuisance functions. For a $d$%
-dimensional continuous covariate and bandwidth $h$, the expected number of
observations from group $g$ falling in a neighborhood of $x$ is proportional
to $n p_g(x) f_X(x) h^d$. Thus, the method is most suitable when the number of
continuous covariates is small and each group contains a sufficiently large
local sample around the covariate values relevant for estimation. As a
practical diagnostic, researchers can directly examine the number of
observations receiving positive kernel weight. For comparison, %
\citet{calonico2019nprobust} impose a minimum local sample size of $21$
observations in their implementation of local-polynomial estimation.

Condition (ii) of Theorem \ref{thm:asymp-dist-2} specifies the Bahadur
(asymptotic linear) representations of the local polynomial estimators. Such
representations are well established, with primitive conditions provided,
for example, by \citet{kong2010uniform}. A useful way to express the leading
term is through the equivalent kernel representation of local polynomial
estimators \citep[Section 3.2.2]{fan1996local}. Specifically, the equivalent
kernel is $K^{\text{eq}}_{\bar{s}}(u)=\iota_1^{\prime }\left( \int b_{\bar{s}%
}(u)b_{\bar{s}}(u)^{\prime }K(u)\,du\right)^{-1}b_{\bar{s}}(u)K(u)$. It
satisfies the required higher-order moment conditions and yields a leading
bias of order $h^{\bar{s}+1}$.\footnote{%
For illustration, consider the scalar case with a symmetric primitive kernel 
$K$. For the local-linear estimator, $K^{\text{eq}}_1$ coincides with the
primitive kernel $K$. For the local quadratic estimator, $K^{\text{eq}}_2(u)=%
\frac{\mu_4-\mu_2u^2}{\mu_4-\mu_2^2}K(u)$, where $\mu_j\equiv \int
u^jK(u)\,du$.}

Condition (iii) of Theorem \ref{thm:asymp-dist-2} first assumes that the
relevant nuisance functions are sufficiently smooth for the corresponding
nonparametric analysis. It then requires that the synthetic control weights
be uniformly identified. Assumption \ref{asm:SC} implies the pre-treatment
synthetic control equations $M(x)\bm{w}_{0}(x)=m_{1}(x)$. The lower bound on
the smallest eigenvalue of $M(x)^{\prime }M(x)$ ensures that $M(x)$ has full
column rank uniformly over $x$. Intuitively, this requires the relative
pre-treatment trajectories of the control groups to contain sufficiently
independent time-series variation. This condition also ensures the stability
of the estimated weights: if the minimum singular value of $M(x)$ is small,
then a small perturbation of the estimated pre-treatment trajectories may
produce a large change in the estimated weights.

A convenient scale-free diagnostic for this stability condition is the
condition number $\kappa(M(x))$ of the pre-treatment donor-difference
matrix, defined as the ratio of its largest to smallest singular value.%
\footnote{%
We thank an anonymous referee for suggesting this diagnostic.} We use the
condition number of the estimated matrix as an empirical diagnostic in
Section \ref{sec:empirical} and examine the finite-sample performance of the
procedure across different condition numbers in Section \ref{sec:simulation}.

Theorem \ref{thm:asymp-dist-2} derives the asymptotic variance of $\hat{%
\theta}$ under Assumption \ref{asm:SC}. The influence function $\psi _{%
\mathrm{SC}}$ is more involved than its counterpart under Assumption \ref%
{asm:PT}, but the proof makes its structure transparent: it is the influence
function under Assumption \ref{asm:PT} plus two correction terms, $\psi _{1}$
and $\psi _{2}$ in \eqref{eqn:psi1} and \eqref{eqn:psi2}, which account for
the first-stage estimation of the propensity score ratios and the synthetic
control weights. Both corrections are linear in the vector of estimable
deviations from Assumption \ref{asm:PT}, $v_{\Delta }(X)=\big(m_{2,\Delta
}(X)-m_{\Delta }(X),\ldots ,m_{N_{G}+1,\Delta }(X)-m_{\Delta }(X)\big)%
^{\prime }$, with coefficients proportional to $f_{X|G=1}(X)/f_{X}(X)$, the
ratio of the covariate density of the treated group to that of the
population. Deviations at covariate values where treated units are scarce
therefore contribute little to the corrections. The corrections vanish when $v_{\Delta
}\equiv 0$, in which case $V_{\mathrm{SC}}$ equals $V_{\mathrm{PT}}$.
Because they vary continuously with $v_{\Delta }$, a small departure from
Assumption \ref{asm:PT} leaves $V_{\mathrm{SC}}$ near $V_{\mathrm{PT}}$.

The asymptotic properties of our estimator fall within the broader framework
of semiparametric two-step estimation theory %
\citep[e.g.,][]{chen2003estimation}. Our main contribution in Theorem \ref%
{thm:asymp-dist-2} is to derive the adjustment to the asymptotic variance to
account for the nuisance estimators, which corresponds to Condition (2.6) in 
\cite{chen2003estimation}. This derivation is particularly challenging in
our setting due to the sophisticated way in which the nuisance conditional
mean functions $m_{g,t}\left( \cdot \right) $ enter the estimating equation
through the synthetic control weights.

While Theorem \ref{thm:asymp-dist-2} derives the asymptotic distribution of $%
\hat{\theta}$ under a local polynomial specification for the nuisance
estimators, our result is expected to extend more broadly. The same
asymptotic distribution should hold when alternative nonparametric
estimation methods, such as sieve estimators, are employed. This is because,
according to \cite{newey1994asymptotic}, the adjustment term in the
asymptotic variance should be the same regardless of the estimator used,
though the technical details may differ.

The nonparametric nuisance estimators discussed thus far pertain to the case
of continuous covariates. When discrete covariates are present, a natural
approach is to partition the data according to their levels and perform the
estimation procedure separately within each partition. The theoretical
framework remains valid for each subgroup defined by the discrete
covariates. This approach will be implemented in our empirical analysis in
Section \ref{sec:empirical}.

\subsection{Bootstrap inference}

To conduct inference, we propose a multiplier bootstrap method to
approximate the asymptotic distribution of $\hat{\theta}$. While analytic
standard errors are available from the asymptotic variance formulas in
Theorems \ref{thm:asymp-dist-1} and \ref{thm:asymp-dist-2}, they involve
complicated expressions, particularly under the synthetic control condition
(cf. the asymptotic variance $V_{\mathrm{SC}}$). Moreover, $\hat{\theta}$ is
highly nonlinear in some of the nuisance functions, most notably through the
inverse $(\hat{M}(x)^{\prime }\hat{M}(x))^{-1}$ in the weight estimator, and
this nonlinearity may make the first-order approximation underlying the
analytic standard error less reliable in finite samples, especially when $%
M(x)$ is ill-conditioned, as the simulations in Section \ref%
{subsec:analytic-sim} confirm. The bootstrap method circumvents these
complex variance formulas and provides a unified approach to inference, as
the bootstrap distribution consistently approximates the sampling
distribution of the estimator, regardless of the identification assumption.

The multiplier bootstrap is performed at the individual level as follows.
Let $\mathcal{W}_{n}\equiv (W_{1},\ldots,W_{n})$ denote the bootstrap
weights. For any function $f$, define the bootstrap operator $P_{n}^{\ast
}[\cdot]$ as $P_{n}^{\ast }[f(S)]\equiv \frac{1}{n}%
\sum_{i=1}^{n}W_{i}f(S_{i})$, which represents an empirical operator
constructed using the bootstrap weights, and should not be confused with any
empirical measure. Let $\hat{\pi}_{1}^{\ast }\equiv P_{n}^{\ast }[\mathcal{G}%
_{1}]$ denote the bootstrap estimator for $\pi_{1}$.

To construct the bootstrap nuisance estimators $\hat{r}_{1,g}^{\ast }$ and $%
\hat{m}_{g,t}^{\ast }$, we follow the same procedure as the local polynomial
estimators in equations (\ref{eqn:hat_r-def}) to (\ref{eqn:hat_m_gt-def}),
but replace the empirical operator $P_{n}$ with the bootstrap operator $%
P_{n}^{\ast }$. The bootstrap synthetic weights estimator $\hat{w}_{g}^{\ast
}$ is constructed using (\ref{eqn:hat_w-def}), replacing $m_{g,t}$ with $%
\hat{m}_{g,t}^{\ast }$. In the multiplier bootstrap process, we use the same
sample splitting procedure as before. Specifically, for an observation $%
X_{i} $ with index $i\in I_{\ell }$, we use the subsamples $\{S_{j}\}$ and
weights $\{W_{j}\}$ with indices $j\in I_{\ell }^{c}$ to compute the
nonparametric estimates $\hat{m}_{g,t}^{\ast ,\ell }(X_{i})$, $\hat{r}%
_{1,g}^{\ast ,\ell }(X_{i})$, and $\hat{w}^{\ast ,\ell }(X_{i})$. The
bootstrap estimator for $\theta $ is%
\begin{equation}
\hat{\theta}^{\ast }\equiv \frac{1}{n}\sum_{\ell =1}^{L}\sum_{i\in I_{\ell
}}W_{i}\phi (S_{i};\hat{m}_{\Delta }^{\ast ,\ell },\hat{r}^{\ast ,\ell },%
\hat{w}^{\ast ,\ell };\hat{\pi}_{1}^{\ast }).  \label{eqn:theta-hat-star-def}
\end{equation}%
With a slight abuse of notation, we use, in the above, the ratios $\hat{r}%
^{\ast ,\ell }\equiv (\hat{r}_{1,g}^{\ast ,\ell },g=2,\ldots ,N_{G}+1)$ in
place of the raw propensity scores in $\phi$. Correspondingly, $\hat{p}%
_{1}^{\ast }/\hat{p}_{g}^{\ast }$ denotes $\hat{r} _{1,g}^{\ast }$ in
Theorem \ref{thm:bootstrap} and elsewhere.

The bootstrap weights are required to satisfy the following mild conditions.

	\begin{namedassumption}{BW}[Bootstrap Weights]%
\label{asm:bootstrap-weights} The bootstrap weights $\{W_i\}_{i=1}^n$ are
iid across $i=1,\ldots, n$ and independent of the data $\mathcal{S}_n =
(S_i, 1 \leq i \leq n)$. In addition, $\mathbb{E}[W_i] = 1$, $\mathrm{var}%
(W_i) = 1$, and there exists some $\delta>0$ such that $\mathbb{E}%
[|W_i|^{2+\delta}] < \infty$. 
\end{namedassumption}%

In practice, Assumption \ref{asm:bootstrap-weights} can be implemented with
iid exponential bootstrap weights, as in our empirical and simulation
studies. This multiplier bootstrap approach differs from the nonparametric
bootstrap, where the bootstrap weights $(W_{1},\ldots,W_{n})$ follow a
multinomial distribution and are not independent. This dependence introduces
additional technical challenges in establishing bootstrap consistency.
Consequently, we consider only the multiplier bootstrap in this paper.
Nonetheless, the multiplier bootstrap remains a practical, widely used, and
theoretically valid inference method, as demonstrated by our next theorem.

Before presenting the theorem, we introduce the following convention: for
any random variable $\Delta _{n}$ that is a function of both $\mathcal{W}%
_{n} $ and $\mathcal{S}_{n}$, we say that $\Delta _{n}=o_{p}(1)$ if it
converges in probability to zero under the joint distribution of $(\mathcal{W%
}_{n},\mathcal{S}_{n})$, which follows a product law due to the independence
between $\mathcal{W}_{n}$ and $\mathcal{S}_{n}$. We say that $\Delta
_{n}=O_{p}(1)$ if it is bounded in probability under the joint distribution
of $(\mathcal{W}_{n},\mathcal{S}_{n})$.

\begin{theorem}
\label{thm:bootstrap} Let Assumption \ref{asm:bootstrap-weights} hold. Assume that the elements of $\hat{w}^{\ast }$ sum to one and are bounded
in probability, i.e., $\sum_{g\geq 2}\hat{w}_{g}^{\ast }(x)=1$ for all $x\in
\mathcal{X}$ and $\lVert \hat{w}_{g}^{\ast }\rVert _{\infty }=O_{p}(1),g\geq
2$.

\begin{enumerate}
\item {In addition to the assumptions of Theorem \ref{thm:asymp-dist-1},}
assume that (1) the bootstrap nuisance estimators are $L_{2}$-consistent, i.e.,
\begin{align*}
\left\Vert \hat{m}_{\Delta }^{\ast }-m_{\Delta }\right\Vert _{L_{2}(F_{X})}&
=o_{p}(1), \\
\left\Vert \hat{p}_{1}^{\ast }/\hat{p}_{g}^{\ast }-p_{1}/p_{g}\right\Vert
_{L_{2}(F_{X})}& =o_{p}(1),g\geq 2, \\
\left\Vert \hat{w}_{g}^{\ast }-\omega _{g}\right\Vert _{L_{2}(F_{X})}&
=o_{p}(1),g\geq 2,
\end{align*}%
and (2) their rates satisfy that 
\begin{equation*}
\left\Vert \hat{p}_{1}^{\ast }/\hat{p}_{g}^{\ast }-p_{1}/p_{g}\right\Vert
_{L_{2}(F_{X})}\left\Vert \hat{m}_{\Delta }^{\ast }-m_{\Delta }\right\Vert
_{L_{2}(F_{X})}=o_{p}(n^{-1/2}),g\geq 2.
\end{equation*}%
If, moreover, $V_{\mathrm{PT}}>0$ and $\mathbb{E}[|\phi (S;m_{\Delta
},p,\omega;\pi _{1})-\theta \mathcal{G}_{1}/\pi _{1}|^{2+\delta }]<\infty $
for some $\delta >0$, then 
\begin{equation}
\sup_{a\in \mathbb{R}}\left\vert \mathbb{P}^{\ast }\left( \sqrt{n}(\hat{%
\theta}^{\ast }-\hat{\theta})\leq a\big|\mathcal{S}_{n}\right) -\mathbb{P}%
\left( \sqrt{n}(\hat{\theta}-\theta )\leq a\right) \right\vert =o_{p}(1),
\label{eqn:bootstrap-consistency}
\end{equation}%
where $\mathbb{P}^{\ast }(\cdot |\mathcal{S}_{n})$ is the conditional
probability distribution given the sample $\mathcal{S}_{n}$.

\item {In addition to the assumptions of Theorem \ref{thm:asymp-dist-2},}
assume that the bootstrap nuisance estimators satisfy the following
conditions: for all $g$, 
\begin{equation*}
\left\Vert \hat{m}_{g,t}^{\ast }-m_{g,t}\right\Vert
_{L_{2}(F_{X})},\left\Vert \hat{r}_{1,g}^{\ast }-r_{1,g}\right\Vert
_{L_{2}(F_{X})},\left\Vert \hat{w}_{g}^{\ast }-w_{g}\right\Vert
_{L_{2}(F_{X})}=o_{p}(n^{-1/4}),
\end{equation*}%
and 
\begin{align*}
\hat{r}_{1,g}^{\ast }(x)-r_{1,g}(x)& =P_{n}^{\ast }\left[ \frac{1}{h^d}K^{%
\text{eq}}_{\bar{s}}\left(\frac{X-x}{h}\right) \frac{\mathcal{G}%
_{1}-r_{1,g}(X)\mathcal{G}_{g}}{p_{g}(X)f_{X}(x)}\right] +o_{p}(n^{-1/2}), \\
\hat{m}_{g,t}^{\ast }(x)-m_{g,t}(x)& =P_{n}^{\ast }\left[\frac{1}{h^d}K^{%
\text{eq}}_{\bar{s}}\left(\frac{X-x}{h}\right) \frac{\mathcal{G}%
_{g}(Y_{t}-m_{g,t}(X))}{p_{g}(X)f_{X}(x)} \right] +o_{p}(n^{-1/2}),1\leq
t\leq \mathcal{T}-1,
\end{align*}%
uniformly over $x\in \mathcal{X}$.\footnote{%
We define $\hat{r}_{1,g}^{\ast }\left( \cdot \right) $ and $\hat{w}%
_{g}^{\ast }\left( \cdot \right) $ to be the constant function $\boldsymbol{1%
}\left( \cdot \right) $ when $g=1.$} If, moreover, $V_{\mathrm{SC}}>0$ and
the influence function $\psi _{\mathrm{SC}}\left( S\right) $ defined in (\ref%
{eqn:psi_SC}) has a finite $2+\delta $ moment for some $\delta >0$, then (%
\ref{eqn:bootstrap-consistency}) holds.
\end{enumerate}
\end{theorem}

In Theorem \ref{thm:bootstrap}, all the $o_{p}\left( \cdot \right) $ and $%
O_{p}\left( \cdot \right) $ terms in the conditions are understood to hold
under the joint distribution of the bootstrap weights $\mathcal{W}_{n}$ and
the sample $\mathcal{S}_{n}$. These conditions are straightforward to
verify, as the multiplier bootstrap weights are independent of the sample.

In the bootstrap setting, there is another notion of convergence in
probability, which we use in the proof of Theorem \ref{thm:bootstrap}. We
say that $\Delta _{n}=o_{p}^{\ast }(1)$ if for any $\epsilon >0$, $\mathbb{P}%
^{\ast }(|\Delta _{n}|>\epsilon $ $|$ $\mathcal{S}_{n})=o_{p}(1)$.
Similarly, there is another notion of boundedness in probability. We say
that $\Delta _{n}=O_{p}^{\ast }(1)$ if for any $\epsilon >0$, there exists a
positive constant $C_{\epsilon }$ such that $\mathbb{P}($ $\mathbb{P}^{\ast
}(|\Delta _{n}|>C_{\epsilon }|\mathcal{S}_{n})>\epsilon )\rightarrow 0$. In
the proof of the theorem, we utilize Lemma 3 from \cite{cheng2010bootstrap}
to convert $o_{p}(1)$ and $O_{p}(1)$ terms, respectively, into $o_{p}^{\ast
}(1) $ and $O_{p}^{\ast }(1)$ terms, thereby establishing the result of
conditional weak convergence.

The difference between the bootstrap theory of the estimator under parallel
trends and synthetic control is analogous to that in estimation theory.
Under the parallel trends assumption, the conditions on the bootstrap
estimators of the nuisance functions are relatively weaker (only requiring
consistency, a product rate condition, and a uniform bound on the weights)
due to orthogonality.\footnote{\label{fn:pt-simple-bootstrap} Notably, these
bootstrap estimators of the nuisance functions can coincide with the
original sample estimators in Theorem \ref{thm:asymp-dist-1}, as they
automatically satisfy the rate conditions. This is because the moment
function is insensitive to the nuisance estimators. When there is strong
evidence supporting Assumption \ref{asm:PT}, computation can therefore be
substantially simplified by holding the original nuisance estimates fixed
across bootstrap replications. However, this shortcut is not generally valid
under the synthetic control condition and hence cannot be used when the
identifying assumption is uncertain.} Under the synthetic control
assumption, the bootstrap nuisance estimators are additionally required to
satisfy the Bahadur representation adjusted by the bootstrap weights.

To construct a $100(1-\alpha)\%$ confidence interval based on Theorem \ref%
{thm:bootstrap}, we take the $(1-\alpha)$th quantile of $|\hat{\theta}^* - 
\hat{\theta}| $, denoted as $|\hat{\theta}^* - \hat{\theta}|_{1-\alpha}$.
The confidence interval is then given by $[\hat{\theta} - |\hat{\theta}^* - 
\hat{\theta}|_{1-\alpha}, \hat{\theta} + |\hat{\theta}^* - \hat{\theta}%
|_{1-\alpha}]$.

\section{Repeated cross-sectional data}

\label{sec:rc}

In practice, the ideal panel data structure, where we have repeated
observations of the same individuals or units across multiple time periods,
is not always available. Researchers commonly utilize repeated
cross-sectional data as an alternative, where observations are collected
across distinct time periods without tracking the same individuals. This
data structure is prevalent in scenarios where longitudinal linkage is
either inherently impractical or ruled out by design. Notable applications
include administrative records of episodic events, such as insurance claims %
\citep{meyer1995workers} and car accidents \citep{cohen2003effects},
longitudinal analyses of fixed-age cohorts over extended periods %
\citep{corak2001death}, and nationally representative surveys, such as the
Current Population Survey \citep{acemoglu2001consequences} and the Canadian
General Social Surveys \citep{finkelstein2002effect}, where different
individuals are sampled in each wave.

In repeated cross-sectional data, we do not observe $\{Y_{t}:t=1,\ldots ,%
\mathcal{T}\}$ directly. Instead, each data entry includes a time indicator
variable (time stamp) $T\in \{1,\ldots ,\mathcal{T}\}$, which specifies the
time period for that entry. Denote $\mathfrak{T}_{t}\equiv \mathbf{1}\{T=t\}$%
, which is a binary variable indicating whether the time stamp equals period 
$t$. The observed outcome is then given by 
\begin{equation*}
Y\equiv \sum_{t=1}^{\mathcal{T}}\mathfrak{T}_{t}Y_{t}.
\end{equation*}%
The observed variables in the repeated cross-sectional dataset are $S^{\func{%
rc}}\equiv (Y,G,T,X)$. We impose the following conditions on the sampling
scheme.

	\begin{namedassumption}{RCS}[Repeated Cross-Sectional Sampling] \label{asm:rc}
We observe an iid sample $\{S^{\func{rc}}_i = (Y_i, G_i, T_i, X_i): 1 \leq i
\leq n\}$ of $S^{\func{rc}} = (Y, G, T, X)$. 
\end{namedassumption}%

	\begin{namedassumption}{TI}[Time Invariance] \label{asm:time-invariance}%
The joint distribution of the potential outcomes, covariates, and group
indicator is independent of the time indicator $T$, i.e., 
\begin{align*}
(Y_1(0),\ldots,Y_{\mathcal{T}}(0),Y_{\mathcal{T}}(1),X,G)\perp \!\!\! \perp
T.
\end{align*}
Furthermore, for each period $t$, $\lambda_t \equiv \mathbb{P}(T=t)>0.$ 
\end{namedassumption}%

Under the repeated cross-sectional sample scheme in Assumption \ref{asm:rc}, 
$T_{i}$ is a categorical variable indicating the time period to which
observation $i$ belongs, and $\mathfrak{T}_{ti}=\mathbf{1}\{T_{i}=t\}$ is a
binary variable indicating whether observation $i$ belongs to time period $t$%
.

Assumption \ref{asm:time-invariance} aligns with the time invariance
conditions commonly found in the DiD literature, such as Assumption 3.3 in 
\cite{abadie2005semiparametric}. Under this assumption, the joint
distribution of $(Y,G,X,T)$ can be expressed through the following mixture
representation: 
\begin{equation*}
\mathbb{P}(Y\leq y,G=g,X\leq x,T=t)=\lambda _{t}\mathbb{P}(Y_{t}\leq
y,G=g,X\leq x).
\end{equation*}%
Similar to the repeated cross-sectional settings discussed in \cite%
{sant2020doubly} and \cite{callaway2021difference}, Assumption \ref%
{asm:time-invariance} excludes compositional changes, that is, changes in
the distribution of $G$ and $X$ over time. Recent work by \cite%
{santanna2026compositional} has explored compositional changes in the DiD
setting. We leave the investigation of such changes within our framework for
future research.

Since each $Y_{t}$ is not directly observable in repeated cross-sections,
the nuisance parameters $m_{g,t}$ and $m_{\Delta }$ need to be redefined. We
introduce some additional notation. Let $\mu _{g,t}\left( x\right) \equiv 
\mathbb{E}[\mathfrak{T}_{t}Y|G=g,X=x]$ and $\mu _{G\neq 1,t}\left( x\right)
\equiv \mathbb{E}[\mathfrak{T}_{t}Y|G\neq 1,X=x]$, which are directly
identifiable based on the observation of $S^{\text{rc}}$. Then, under
Assumption \ref{asm:time-invariance}, 
\begin{align*}
\mu _{g,t}\left( x\right) & =\mathbb{E}[\mathfrak{T}_{t}Y|G=g,X=x,\mathfrak{T%
}_{t}=1]\mathbb{P} \left( \mathfrak{T}_{t}=1|G=g,X=x\right) \\
& =\mathbb{E}[Y|G=g,X=x,\mathfrak{T}_{t}=1]\cdot \lambda _{t}\equiv
m_{g,t}\left( x\right) \lambda _{t},
\end{align*}%
where $m_{g,t}\left( x\right) =\mathbb{E}[Y|G=g,X=x,\mathfrak{T}_{t}=1].$
Similarly, under Assumption \ref{asm:time-invariance}, 
\begin{equation*}
\mu _{G\neq 1,t}\left( x\right) =\mathbb{E}[Y|G\neq 1,X=x,\mathfrak{T}%
_{t}=1]\cdot \lambda _{t}\equiv m_{G\neq 1,t}\left( x\right) \lambda _{t},
\end{equation*}%
where $m_{G\neq 1,t}\left( x\right) =\mathbb{E}[Y|G\neq 1,X=x,\mathfrak{T}%
_{t}=1].$ Define 
\begin{equation*}
m_{\Delta }\left( x\right) \equiv m_{G\neq 1,\mathcal{T}}\left( x\right)
-m_{G\neq 1,\mathcal{T}-1}\left( x\right) =\frac{\mu _{G\neq 1,\mathcal{T}%
}\left( x\right) }{\lambda _{\mathcal{T}}}-\frac{\mu _{G\neq 1,\mathcal{T}%
-1}\left( x\right) }{\lambda _{\mathcal{T}-1}}.
\end{equation*}

We will still maintain Assumption \ref{asm:PT} or \ref{asm:SC}. To
facilitate interpretation, we rewrite them in an equivalent form. For
Assumption \ref{asm:PT}, we express it as: 
\begin{align*}
& \mathbb{E}\left[ Y_{\mathcal{T}}\left( 0\right) |G=1,X\right] -\mathbb{E}%
\left[ Y_{\mathcal{T}-1}\left( 0\right) |G=1,X\right] = \mathbb{E}\left[ Y_{%
\mathcal{T}}\left( 0\right) |G=g,X\right] -\mathbb{E}\left[ Y_{\mathcal{T}%
-1}\left( 0\right) |G=g,X\right]
\end{align*}%
or equivalently, with $Y(0)\equiv \sum_{t=1}^{\mathcal{T}}\mathfrak{T}%
_{t}Y_{t}(0)$ denoting the untreated potential outcome in the sampled
period, 
\begin{align*}
& \mathbb{E}\left[ Y\left( 0\right) |G=1,X,\mathfrak{T}_{\mathcal{T}}=1%
\right] -\mathbb{E}\left[ Y\left( 0\right) |G=1,X,\mathfrak{T}_{\mathcal{T}%
-1}=1\right] \\
=\,& \mathbb{E}\left[ Y\left( 0\right) |G=g,X,\mathfrak{T}_{\mathcal{T}}=1%
\right] -\mathbb{E}\left[ Y\left( 0\right) |G=g,X,\mathfrak{T}_{\mathcal{T}%
-1}=1\right] .
\end{align*}%
This says that the change in the conditional mean of the baseline outcome
for the treatment group from period $\mathcal{T}-1$ to period $\mathcal{T}$
is the same as that for any of the control groups. For each of the two time
periods, the conditional mean is taken only over the observable individuals.

For Assumption \ref{asm:SC}, we express it as 
\begin{equation*}
\mathbb{E}\left[ Y\left( 0\right) |G=1,X=x,\mathfrak{T}_{t}=1\right]
=\sum_{g=2}^{N_{G}+1}w_{g}\left( x\right) \mathbb{E}\left[ Y\left( 0\right)
|G=g,X=x,\mathfrak{T}_{t}=1\right] ,
\end{equation*}%
for $t=1,\ldots ,\mathcal{T}$. So the synthetic control condition is imposed
on the conditional means of the baseline outcome for the observable
individuals in each group. The weight $w_{g}\left( x\right) $ can be
identified from the following equations: 
\begin{equation*}
\mu _{1,t}\left( x\right) =\sum_{g=2}^{N_{G}+1}w_{g}\left( x\right) \mu
_{g,t}\left( x\right) \text{ for }t=1,\ldots ,\mathcal{T}-1,
\end{equation*}%
where the factor $1/\lambda _{t}$ on both sides cancels out. Similar to (\ref%
{eqn:hat_w-def}), we can solve for $\bm{w}_{0}(x)$, which is the vector of
weights excluding the last entry, as $\bm{w}_{0}(x)=(M_{\text{rc}}{(x)}%
^{\prime }M_{\text{rc}}(x))^{-1}M_{\text{rc}}{(x)}^{\prime }m_{\text{rc}%
,1}(x)$, where 
\begin{equation*}
M_{\text{rc}}(x)\equiv 
\begin{pmatrix}
\mu _{-1,1}{(x)}^{\prime } \\ 
\mu _{-1,2}{(x)}^{\prime } \\ 
\vdots \\ 
\mu _{-1,\mathcal{T}-1}{(x)}^{\prime }%
\end{pmatrix}%
,m_{\text{rc},1}(x)\equiv 
\begin{pmatrix}
\mu _{1,1}(x)-\mu _{N_{G}+1,1}(x) \\ 
\vdots \\ 
\mu _{1,\mathcal{T}-1}(x)-\mu _{N_{G}+1,\mathcal{T}-1}(x)%
\end{pmatrix}%
,
\end{equation*}%
with $\mu _{-1,t}(x)\equiv (\mu _{2,t}(x)-\mu _{N_{G}+1,t}(x),\ldots ,\mu
_{N_{G},t}(x)-\mu _{N_{G}+1,t}(x))^{\prime }$. The matrix $M_{\text{rc}}(x)$
is the repeated-cross-sectional analogue of the pre-treatment
donor-difference matrix.

We write the doubly robust moment function for $\theta $ in the repeated
cross-sectional setting as 
\begin{align*}
& \phi_{\text{rc}}(S^{\text{rc}};\mu _{G\neq 1,\mathcal{T}},\mu _{G\neq 1,%
\mathcal{T}-1},p,w;\pi _{1},\lambda _{\mathcal{T}},\lambda _{\mathcal{T}-1})
\\
=\,& \frac{1}{\pi _{1}}\left( \mathcal{G}_{1}-\sum_{g=2}^{N_{G}+1}\mathcal{G}%
_{g}w_{g}(X)\frac{p_{1}(X)}{p_{g}(X)}\right) \left( \left( \frac{\mathfrak{T}%
_{\mathcal{T}}}{\lambda _{\mathcal{T}}}-\frac{\mathfrak{T}_{\mathcal{T}-1}}{%
\lambda _{\mathcal{T}-1}}\right) Y-\left( \frac{\mu _{G\neq 1,\mathcal{T}}(X)%
}{\lambda _{\mathcal{T}}}-\frac{\mu _{G\neq 1,\mathcal{T}-1}(X)}{\lambda _{%
\mathcal{T}-1}}\right) \right) .
\end{align*}

\begin{corollary}
\label{thm:rc-identification} Let Assumptions \ref{asm:NA}, \ref{asm:overlap}%
, \ref{asm:rc}, and \ref{asm:time-invariance} hold. \ 

\begin{enumerate}
\item If Assumption \ref{asm:PT} holds, then 
\begin{align*}
\theta & =\mathbb{E}[\phi _{\text{rc}}(S^{\text{rc}};\mu _{G\neq 1,\mathcal{T%
}},\mu _{G\neq 1,\mathcal{T}-1},\tilde{p},\tilde{w};\pi _{1},\lambda _{%
\mathcal{T}},\lambda _{\mathcal{T}-1})] \\
& =\mathbb{E}[\phi _{\text{rc}}(S^{\text{rc}};\tilde{\mu}_{G\neq 1,\mathcal{T%
}},\tilde{\mu}_{G\neq 1,\mathcal{T}-1},p,\tilde{w};\pi _{1},\lambda _{%
\mathcal{T}},\lambda _{\mathcal{T}-1})] \\
& =\mathbb{E}[\phi _{\text{rc}}(S^{\text{rc}};\mu _{G\neq 1,\mathcal{T}},\mu
_{G\neq 1,\mathcal{T}-1},p,\tilde{w};\pi _{1},\lambda _{\mathcal{T}},\lambda
_{\mathcal{T}-1})],
\end{align*}%
for any vector of weight functions $\tilde{w}$ whose elements sum to one,
any positive functions $\tilde{p}$, and any functions $(\tilde{\mu}_{G\neq 1,%
\mathcal{T}},\tilde{\mu}_{G\neq 1,\mathcal{T}-1})$.

\item If Assumption \ref{asm:SC} holds, then 
\begin{equation*}
\theta =\mathbb{E}[\phi _{\text{rc}}(S^{\text{rc}};\tilde{\mu}_{G\neq 1,%
\mathcal{T}},\tilde{\mu}_{G\neq 1,\mathcal{T}-1},p,w;\pi _{1},\lambda _{%
\mathcal{T}},\lambda _{\mathcal{T}-1})],
\end{equation*}%
for any functions $(\tilde{\mu}_{G\neq 1,\mathcal{T}},\tilde{\mu}_{G\neq 1,%
\mathcal{T}-1})$, where $w$ is the vector of weight functions satisfying (%
\ref{eqn:w-id}).
\end{enumerate}
\end{corollary}

For estimation, we consider the same cross-fitting method as in the panel
setting. The estimator for $\theta $ is 
\begin{equation*}
\hat{\theta}_{\text{rc}}\equiv \frac{1}{n}\sum_{\ell =1}^{L}\sum_{i\in
I_{\ell }}\phi _{\text{rc}}(S_{i}^{\text{rc}};\hat{\mu}_{G\neq 1,\mathcal{T}%
}^{\ell },\hat{\mu}_{G\neq 1,\mathcal{T}-1}^{\ell },\hat{p}^{\ell },\hat{w}%
^{\ell };\hat{\pi}_{1},\hat{\lambda}_{\mathcal{T}},\hat{\lambda}_{\mathcal{T}%
-1}),
\end{equation*}%
where $\hat{\mu}_{G\neq 1,\mathcal{T}}^{\ell },\hat{\mu}_{G\neq 1,\mathcal{T}%
-1}^{\ell },\hat{p}^{\ell },$ and $\hat{w}^{\ell }$ are nuisance estimators
constructed using samples not in the $\ell $th fold, and $\hat{\lambda}%
_{t}\equiv P_{n}\mathfrak{T}_{t}$ is the estimator for $\lambda _{t}$, for $%
t=\mathcal{T}-1,\mathcal{T}$.

\begin{theorem}
\label{thm:rc-PT} Let Assumptions \ref{asm:NA}, \ref{asm:overlap}, \ref%
{asm:PT}, \ref{asm:rc}, \ref{asm:time-invariance}, and the following
conditions hold:

\begin{enumerate}
\item For each $g\geq 1$ and $t\in \{\mathcal{T}-1,\mathcal{T}\}$, $\mathbb{E%
}[(\mathfrak{T}_{t}Y-\mu _{G\neq 1,t}(X))^{2}|X,G=g]$, $1/p_{g}$, and $%
\omega _{g}$ are bounded functions of $X$.

\item The first-stage estimators satisfy that (1) the estimated weight
functions sum to one and are bounded in probability, i.e., $\sum_{g\geq 2}%
\hat{w}_{g}\left( x\right) =1$ for all $x\in \mathcal{X}$ and $\lVert \hat{w}%
_{g}\rVert _{\infty }=O_{p}(1),g\geq 2$, (2) the estimators are $L_{2}$%
-consistent, i.e., 
\begin{align*}
\left\Vert \hat{\mu}_{G\neq 1,t}-\mu _{G\neq 1,t}\right\Vert
_{L_{2}(F_{X})}& =o_{p}(1),t=\mathcal{T}-1,\mathcal{T}, \\
\left\Vert \hat{p}_{1}/\hat{p}_{g}-p_{1}/p_{g}\right\Vert _{L_{2}(F_{X})}&
=o_{p}(1),g\geq 2, \\
\left\Vert \hat{w}_{g}-\omega _{g}\right\Vert _{L_{2}(F_{X})}&
=o_{p}(1),g\geq 2,
\end{align*}%
and (3) their rates satisfy that 
\begin{equation*}
\left\Vert \hat{p}_{1}/\hat{p}_{g}-p_{1}/p_{g}\right\Vert
_{L_{2}(F_{X})}\left\Vert \hat{\mu}_{G\neq 1,t}-\mu _{G\neq 1,t}\right\Vert
_{L_{2}(F_{X})}=o_{p}(n^{-1/2}),g\geq 2,t=\mathcal{T}-1,\mathcal{T}.
\end{equation*}
\end{enumerate}

Then $\sqrt{n}(\hat{\theta}_{\text{rc}}-\theta )\overset{d}{\rightarrow }%
N(0,V_{\mathrm{PT}}^{\text{rc}})$, where $V_{\mathrm{PT}}^{\text{rc}}=%
\mathbb{E}[\psi _{\mathrm{PT}}^{\text{rc}}(S^{\text{rc}})^{2}]$, and $\psi _{%
\mathrm{PT}}^{\text{rc}}(S^{\text{rc}})$ is the influence function given in (%
\ref{eqn:IF-theta-rc-PT}) in the proof.
\end{theorem}

\begin{theorem}
\label{thm:rc-SC} Let Assumptions \ref{asm:NA}, \ref{asm:overlap}, \ref%
{asm:SC}, \ref{asm:rc}, \ref{asm:time-invariance}, conditions (i) and (ii)
of Theorem \ref{thm:rc-PT}, and the following conditions hold:

\begin{enumerate}
\item The convergence rates of the nuisance estimators satisfy 
\begin{equation*}
\left\Vert \hat{\mu}_{g,t}-\mu _{g,t}\right\Vert _{L_{2}(F_{X})},\left\Vert 
\hat{r}_{1,g}-r_{1,g}\right\Vert _{L_{2}(F_{X})},\left\Vert \hat{w}%
_{g}-w_{g}\right\Vert _{L_{2}(F_{X})}=o_{p}(n^{-1/4}),\text{ for all }g.
\end{equation*}

\item Let $h=o\big(n^{-1/\{2(\bar{s}+1)\}}\big)$ be an undersmoothing
bandwidth, and assume that $\mathcal{X}$ is compact. Uniformly over $x\in%
\mathcal{X}$, the nuisance estimators admit the following Bahadur
representations: for all $g$, 
\begin{align*}
\hat{r}_{1,g}(x)-r_{1,g}(x)& =P_{n}\left[ \frac{1}{h^d}K^{\text{eq}}_{\bar{s}%
}\left(\frac{X-x}{h}\right)\frac{\mathcal{G}_{1}-r_{1,g}(X)\mathcal{G}_{g}}{%
p_{g}(X)f_{X}(x)}\right] +o_{p}(n^{-1/2}), \\
\hat{\mu}_{g,t}(x)-\mu _{g,t}(x)& =P_{n}\left[ \frac{1}{h^d}K^{\text{eq}}_{%
\bar{s}}\left(\frac{X-x}{h}\right)\frac{\mathcal{G}_{g}(\mathfrak{T}%
_{t}Y-\mu _{g,t}(X))}{p_{g}(X)f_{X}(x)}\right] +o_{p}(n^{-1/2}),1\leq t\leq 
\mathcal{T}-1,
\end{align*}%
where the kernel $K^{\text{eq}}_{\bar{s}}$ is as in Theorem \ref%
{thm:asymp-dist-2}.

\item The functions $w_{g}$, $\mu _{g,t}$, and $p_{g}$, for $1\leq g\leq
N_{G}+1,1\leq t\leq \mathcal{T},$ are $(\bar{s}+1)$-times continuously
differentiable with $\bar{s}\geq d$. There exists a constant $c>0$ such that
the smallest eigenvalue of $M_{\text{rc}}(x)^{\prime }M_{\text{rc}}(x)$ is
larger than $c$ for any $x\in \mathcal{X}$.
\end{enumerate}

Let $\psi _{\mathrm{SC}}^{\text{rc}}(S^{\text{rc}})$ be the influence
function defined in (\ref{eqn:IF-theta-rc-SC}) in the proof, and assume that 
$V_{\mathrm{SC}}^{\text{rc}}=\mathbb{E}[\psi _{\mathrm{SC}}^{\text{rc}}(S^{%
\text{rc}})^{2}]$ is finite. Then $\sqrt{n}(\hat{\theta}_{\text{rc}}-\theta )%
\overset{d}{\rightarrow }N(0,V_{\mathrm{SC}}^{\text{rc}}).$
\end{theorem}

Similar to the panel setting, we construct the bootstrap estimator as 
\begin{equation*}
\hat{\theta}_{\text{rc}}^{\ast }\equiv \frac{1}{n}\sum_{\ell
=1}^{L}\sum_{i\in I_{\ell }}W_{i}\phi _{\text{rc}}(S_{i}^{\text{rc}};\hat{\mu%
}_{G\neq 1,\mathcal{T}}^{\ast ,\ell },\hat{\mu}_{G\neq 1,\mathcal{T}%
-1}^{\ast ,\ell },\hat{p}^{\ast ,\ell },\hat{w}^{\ast ,\ell };\hat{\pi}%
_{1}^{\ast },\hat{\lambda}_{\mathcal{T}}^{\ast },\hat{\lambda}_{\mathcal{T}%
-1}^{\ast }),
\end{equation*}%
where $\hat{\lambda}_{t}^{\ast }\equiv P_{n}^{\ast }\mathfrak{T}_{t}/P_{n}^{\ast }[1],t=%
\mathcal{T}-1,\mathcal{T}$, and the bootstrap nonparametric nuisance
estimators are constructed analogously to the panel setting.

\begin{corollary}
\label{cor:rc-bootstrap} Suppose that the bootstrap weights satisfy
Assumption \ref{asm:bootstrap-weights}, and that the elements of $%
\hat{w}^{\ast }$ sum to one and are bounded in probability, i.e., $%
\sum_{g\geq 2}\hat{w}_{g}^{\ast }\left( x\right) =1$ for all $x\in \mathcal{X}$ and
$\lVert \hat{w}_{g}^{\ast }\rVert _{\infty }=O_{p}(1),g\geq 2$.

\begin{enumerate}
\item Let the assumptions of Theorem \ref{thm:rc-PT} hold. We further assume
that (1) the bootstrap nuisance estimators are $L_{2}$-consistent, i.e.,
\begin{align*}
\left\Vert \hat{\mu}_{G\neq 1,t}^{\ast }-\mu _{G\neq 1,t}\right\Vert
_{L_{2}(F_{X})}& =o_{p}(1),t=\mathcal{T}-1,\mathcal{T}, \\
\left\Vert \hat{p}_{1}^{\ast }/\hat{p}_{g}^{\ast }-p_{1}/p_{g}\right\Vert
_{L_{2}(F_{X})}& =o_{p}(1),g\geq 2, \\
\left\Vert \hat{w}_{g}^{\ast }-\omega _{g}\right\Vert _{L_{2}(F_{X})}&
=o_{p}(1),g\geq 2,
\end{align*}%
and (2) their rates satisfy that
\begin{equation*}
\left\Vert \hat{p}_{1}^{\ast }/\hat{p}_{g}^{\ast }-p_{1}/p_{g}\right\Vert
_{L_{2}(F_{X})}\left\Vert \hat{\mu}_{G\neq 1,t}^{\ast }-\mu _{G\neq
1,t}\right\Vert _{L_{2}(F_{X})}=o_{p}(n^{-1/2}),g\geq 2,t=\mathcal{T}-1,%
\mathcal{T}.
\end{equation*}%
In addition, assume that $V_{\mathrm{PT}}^{\text{rc}}>0$ and that the
influence function defined in (\ref{eqn:IF-theta-rc-PT}) has a finite $%
2+\delta $ moment for some positive $\delta $. Then we have%
\begin{equation}
\sup_{a\in \mathbb{R}}\left\vert \mathbb{P}^{\ast }\left( \sqrt{n}(\hat{%
\theta}_{\text{rc}}^{\ast }-\hat{\theta}_{\text{rc}})\leq a|\mathcal{S}%
_{n}\right) -\mathbb{P}\left( \sqrt{n}(\hat{\theta}_{\text{rc}}-\theta )\leq
a\right) \right\vert =o_{p}(1).  \label{eqn:rc-bootstrap-consistency}
\end{equation}

\item Let the assumptions of Theorem \ref{thm:rc-SC} hold. We further assume
that the bootstrap nuisance estimators satisfy the following conditions: for
all $g$, 
\begin{equation*}
\left\Vert \hat{\mu}_{g,t}^{\ast }-\mu _{g,t}\right\Vert
_{L_{2}(F_{X})},\left\Vert \hat{r}_{1,g}^{\ast }-r_{1,g}\right\Vert
_{L_{2}(F_{X})},\left\Vert \hat{w}_{g}^{\ast }-w_{g}\right\Vert
_{L_{2}(F_{X})}=o_{p}(n^{-1/4}),
\end{equation*}%
and 
\begin{align*}
\hat{r}_{1,g}^{\ast }(x)-r_{1,g}(x)& =P_{n}^{\ast }\left[ \frac{1}{h^d}K^{%
\text{eq}}_{\bar{s}}\left(\frac{X-x}{h}\right)\frac{\mathcal{G}%
_{1}-r_{1,g}(X)\mathcal{G}_{g}}{p_{g}(X)f_{X}(x)}\right] +o_{p}(n^{-1/2}), \\
\hat{\mu}_{g,t}^{\ast }(x)-\mu _{g,t}(x)& =P_{n}^{\ast }\left[ \frac{1}{h^d}%
K^{\text{eq}}_{\bar{s}}\left(\frac{X-x}{h}\right)\frac{\mathcal{G}_{g}(%
\mathfrak{T}_{t}Y-\mu _{g,t}(X))}{p_{g}(X)f_{X}(x)}\right]
+o_{p}(n^{-1/2}),1\leq t\leq \mathcal{T}-1,
\end{align*}%
uniformly over $x\in \mathcal{X}$. In addition, assume that $V_{\mathrm{SC}%
}^{\text{rc}}>0$ and that the influence function $\psi _{\mathrm{SC}}^{\text{%
rc}}(S^{\text{rc}})$ defined in (\ref{eqn:IF-theta-rc-SC}) has a finite $%
2+\delta $ moment for some positive $\delta $. Then (\ref%
{eqn:rc-bootstrap-consistency}) holds.
\end{enumerate}
\end{corollary}

\section{Staggered treatment adoption}

\label{sec:staggered}

Staggered designs, where different groups adopt the treatment at varying
times, are a common feature of panel datasets (see, e.g., %
\citealp{baker2026practitioner}). Our methodology can be readily extended to
accommodate such staggered designs. Let $D_{t}$ represent the treatment
indicator in period $t$. In the settings of previous sections, $D_{t}$
equals one if and only if the unit belongs to the treatment group and the
time period is the final one. Under the staggered design, however, $D_{t} $
can equal one for any group and in any time period. We focus on the common
case where the treatment status is irreversible so that $D_{t}\geq D_{t-1}$.
In other words, once a unit becomes treated, it remains treated throughout.

The initial treatment time for an individual, $\min \{t: D_t = 1\}$, is
determined by the group to which the individual belongs, and we write it as $%
\gamma(G)$. In the empirical contexts where states serve as the aggregate
units, this implies that different states will have their own timing for
policy changes. However, once a state implements the change, the entire
state adopts the new policy uniformly. Additionally, multiple states are
permitted to adopt the treatment at the same time. We adopt the convention
that the minimum element of an empty set is $\infty$. When $\gamma(G) =
\infty$, the individual remains untreated throughout the entire time span of
the dataset. The eventually treated units are those with $\gamma(G) < \infty$%
. We assume that no group is treated in the first period, that is, $%
\gamma(g)\geq 2$ for all $g$. A group treated from period 1 has no
pre-treatment period, so its ATT cannot be identified.

In general, potential outcomes in a staggered treatment design should be
defined based on all potential treatment trajectories $\left( D_{1},\ldots
,D_{\mathcal{T}}\right) $, and we may denote them as $\tilde{Y}_{t}\left(
d_{1},\ldots ,d_{\mathcal{T}}\right) $. Given the irreversibility of the
treatment, we adopt the simple notation $Y_{t}(s)=\tilde{Y}_{t}\left(
d_{1},\ldots ,d_{\mathcal{T}}\right) $ where $s=\min \{\tau :d_{\tau }=1\}.$ That
is, $Y_{t}(s)$ represents the potential outcome at time $t$ had an
individual belonged to a group whose time of treatment adoption was $s$. The
observed outcome is $Y_{t}\equiv Y_{t}(\gamma(G))=\sum_{g=1}^{N_{G}+1}%
\mathcal{G}_{g}Y_{t}(\gamma(g))$. Following the literature, we assume that,
prior to treatment, a unit's potential outcomes are equal to its
never-treated potential outcomes, a generalization of the no-anticipation
assumption to the staggered treatment setting.

	\begin{namedassumption}{NA-stg}[No Anticipation in Staggered Design] \label{asm:NA-staggered}
For $t < \gamma(g)$, $Y_{t}(\gamma(g)) = Y_t(\infty)$. 
\end{namedassumption}%

The ATT of group $g$ (with $\gamma \left( g\right) <\infty $) in the
post-treatment period $t\geq \gamma \left( g\right) $ is 
\begin{equation*}
\func{ATT}(g,t)\equiv \mathbb{E}[Y_{t}(\gamma (g))-Y_{t}(\infty )|G=g],\text{
for }t\geq \gamma (g).
\end{equation*}%
To analyze how the treatment effect evolves over time, a common approach is
to transform calendar time $t$ into event time $e=t-\gamma (g)$, which
measures the time elapsed relative to the initiation of treatment for group $%
g$. Here, we focus on the causal effect parameters: 
\begin{equation}
\left\{ \func{ATT}(g,\gamma (g)+e):e\in \{0,1,\ldots ,\bar{e}\}\right\} ,
\label{es_ATTs}
\end{equation}%
where $\bar{e}$ is a positive integer satisfying $\gamma (g)+\bar{e}\leq 
\mathcal{T}$. These ATT parameters are the building blocks of the event
study parameters. For example, one may compute the following event study
parameter 
\begin{align*}
\func{ES}(e)& \equiv \mathbb{E}[\func{ATT}(G,\gamma (G)+e)|\gamma (G)+e\leq 
\mathcal{T}] \\
&= \sum_{\substack{ g:\gamma(g)<\infty,  \\ \gamma(g)+e\leq\mathcal{T}}} 
\frac{\mathbb{P}(G=g)} {\mathbb{P}\!\left(\gamma(G)+e\leq\mathcal{T}\right)} 
\func{ATT}(g,\gamma(g)+e), e\geq 0,
\end{align*}%
which measures an average effect $e$ periods after treatment. Once the
identification and estimation of the ATTs have been established, the event
study parameters can be derived straightforwardly.

To identify the ATTs, we generalize the parallel trends and synthetic
control assumptions to the staggered design. We need to first specify the
set of control groups for comparison, commonly known as the donor pool in
the literature. For an event time $e\geq 0$, the set of possible donor
groups for the identification of $\func{ATT}(g,\gamma (g)+e)$ consists of
groups that are untreated up to the period $\gamma (g)+e$: $\mathcal{D}%
_{g,e}\equiv \{g^{\prime }:\gamma (g^{\prime })>\gamma (g)+e\}.$ This
definition of the donor groups depends on both the treated group $g$ and
event time $e$, which could make the method overly complex. To avoid a
time-varying donor group, we follow \cite{ben2022synthetic} and focus on a
fixed set of donor groups $\mathcal{D}_{g,\bar{e}}$. This approach ensures
that the estimated counterfactual outcomes remain stable and are not
artificially affected by variations in the composition of the donor groups
across event time. Nevertheless, our proposed method can easily accommodate
varying donor groups by replacing $\mathcal{D}_{g,\bar{e}}$ with $\mathcal{D}%
_{g,e}$ in the procedure. Note that $\mathcal{D}_{g,\bar{e}}$ will always be
non-empty if there are never-treated groups with $\gamma (g^{\prime })=\infty $ in the
data.

The following assumptions are maintained for the identification of the ATTs
given in (\ref{es_ATTs}).

	\begin{namedassumption}{PT-stg}[Parallel Trends in Staggered Design] \label{asm:PT-staggered}
For any $g$ and $g^{\prime }$ such that $g^{\prime }\in \mathcal{D}_{g,\bar{e%
}}$, 
\begin{align*}
\mathbb{E}[Y_t(\infty) - Y_{t-1}(\infty) | G=g,X] = \mathbb{E}[Y_t(\infty) -
Y_{t-1}(\infty) | G=g^{\prime },X], \gamma(g) \leq t \leq \gamma(g) + \bar{e}%
.
\end{align*}
\end{namedassumption}%

	\begin{namedassumption}{SC-stg}[Synthetic Control in Staggered Design] \label{asm:SC-staggered}
For each eventually treated group $g$, there exists a vector of weight
functions $w^g \equiv (w^g_{g^{\prime }},g^{\prime }\in \mathcal{D}_{g,\bar{e%
}})$ whose elements sum to one and satisfy that for almost all $x \in 
\mathcal{X}$, 
\begin{align}
\mathbb{E}[Y_{t}(\infty)|G=g,X=x] = \sum_{g^{\prime }\in \mathcal{D}_{g,\bar{%
e}}} w^g_{g^{\prime }}(x) \mathbb{E}[Y_{t}(\infty)|G=g^{\prime },X=x], t
=1,\ldots,\gamma(g)+\bar{e}.  \label{eqn:w-id-staggered}
\end{align}
Furthermore, $\gamma(g) \geq |\mathcal{D}_{g,\bar{e}}| \geq 2$, for all $g$
with $\gamma(g)<\infty$, where $| \cdot |$ denotes the cardinality of a set. 
\end{namedassumption}%

Similar to \cite{callaway2021difference}, Assumption \ref{asm:PT-staggered}
assumes that the parallel trends assumption holds only starting from the
post-treatment periods. For a given $g$ and $e$, the identification of $%
\func{ATT}(g,\gamma (g)+e)$ follows a logic analogous to the single treated
group and single treatment period framework studied in Section \ref%
{sec:setup}. The adjustments involve shifting the treatment period from $%
\mathcal{T}$ to $\gamma (g)+e$ and the pre-treatment periods from $%
\{1,\ldots ,\mathcal{T}-1\}$ to $\{1,\ldots ,\gamma (g)-1\}$. Additionally,
the set of control groups is updated from $\{2,\ldots ,N_{G}+1\}$ to $%
\mathcal{D}_{g,\bar{e}}$. The condition $\gamma (g)\geq |\mathcal{D}_{g,\bar{%
e}}|$ ensures that there is a sufficient number of pre-treatment periods for
each group to identify the weights. This is similar to the setup in \cite%
{ben2022synthetic} in that all eventually treated groups remain untreated
for some time periods before receiving treatment.

The moment function for estimating $\func{ATT}(g,\gamma (g)+e),$ $e\geq 0,$
is defined as 
\begin{align*}
& \phi ^{g,e}(S;m_{\Delta ,g,\gamma (g)+e},p,w^{g};\pi _{g}) \\
=& \frac{1}{\pi _{g}}\left( \mathcal{G}_{g}-\sum_{g^{\prime }\in \mathcal{D}%
_{g,\bar{e}}}w_{g^{\prime }}^{g}(X)\frac{p_{g}(X)}{p_{g^{\prime }}(X)}%
\mathcal{G}_{g^{\prime }}\right) (Y_{\gamma (g)+e}-Y_{\gamma
(g)-1}-m_{\Delta ,g,\gamma (g)+e}(X)),
\end{align*}%
where $\pi _{g}\equiv \mathbb{P}(G=g)$ and $m_{\Delta ,g,\gamma (g)+e}\left(
x\right) \equiv \mathbb{E}[Y_{\gamma (g)+e}-Y_{\gamma (g)-1}|G\in \mathcal{D}%
_{g,\bar{e}},X=x]$. The dependence of $m_{\Delta ,g,\gamma (g)+e}\left(
x\right) $, $w_{g^{\prime }}^{g}\left( \cdot \right) $ and $\phi ^{g,e}$ on $%
\bar{e}$ is suppressed.

\begin{corollary}
\label{cor:staggered} Let Assumptions \ref{asm:panel}, \ref{asm:overlap},
and \ref{asm:NA-staggered} hold. \ 

\begin{enumerate}
\item If Assumption \ref{asm:PT-staggered} holds and $0\leq e\leq \bar{e}$,
then 
\begin{eqnarray*}
\func{ATT}(g,\gamma (g)+e) &=&\mathbb{E}[\phi ^{g,e}(S;m_{\Delta ,g,\gamma
(g)+e},\tilde{p},\tilde{w}^{g};\pi _{g})] \\
&=&\mathbb{E}[\phi ^{g,e}(S;\tilde{m}_{\Delta ,g,\gamma (g)+e},p,\tilde{w}%
^{g};\pi _{g})] \\
&=&\mathbb{E}[\phi ^{g,e}(S;m_{\Delta ,g,\gamma (g)+e},p,\tilde{w}^{g};\pi
_{g})],
\end{eqnarray*}%
for any weight functions $\tilde{w}^{g}$ satisfying $\sum_{g^{\prime }\in 
\mathcal{D}_{g,\bar{e}}}\tilde{w}_{g^{\prime }}^{g}(x)=1$ for all $x\in 
\mathcal{X}$, any positive functions $\tilde{p}$, and any functions $\tilde{m%
}_{\Delta ,g,\gamma (g)+e}$.

\item If Assumption \ref{asm:SC-staggered} holds and $0\leq e\leq \bar{e}$,
then 
\begin{equation*}
\func{ATT}(g,\gamma (g)+e)=\mathbb{E}[\phi ^{g,e}(S;\tilde{m}_{\Delta
,g,\gamma (g)+e},p,w^{g};\pi _{g})],
\end{equation*}%
for any function $\tilde{m}_{\Delta ,g,\gamma (g)+e}$, where $w^{g}$ is the
vector of weight functions satisfying (\ref{eqn:w-id-staggered}).
\end{enumerate}
\end{corollary}

The semiparametric estimation and bootstrap inference proceed as in Section %
\ref{sec:asymptotics} and are omitted here for brevity. The results under
the staggered adoption design can be extended to the repeated
cross-sectional setting in a manner analogous to Section \ref{sec:rc}.

\section{Empirical study}

\label{sec:empirical}

In this section, we apply our method to study the impact of the minimum wage
on family income, utilizing the natural experiment of the 2003 minimum-wage
increase in Alaska, as examined by \cite{gunsilius2023distributional}.%
\footnote{%
Replication code for all empirical and simulation analyses is available at %
\url{https://github.com/XieHaitian/DiD_meets_SC}.} The effect of
minimum-wage policies is of significant importance for labor market dynamics
and poverty alleviation strategies, making it a key area of concern for
economists and policymakers. However, the impact is complex and not clear 
\textit{a priori}, often involving a trade-off between higher earnings for
certain low-wage workers and potential negative outcomes for others, such as
job losses, reduced work hours, or inflation pressures, which may diminish
the overall gains. The minimum wage is also the focus of \cite%
{card1994minimum}, a seminal paper often regarded as one of the most
influential studies in the DiD literature.

\subsection{Data}

\cite{gunsilius2023distributional} used a subset of the data provided by 
\cite{dube2019minimum}, which was derived from individual-level data in the
Current Population Survey (CPS). \cite{dube2019minimum} employed a two-way
(state and time) fixed effects regression, along with other specifications
that included additional controls and trends. This approach can be viewed as
a parametric specification of the parallel trends assumption in Assumption %
\ref{asm:PT}, though such a parametric form is stronger than necessary for
our purpose. In contrast, \cite{gunsilius2023distributional} proposed the
distributional synthetic control method and used it to construct a synthetic
Alaska by matching state-level quantile curves of family income. Unlike
their approach, our method allows for consistent estimation and valid
inference of the ATT, regardless of whether the parallel trends or synthetic
control assumption holds.\footnote{\cite{dube2019minimum} primarily focused
on the poverty-reducing effects and estimated the unconditional quantile
partial effect of minimum-wage increases on family income for families at
lower income quantiles. Similarly, \cite{gunsilius2023distributional}
explored shifts across the entire family income distribution to illustrate
the distributional synthetic control method. To maintain clarity and focus
in our application, we restrict our discussion to the effect on the mean
family income.}

We provide a brief description of the data and variables used in the
analysis. In the CPS, different households are rotated in and out of the
survey at different time intervals. As a result, similar to \cite%
{gunsilius2023distributional}, we treat the dataset as repeated
cross-sectional data. The dataset covers a six-year period, from 1998 to
2003, during which the treatment group, Alaska, increased its minimum wage
from \$5.65 to \$7.15 in the final year. During this period, 33 other states
did not change their minimum wage and thus could potentially serve as
control states.

We select four control states based on the analysis in \cite%
{gunsilius2023distributional}, namely Virginia, New Hampshire, Maryland, and
Utah, which receive the largest weights in their synthetic control analysis.
We deliberately restrict attention to a small donor pool, since including
too many control states can make estimation of the synthetic control weights
less stable. For the selected donor pool, the estimated pre-treatment
donor-difference matrix is well conditioned: across covariate values, the
95\% quantile of $\kappa(\hat{M}_{\text{rc}}(x))$ is 13.07. In Section \ref%
{subsec:condition-number-experiment}, we show that the proposed method
performs well around this level of conditioning.

We define a household as an individual unit $i$ in the study. The outcome
variable $Y$ is the equivalized family income, defined as multiples of the
federal poverty threshold, as in \cite{dube2019minimum}. The grouping
variable $G$ represents the state in which the household resides. We aim to
maintain a set of covariates comparable to those in \cite{dube2019minimum},
incorporating one continuous variable (age) and two discrete variables
(education level and number of children).\footnote{%
Given that the household is the unit of analysis in our study, we exclude
person-level covariates such as gender, race, and ethnicity, which were
included in \cite{dube2019minimum}. Additionally, while %
\citet{dube2019minimum} included family size as a covariate, we omit it here
due to its high correlation with the number of children, with a correlation
coefficient of 0.75 in the dataset.} To aggregate individual education
levels at the household level, we count the number of individuals with some
college education and categorize households into three levels: 0, 1, and 2
or more. For the number of children, we similarly consider three categories:
0, 1, and 2 or more. In the tables presented later, the ``2 or more''
category is denoted simply by ``2''. These categorizations yield a
relatively balanced sample size across the nine resulting subgroups, with
the number of observations ranging from around 1,500 to 4,200. Such group
sizes are compatible with our asymptotic theory, in which the number of
units within each group grows large while the number of groups and periods
stays fixed. The application therefore shows that our method can be applied
in empirically relevant settings.

\paragraph{Compositional change diagnostics}

Assumption \ref{asm:time-invariance} requires the distribution of $(G,X)$ to
be stable over time. We discuss the empirical relevance of this
compositional requirement in the present application. A natural concern is
that the 2003 minimum-wage increase may have altered the composition of the
Alaska sample by attracting in-migrants to the state. Three pieces of
evidence suggest that this channel is unlikely to be quantitatively
consequential. First, official population estimates from %
\citet{AKDOLWD2004Population} show that net migration around the policy
change was small: in 2002--2003, net migration accounted for only 0.08\% of
Alaska's population. Second, the existing evidence on minimum wages and
internal migration does not support the view that minimum-wage increases
generate large in-migration responses; if anything, the most directly
related evidence finds reduced out-migration of near-minimum-wage workers
and little evidence of increased in-migration %
\citep{MintonWheaton2023Migration}. Third, in our CPS sample, a multinomial
logit regression of survey year on group and covariates yields a McFadden
pseudo-$R^2$ of only 0.01, indicating that $(G,X)$ has very limited
predictive content for the survey year $T$ in which a household is sampled.%
\footnote{%
McFadden's pseudo-R-squared is defined as one minus the ratio of the
maximized log-likelihood of the fitted model to that of the intercept-only
model. It is a commonly used goodness-of-fit measure for regression models
with categorical outcomes \citep{McFadden1979QuantitativeMethods}.}

\subsection{Results}

In implementing our proposed method, we estimate conditional expectations
using local-linear regression. Bandwidths are chosen by undersmoothing
data-driven, mean-squared-error (MSE)-optimal bandwidths.\footnote{%
For the nonparametric outcome regression, we first compute the MSE-optimal
bandwidth $h_{\text{cct}}$ from a local polynomial regression of the outcome
on age at the median age, using the \texttt{nprobust} package %
\citep{calonico2019nprobust}. We then set our bandwidth to $n^{1/5-1/3.5}
\times h_{\text{cct}}$. Because the MSE-optimal bandwidth is of rate $%
n^{-1/5}$, multiplying by $n^{1/5-1/3.5}$ yields an undersmoothed bandwidth
of order $n^{-1/3.5}$, which is $o(n^{-1/4})$, as required by our theory
(cf. the undersmoothing adjustment in \cite{shen2016distributional}). We
select the propensity-score bandwidth in the same way. Other reasonable
choices of the undersmoothing power close to $1/3.5$ give similar results.}
For cross-fitting, we use $L = 2$ folds.

We also implement the doubly robust DiD (DRDiD) estimator of %
\citet{sant2020doubly} and the distributional synthetic control method of %
\citet{gunsilius2023distributional}. For DRDiD, we estimate the nuisance
functions using bandwidths of order $n^{-1/5}$, corresponding to the
MSE-optimal rate for its orthogonal moment conditions. Confidence intervals
are constructed using the estimated analytic standard errors.

For distributional SC, within each education-by-children subgroup, we follow
the procedure of \citet{gunsilius2023distributional} by first estimating
synthetic control weights separately in each pre-treatment year using 1,000
equally spaced midpoint quantiles, and then averaging the resulting weight
vectors across 1998--2002. We apply these average weights to the donors'
2003 quantile functions to construct Alaska's synthetic counterfactual
distribution. The reported scalar effect is the integral of the difference
between the observed and synthetic quantile functions over the unit
interval, which is equivalently the difference in means. Inference is based
on a nonparametric bootstrap within state-by-year cells, as suggested by %
\citet{gunsilius2023distributional}. Because distributional SC does not
provide a straightforward way to flexibly incorporate the continuous
covariates used by our estimator, the age covariate is omitted from its
implementation, and the resulting comparison should therefore be interpreted
with caution.

\begin{table}[htbp!]
\centering
\small
\setlength{\tabcolsep}{3.5pt}
\resizebox{\textwidth}{!}{
\begin{tabular}{ccccccccc}
\toprule
& & &
\multicolumn{2}{c}{DiD-SC}
& \multicolumn{2}{c}{DRDiD}
& \multicolumn{2}{c}{Distributional SC} \\
\cmidrule(lr){4-5}\cmidrule(lr){6-7}\cmidrule(lr){8-9}
Education & Children & $n$
& Estimate & 95\% CI
& Estimate & 95\% CI
& Estimate & 95\% CI \\
\midrule
0 & 0 & 4097 & -0.604 & [-2.286, 1.078] & -0.825 & [-1.701, 0.052] & -0.269 & [-0.791, 0.253] \\
0 & 1 & 1459 &  1.200 & [-1.136, 3.535] &  0.860 & [-0.186, 1.905] &  0.589 & [-0.370, 1.548] \\
0 & 2 & 1866 &  0.045 & [-1.349, 1.439] &  0.292 & [-0.363, 0.947] &  0.357 & [-0.209, 0.923] \\
1 & 0 & 4207 &  0.181 & [-1.555, 1.917] & -0.500 & [-1.395, 0.396] & -0.622 & [-1.365, 0.122] \\
1 & 1 & 1678 &  0.356 & [-2.343, 3.056] &  0.748 & [-0.207, 1.704] & -0.233 & [-1.086, 0.620] \\
1 & 2 & 2055 &  0.088 & [-1.645, 1.821] & -0.271 & [-0.870, 0.327] & -0.173 & [-0.631, 0.286] \\
2 & 0 & 2567 &  0.723 & [-3.410, 4.856] & -0.625 & [-2.213, 0.963] & -1.317 & [-2.704, 0.071] \\
2 & 1 & 2068 & -0.850 & [-4.495, 2.794] &  0.303 & [-1.004, 1.610] & -0.193 & [-1.330, 0.944] \\
2 & 2 & 3418 & -1.015 & [-3.114, 1.084] & -0.454 & [-1.460, 0.552] & -0.387 & [-1.221, 0.447] \\
\bottomrule
\end{tabular}
}
\caption{Estimated ATT from the 2003 minimum-wage increase in Alaska}
\label{tab:empirical-main}
\caption*{\footnotesize
The table reports the estimated ATT on equivalized family income for nine subgroups defined by education level and number of children. Education denotes the number of individuals in the household with some college education. The confidence intervals (CI) for DiD-SC, DRDiD, and distributional SC are constructed using the multiplier bootstrap with 500 replications, the analytic standard error, and the nonparametric bootstrap with 1,000 replications, respectively.
}
\end{table}%

Table \ref{tab:empirical-main} reports the estimated ATTs and corresponding
95\% confidence intervals obtained using our proposed method, DRDiD, and
distributional SC. In line with the findings of \cite%
{gunsilius2023distributional}, our results indicate no statistically
significant immediate effect on family income following the minimum-wage
increase in Alaska. Several factors may contribute to this outcome. Firms
facing higher labor costs might respond by reducing working hours or
adjusting non-wage benefits to offset the increased wage expenses. In
competitive labor markets, the transmission of higher minimum wages to
overall family income can be limited by substitution effects, where
employers may prefer to replace low-wage workers with more experienced
employees.

The DRDiD and distributional SC confidence intervals also include zero for
all education-children subgroups, and each method produces point estimates
with the same sign as our estimator in five of the nine subgroups. Thus,
despite relying on different identifying restrictions and information sets,
all three methods lead to the same overall empirical conclusion. This
agreement provides a useful robustness check for our findings, while the
proposed method has the additional advantage of remaining valid under either
identification assumption and of flexibly incorporating the continuous age
covariate.

To assess the computational burden, we record wall-clock time on an Apple M4
machine with 32 GiB of memory. A point estimate requires a median of 0.22
seconds per subgroup, while the complete analysis of all nine subgroups with
500 bootstrap draws takes 224.2 seconds using eight parallel workers. Thus,
the computational cost remains modest in this application, suggesting that
the proposed method is practical for routine empirical use.

This computational efficiency is aided by several implementation features:
the numerical routines are compiled using just-in-time compilation,
intermediate arrays are preallocated and reused, the local-linear nuisance
estimators (including the propensity-score ratios) are computed from their
closed-form weighted least-squares solutions rather than through generic
numerical optimization, and the independent bootstrap replications are
parallelized across CPU cores.

In Supplementary Appendix \ref{sec:additional-empirical}, we present additional empirical
results, including sensitivity analyses for the bandwidth choice and
donor-pool selection, as well as diagnostic tests for Assumptions \ref%
{asm:PT} and \ref{asm:SC}.

\section{Simulation}

\label{sec:simulation}

\subsection{Data-generating processes}

\label{subsec:DGPs-sim}

We outline here how we generate the data for the simulation, leaving the
details to Supplementary Appendix \ref{sec:simulation-design-details}. We
set $\mathcal{T}=6$ and $N_G=4$. The marginal distribution of $X$ is taken
to be a discrete uniform distribution over an equally spaced grid of 101
points on $[0,1]$,\footnote{%
This ratio of the number of covariate grid points to the sample size is
similar to that in the empirical study.} and the conditional distribution $%
G|X$ is generated from a multinomial logit model estimated from the
empirical data. The treated potential outcome in the final period is set to $%
Y_{i,\mathcal{T}}(1)=Y_{i,\mathcal{T}}(0)+0.7\sin(2\pi X_i)+1$, so that the
treatment effect is heterogeneous and its variation is comparable in
magnitude to the untreated outcome trend, while $Y_{i,t}(1)=Y_{i,t}(0)$ for $%
t<\mathcal{T}$. The true ATT is 1.0119, which is different from 1 because
the conditional distribution of $X_i$ given $G=1$ is not uniform over the
grid points on $[0,1]$.

We calibrate $Y_t(0)$ of the four control groups $g=2,\ldots ,N_{G}+1$ to
the empirical outcome dynamics of Maryland, New Hampshire, Utah, and
Virginia, respectively, so that the reference donor group $N_{G}+1$ in %
\eqref{eqn:M-m1-def} is Virginia. The treated group corresponds to Alaska,
and its untreated path is constructed from the donor paths in the three
designs below. For each control group $g$, we first construct a preliminary
mean outcome path according to 
\begin{align}
m_{g,t}^{R}(x) = \alpha_g+\delta_t+\lambda(g,x)^{\prime}F_t+R_{g,t}, \qquad {%
\lambda(g,x)}=\mathcal{B}_g(1,x,x^2)^{\prime},\qquad {t=1,\ldots ,\mathcal{T}%
}.  \label{eq:simulation-reference-mean}
\end{align}
The donor group effects $\alpha_g$, time effects $\delta_t$, common factors $%
F_t\in\mathbb{R}^{2 \times 1}$, and quadratic loading coefficients $\mathcal{%
B}_g\in\mathbb{R}^{2\times3}$ are calibrated from the empirical sample.
Specifically, we take $\alpha_g$ and $\delta_t$ to be the estimated group
and time effects, fit a rank-two factor model to the outcome variation that
remains once they are removed, and regress the estimated factor loadings on
a quadratic function of the covariate separately for each donor.

The contrasts of the calibrated components against the reference donor group
are close to degenerate: the matrix they form is close to rank one at every
covariate value, so the stability condition in Theorem \ref{thm:asymp-dist-2}%
(iii) would fail. We therefore include in \eqref{eq:simulation-reference-mean%
} a deterministic rank-completion term $R_{g,t}$, which adds variation
orthogonal to the empirical factors and donor fixed-effect contrasts, with
its scale calibrated from the empirical singular values. This determines $%
R_{g,t}$ only for $t=1,\ldots ,\mathcal{T}-1$. Setting $R_{g,\mathcal{T}%
}=R_{g,\mathcal{T}-1}$ then makes the term cancel from every group's change
between the final pre-treatment and post-treatment periods, so it can
neither create nor remove a violation of Assumption \ref{asm:PT}.

Let $M^{R}(x)$ collect the preliminary mean outcome paths of donor groups $%
2,\ldots ,N_{G}$ over the pre-treatment periods, relative to the reference
donor group $N_{G}+1$, with rows and columns arranged as in %
\eqref{eqn:M-m1-def}. Write $M^{R}(x)=U(x)\func{diag}\{s_1(x),s_2(x),s_3(x)%
\}V(x)^{\prime}$ for its singular value decomposition, taken pointwise at
each covariate value $x$. We preserve the empirical singular vectors and
largest singular value and set 
\begin{align}
M(x) = U(x)\func{diag} \left\{ s_1(x),\ \max\left[s_2(x),\frac{s_1(x)}{\kappa%
}\right],\ \frac{s_1(x)}{\kappa} \right\}V(x)^{\prime},
\label{eq:simulation-conditioned-matrix}
\end{align}
where $\kappa=13.07$ is calibrated from the empirical condition number. This
imposes the same condition number at every $x$. The reference donor's path
is left unchanged, and the other donor pre-treatment means $m_{g,t}(x)$ are
constructed by adding the corresponding columns of $M(x)$ to it: 
\begin{align*}
m_{N_{G}+1,t}(x)& =m_{N_{G}+1,t}^{R}(x),\qquad t\leq \mathcal{T}, \\
m_{g,t}(x)& =m_{N_{G}+1,t}^{R}(x)+[M(x)]_{t,g},\qquad g=2,\ldots ,N_{G},\
t\leq \mathcal{T}-1,
\end{align*}
where $[M(x)]_{t,g}$ denotes the entry of $M(x)$ in row $t$ and the column
corresponding to donor group $g$. The donor contrasts $%
m_{g,t}(x)-m_{N_{G}+1,t}(x)$ are then exactly these entries, so $M(x)$ is
the pre-treatment donor-difference matrix of \eqref{eqn:M-m1-def} for the
resulting design. Between the final pre-treatment and post-treatment
periods, each donor's change relative to the reference donor's change is set
to three times that observed in the empirical data, which makes the
departure from Assumption \ref{asm:PT} in DGP1 below more pronounced.

Untreated outcomes are generated as $Y_{i,t}(0)=m_{G_i,t}(X_i)+%
\varepsilon_{i,t}$, where the error terms $\varepsilon_{i,t}$ are serially
independent and, for each period $t$, are iid draws from a mean-zero normal
distribution whose variance is estimated from the empirical factor-model
residuals.

We consider three DGPs corresponding to the three cases in which at least
one of Assumptions \ref{asm:PT} and \ref{asm:SC} holds.

\paragraph{DGP1: SC only}

We define the synthetic control weights as 
\begin{align}  \label{eqn:weight-construction-sim}
w_g(x) = \frac{0.2}{N_G} + 0.8 \cdot \frac{\exp\left[(g - N_G - 1) x\right]}{%
\sum_{g^{\prime }=2}^{N_G+1} \exp\left[(g^{\prime }- N_G - 1) x\right]},
\quad g = 2, \dots, N_G+1.
\end{align}
The constant component keeps every weight bounded away from zero, while the
softmax component allows the weights to vary smoothly with $x$. The donor
paths retain the heterogeneous post-treatment changes calibrated from the
empirical donor contrasts. The untreated mean of the treatment group is
defined as $m_{1,t}(x) =\sum_{g=2}^{N_G+1}w_g(x)m_{g,t}(x)$ in every period.
Assumption \ref{asm:SC} therefore holds exactly, while Assumption \ref%
{asm:PT} fails because the four donors have different changes between the
final pre-treatment and post-treatment periods.

\paragraph{DGP2: PT only}

We modify DGP1 by first imposing the reference donor's final-period
untreated change on all donors while leaving their pre-treatment paths
unchanged. For the treatment group's path, we add a pre-treatment
perturbation that is orthogonal to the column space of the donor-difference
matrix and therefore cannot be reproduced by any affine combination of the
donors. The perturbation in the final pre-treatment period is carried
unchanged into the post-treatment period, so it cancels from the treatment
group's untreated change. Hence, the treatment group and all donors have the
same untreated change, and Assumption \ref{asm:PT} holds exactly, whereas
Assumption \ref{asm:SC} fails. The orthogonality of the perturbation also
leaves the least-squares solution $\bm{w}_{0}(x)$ of $M(x)\bm{w}%
_{0}(x)=m_{1}(x)$ equal to the designated weights $(w_{g}(x))_{g=2,\ldots
,N_{G}}$.

\paragraph{DGP3: PT and SC}

We modify DGP2 by removing the orthogonal perturbation from the
treatment-group path. Both Assumptions \ref{asm:SC} and \ref{asm:PT} are
then satisfied.

\subsection{Implementation details}

\label{subsec:simulation-implementation}

We implement the estimator using two-fold cross-fitting. Within each Monte
Carlo replication, the same fold assignment is retained across all bootstrap
draws.

All nuisance functions are estimated by local-linear regression as specified
in Section \ref{sec:asymptotics}, using the Epanechnikov kernel. The
bandwidth is $h_n=cn^{-2/7}$, where $c$ varies in the bandwidth experiment
and is fixed at $c=2.5$ in the remaining experiments. A ridge penalty of $%
10^{-6}$ is used in the local-linear regressions for numerical stability.

The synthetic control weights are estimated from the pre-treatment
conditional means by ridge-stabilized affine least squares, using the
reference donor (Virginia) as the baseline donor for taking differences, and
are then projected onto the unit simplex. The ridge penalty in the weight
estimation is also $10^{-6}$.

Inference for our proposed estimator is based on an exponential multiplier
bootstrap. In each bootstrap draw, all nuisance functions and synthetic
control weights are re-estimated using the original cross-fitting folds. We
report symmetric $95\%$ confidence intervals based on the conditional
quantile of $|\hat{\theta}^{*}-\hat{\theta}|$. All results use $500$ Monte
Carlo replications and $500$ bootstrap draws.

To provide a concrete benchmark for the effective local sample size $n
p_g(x)f_X(x)h^d$ discussed following Theorem \ref{thm:asymp-dist-2}, the smallest
group probability in our calibrated design is approximately $0.167$ over the
relevant covariate region. With $d=1$, $h=2.5n^{-2/7}$, two-fold
cross-fitting, and the Epanechnikov kernel supported on $[-1,1]$, the
expected number of observations receiving nonzero kernel weight at an
interior evaluation point for the smallest group is approximately 95 for $n=2%
{,}000$ and 157 for $n=4{,}000$.

\subsection{Bandwidth experiments}

\label{subsec:bandwidth-experiment}

We consider bandwidths of the form $h_n = c n^{-2/7},
c\in\{1,1.5,2,2.5,3,3.5\}.$ The results are reported in Table \ref%
{tab:bandwidth_coverage}. Across all three DGPs, our procedure performs as
predicted by the theory: the bias is negligible relative to the standard
deviation, the standard deviation decreases at the expected rate of $1/\sqrt{%
n}$, and coverage is close to the nominal level. The resulting confidence
intervals have reasonable lengths, and shorten as the sample size increases.
The results are also insensitive to the choice of bandwidth.

\begin{table}[ht!]
\centering
\small
\setlength{\tabcolsep}{3.5pt}
\begin{tabular}{cccccccccc}
\toprule
& \multicolumn{4}{c}{$n=2000$}
& & \multicolumn{4}{c}{$n=4000$} \\
\cmidrule(lr){2-5}\cmidrule(lr){7-10}
$c$ & Bias & SD & CI coverage & CI length
& & Bias & SD & CI coverage & CI length \\
\midrule

\multicolumn{10}{c}{\textit{Panel A: DGP1, SC only}} \\
1.0 & 0.0184 & 0.0559 & 0.988 & 0.2850
& & 0.0057 & 0.0406 & 0.970 & 0.1738 \\
1.5 & 0.0135 & 0.0546 & 0.980 & 0.2592
& & 0.0032 & 0.0399 & 0.966 & 0.1643 \\
2.0 & 0.0114 & 0.0542 & 0.978 & 0.2483
& & 0.0025 & 0.0395 & 0.958 & 0.1605 \\
2.5 & 0.0107 & 0.0541 & 0.972 & 0.2428
& & 0.0021 & 0.0392 & 0.964 & 0.1588 \\
3.0 & 0.0104 & 0.0543 & 0.972 & 0.2397
& & 0.0020 & 0.0391 & 0.954 & 0.1580 \\
3.5 & 0.0103 & 0.0545 & 0.970 & 0.2382
& & 0.0018 & 0.0391 & 0.956 & 0.1575 \\
\hline

\addlinespace
\multicolumn{10}{c}{\textit{Panel B: DGP2, PT only}} \\
1.0 & 0.0030 & 0.0399 & 0.952 & 0.1654
& & -0.0013 & 0.0279 & 0.958 & 0.1127 \\
1.5 & 0.0031 & 0.0397 & 0.950 & 0.1608
& & -0.0013 & 0.0278 & 0.954 & 0.1111 \\
2.0 & 0.0030 & 0.0394 & 0.952 & 0.1587
& & -0.0013 & 0.0277 & 0.954 & 0.1103 \\
2.5 & 0.0031 & 0.0393 & 0.952 & 0.1575
& & -0.0012 & 0.0277 & 0.952 & 0.1098 \\
3.0 & 0.0031 & 0.0392 & 0.952 & 0.1569
& & -0.0012 & 0.0277 & 0.948 & 0.1095 \\
3.5 & 0.0031 & 0.0391 & 0.952 & 0.1565
& & -0.0012 & 0.0277 & 0.950 & 0.1093 \\
\hline

\addlinespace
\multicolumn{10}{c}{\textit{Panel C: DGP3, PT and SC}} \\
1.0 & 0.0030 & 0.0400 & 0.952 & 0.1654
& & -0.0013 & 0.0279 & 0.958 & 0.1127 \\
1.5 & 0.0030 & 0.0397 & 0.950 & 0.1608
& & -0.0013 & 0.0278 & 0.954 & 0.1111 \\
2.0 & 0.0030 & 0.0394 & 0.952 & 0.1587
& & -0.0013 & 0.0278 & 0.952 & 0.1102 \\
2.5 & 0.0030 & 0.0393 & 0.952 & 0.1576
& & -0.0012 & 0.0277 & 0.950 & 0.1097 \\
3.0 & 0.0031 & 0.0392 & 0.950 & 0.1570
& & -0.0011 & 0.0277 & 0.950 & 0.1095 \\
3.5 & 0.0031 & 0.0391 & 0.952 & 0.1566
& & -0.0011 & 0.0277 & 0.950 & 0.1092 \\
\bottomrule
\end{tabular}
\caption{Bandwidth experiments}
\label{tab:bandwidth_coverage}
\caption*{\footnotesize
The table reports simulation results for bandwidth
$h_n=c\,n^{-2/7}$. Bias is the Monte Carlo mean of the estimator minus the true ATT. SD is
the empirical standard deviation of the estimator. CI denotes the 95\%
multiplier bootstrap confidence interval. CI length
is its mean length across Monte Carlo replications. Each configuration is
based on 500 Monte Carlo replications and 500 bootstrap draws per
replication.}
\end{table}%

\subsection{Condition number experiments}

\label{subsec:condition-number-experiment}

We use DGP1 and vary $\kappa$ in \eqref{eq:simulation-conditioned-matrix},
while holding all other features of the DGP fixed. This sets the condition
number of the pre-treatment donor-difference matrix to $\kappa$, through its
smallest singular value, and hence varies the strength of identification of
the synthetic control weights. We choose an increasing grid to span
well-conditioned to severely ill-conditioned designs.

\begin{table}[ht!]
\centering
\small
\setlength{\tabcolsep}{3.5pt}
\begin{tabular}{cccccccccc}
\toprule
& \multicolumn{4}{c}{$n=2000$}
& & \multicolumn{4}{c}{$n=4000$} \\
\cmidrule(lr){2-5}\cmidrule(lr){7-10}
$\kappa$ & Bias & SD & CI coverage & CI length
& & Bias & SD & CI coverage & CI length \\
\midrule
1      & 0.0054 & 0.0475 & 0.962 & 0.2113
& & 0.0040 & 0.0334 & 0.954 & 0.1379 \\
2      & 0.0054 & 0.0475 & 0.964 & 0.2116
& & 0.0040 & 0.0336 & 0.956 & 0.1379 \\
4      & 0.0056 & 0.0480 & 0.976 & 0.2140
& & 0.0041 & 0.0341 & 0.958 & 0.1392 \\
10     & 0.0074 & 0.0522 & 0.972 & 0.2303
& & 0.0048 & 0.0374 & 0.958 & 0.1505 \\
20     & 0.0152 & 0.0633 & 0.970 & 0.2766
& & 0.0092 & 0.0460 & 0.966 & 0.1843 \\
40     & 0.0415 & 0.0814 & 0.958 & 0.3706
& & 0.0288 & 0.0630 & 0.942 & 0.2667 \\
80     & 0.0866 & 0.0948 & 0.902 & 0.4358
& & 0.0728 & 0.0817 & 0.884 & 0.3491 \\
160    & 0.1148 & 0.0990 & 0.872 & 0.4557
& & 0.1107 & 0.0909 & 0.804 & 0.3850 \\
\bottomrule
\end{tabular}
\caption{Condition number experiments}
\label{tab:condition_coverage}
\caption*{\footnotesize
The table reports results from DGP1 with a varying $\kappa$ in \eqref{eq:simulation-conditioned-matrix}. Bias is the Monte Carlo mean of the estimator minus the true ATT. SD is the empirical standard deviation of
the estimator. CI denotes the 95\% multiplier bootstrap confidence interval. CI length is its
mean length across Monte Carlo replications. Each configuration is based on 500 Monte Carlo replications and 500 bootstrap draws per replication.}
\end{table}%

The results are reported in Table \ref{tab:condition_coverage}. When the
matrix $M$ is well conditioned, the bias is small relative to the sampling
variation and the confidence intervals attain approximately nominal
coverage. As $\kappa$ increases, the synthetic control weights become less
stably identified: the bias, standard deviation, and confidence-interval
length all increase, and coverage begins to deteriorate once the condition
number is substantially larger than the empirical benchmark. Importantly,
the empirical condition number $\kappa=13.07$ lies in the well-behaved
region of the experiment. These results provide finite-sample reassurance
that the degree of ill-conditioning observed in our empirical application is
not sufficiently severe to undermine the performance of the proposed
procedure.

\subsection{Local misspecification experiments}

\label{subsec:local-misspecification}

The double robustness result applies to the union of the models defined by
Assumptions \ref{asm:PT} and \ref{asm:SC}. If both assumptions fail, the
proposed moment function is generally biased for the ATT, and the estimator
may converge to a pseudo-true value instead. We study local departures whose
magnitudes are of the same order as the estimator's sampling error.
Specifically, starting from DGP1 and DGP2, we shift only the treatment
group's post-treatment untreated mean by $a/\sqrt{n}$ times the Monte Carlo
standard deviation of $\sqrt{n}(\widehat{\theta}-\theta)$ in the
corresponding baseline DGP. At $a=0$, one identifying assumption holds
exactly while the other fails; for $a\neq0$, both assumptions fail. The $%
n^{-1/2}$ rate matches the estimator's sampling-error scale, and the
normalization gives $a$ an approximately standardized interpretation.

\begin{table}[!ht]
\centering
\resizebox{\textwidth}{!}{
\begin{tabular}{*{9}{c}}
\toprule
& \multicolumn{4}{c}{$n=2000$}
& \multicolumn{4}{c}{$n=4000$} \\
\cmidrule(lr){2-5}\cmidrule(lr){6-9}
$a$
& Bias & SD & CI coverage & CI length
& Bias & SD & CI coverage & CI length \\
\midrule
\multicolumn{9}{c}{\textit{Panel A: Local misspecification with DGP1}} \\
$-2.0$ & $-0.1057$ & 0.0537 & 0.602 & 0.2428 & $-0.0730$ & 0.0381 & 0.570 & 0.1597 \\
$-1.5$ & $-0.0774$ & 0.0537 & 0.788 & 0.2429 & $-0.0530$ & 0.0382 & 0.746 & 0.1597 \\
$-1.0$ & $-0.0490$ & 0.0537 & 0.918 & 0.2429 & $-0.0329$ & 0.0382 & 0.900 & 0.1597 \\
$-0.5$ & $-0.0207$ & 0.0537 & 0.962 & 0.2430 & $-0.0129$ & 0.0382 & 0.952 & 0.1598 \\
$ 0.0$ & $ 0.0076$ & 0.0537 & 0.972 & 0.2431 & $ 0.0071$ & 0.0382 & 0.950 & 0.1598 \\
$ 0.5$ & $ 0.0360$ & 0.0537 & 0.944 & 0.2432 & $ 0.0272$ & 0.0382 & 0.906 & 0.1599 \\
$ 1.0$ & $ 0.0643$ & 0.0537 & 0.858 & 0.2432 & $ 0.0472$ & 0.0382 & 0.814 & 0.1599 \\
$ 1.5$ & $ 0.0927$ & 0.0537 & 0.724 & 0.2433 & $ 0.0672$ & 0.0382 & 0.648 & 0.1599 \\
$ 2.0$ & $ 0.1210$ & 0.0537 & 0.504 & 0.2434 & $ 0.0872$ & 0.0382 & 0.430 & 0.1600 \\
\hline 

\addlinespace
\multicolumn{9}{c}{\textit{Panel B: Local misspecification with DGP2}} \\
$-2.0$ & $-0.0762$ & 0.0395 & 0.496 & 0.1579 & $-0.0542$ & 0.0286 & 0.514 & 0.1093 \\
$-1.5$ & $-0.0567$ & 0.0395 & 0.698 & 0.1580 & $-0.0403$ & 0.0286 & 0.694 & 0.1094 \\
$-1.0$ & $-0.0371$ & 0.0395 & 0.850 & 0.1581 & $-0.0265$ & 0.0286 & 0.828 & 0.1095 \\
$-0.5$ & $-0.0175$ & 0.0395 & 0.940 & 0.1582 & $-0.0127$ & 0.0287 & 0.922 & 0.1095 \\
$ 0.0$ & $ 0.0020$ & 0.0396 & 0.958 & 0.1583 & $ 0.0011$ & 0.0287 & 0.946 & 0.1096 \\
$ 0.5$ & $ 0.0216$ & 0.0396 & 0.924 & 0.1584 & $ 0.0150$ & 0.0287 & 0.904 & 0.1096 \\
$ 1.0$ & $ 0.0411$ & 0.0396 & 0.814 & 0.1585 & $ 0.0288$ & 0.0287 & 0.818 & 0.1097 \\
$ 1.5$ & $ 0.0607$ & 0.0397 & 0.694 & 0.1586 & $ 0.0426$ & 0.0287 & 0.674 & 0.1097 \\
$ 2.0$ & $ 0.0803$ & 0.0397 & 0.500 & 0.1588 & $ 0.0564$ & 0.0287 & 0.478 & 0.1098 \\
\bottomrule
\end{tabular}
}
\caption{Local misspecification experiments}
\label{tab:local_misspecification_experiments}
\caption*{\footnotesize The table reports the estimator's performance under local misspecification of the identification assumptions, with the degree of misspecification indexed by $a$. Bias is the Monte Carlo mean of the estimator minus the true ATT. SD is the empirical standard deviation of the estimator. CI denotes the 95\% multiplier bootstrap confidence interval. CI length is its
mean length across Monte Carlo replications. Each configuration is based on 500 Monte Carlo replications and 500 bootstrap draws per replication.}
\end{table}%

The results are reported in Table \ref%
{tab:local_misspecification_experiments}. At $a=0$, the estimator has
negligible bias and the confidence intervals attain approximately nominal
coverage under either identifying assumption. As $|a|$ increases, the bias
grows smoothly while the standard deviation and confidence-interval length
remain nearly unchanged. Coverage consequently declines in a systematic
manner, although it remains reasonably close to nominal under moderate local
departures. These patterns are consistent with the theory and demonstrate
that the proposed procedure performs reliably when either identifying
assumption holds, while degrading predictably when both assumptions are
locally violated.

\subsection{Analytic confidence intervals}

\label{subsec:analytic-sim}

We compare the analytic confidence intervals with the bootstrap intervals in
the condition number experiment in Section \ref%
{subsec:condition-number-experiment}. The analytic interval under the
synthetic control condition is based on the feasible influence function $%
\widehat{\psi}_{\mathrm{SC}}=\widehat{\psi}_{\mathrm{PT}}+\widehat{\psi}_{1}+%
\widehat{\psi}_{2}$ from Theorem \ref{thm:asymp-dist-2}, with the nuisance
functions estimated by cross-fitted local-linear regression using bandwidth $%
n^{-1/5}$. The analytic standard error is the sample standard deviation of $%
\widehat{\psi}_{\mathrm{SC}}$ divided by $\sqrt{n}$, and the 95\% confidence
interval is constructed using the normal critical value.

The results are reported in Table \ref{tab:analytic-CI}. The analytic
interval fails in opposite directions at the two ends of the range. In
well-conditioned cases, the analytic confidence interval exhibits some
undercoverage. As the condition number increases, the analytic standard
error becomes severely inflated, with its ratio to the Monte Carlo
standard deviation rising well above one and producing overly conservative
and long confidence intervals. In contrast, the bootstrap standard error
remains much more stable across the condition-number range. The analytic
standard error captures the error in estimating $(M(x)^{\prime }M(x))^{-1}$
through a first-order approximation, which deteriorates as $M(x)$ approaches
singularity, whereas the bootstrap method does not rely on that
approximation. This is also consistent with the classical bootstrap
literature, which shows that, for nonlinear statistics, bootstrap procedures
can provide more accurate inference than first-order asymptotic methods %
\citep{hall1988symmetric,hall1988theoretical}.

\begin{table}[!ht]
\centering
\small
\setlength{\tabcolsep}{3.5pt}
\begin{tabular}{cccccccccc}
\toprule
& \multicolumn{4}{c}{$n=2000$}
& & \multicolumn{4}{c}{$n=4000$} \\
\cmidrule(lr){2-5}\cmidrule(lr){7-10}
& & & \multicolumn{2}{c}{SE/SD}
& & & & \multicolumn{2}{c}{SE/SD} \\
\cmidrule(lr){4-5}\cmidrule(lr){9-10}
$\kappa$
& CI coverage & CI length & Analytic & Bootstrap
& & CI coverage & CI length & Analytic & Bootstrap \\
\midrule
1   & 0.939 & 0.174 & 0.957  & 1.137
& & 0.930 & 0.122 & 0.930 & 1.054 \\
2   & 0.937 & 0.174 & 0.958  & 1.138
& & 0.934 & 0.122 & 0.928 & 1.050 \\
4   & 0.937 & 0.176 & 0.961  & 1.137
& & 0.934 & 0.124 & 0.926 & 1.043 \\
10  & 0.931 & 0.196 & 0.975  & 1.120
& & 0.936 & 0.137 & 0.934 & 1.028 \\
20  & 0.953 & 0.335 & 1.331  & 1.080
& & 0.940 & 0.180 & 0.988 & 1.003 \\
40  & 0.992 & 0.688 & 1.853  & 1.061
& & 0.984 & 0.376 & 1.409 & 0.992 \\
80  & 1.000 & 5.723 & 10.986 & 1.049
& & 1.000 & 1.200 & 3.032 & 0.967 \\
160 & 1.000 & 4.223 & 6.698  & 1.052
& & 1.000 & 4.124 & 7.121 & 0.948 \\
\bottomrule
\end{tabular}
\caption{Analytic confidence intervals}
\label{tab:analytic-CI}
\caption*{\footnotesize
The table reports results from DGP1 with varying condition number $\kappa$ in
\eqref{eq:simulation-conditioned-matrix}. CI denotes the 95\% analytic
confidence interval based on the SC influence function. CI length is its
mean length across Monte Carlo replications. SE/SD denotes the ratio of
the mean estimated standard error to the Monte Carlo standard deviation,
where the standard error is computed using either the analytic SC variance
formula or the multiplier bootstrap. Each configuration is based on 500
Monte Carlo replications and 500 bootstrap draws per replication.
}
\end{table}%

\subsection{Comparison with other panel data causal methods}

\label{subsec:comparison-other-methods}

We first compare DiD-SC with the doubly robust DiD estimator (DRDiD)
proposed by \citet{sant2020doubly}. We construct a family of DGPs starting
from DGP3 of Section \ref{subsec:DGPs-sim}, in which both Assumptions \ref%
{asm:PT} and \ref{asm:SC} hold. For a local parameter $a$, we perturb the
post-treatment untreated conditional means of the control groups along the
empirical direction of their post-treatment changes. The size of the
perturbation is $a/\sqrt{n}$ times the asymptotic standard deviation of $%
\sqrt{n}(\hat{\theta}_{\text{DRDiD}}-\theta)$ at $a=0$, where $\hat{\theta}_{%
\text{DRDiD}} $ is the DRDiD estimator. We then reconstruct the
treated-group untreated conditional mean using the exact synthetic control
relation. Consequently, Assumption \ref{asm:SC} continues to hold exactly
for every $a$, whereas Assumption \ref{asm:PT} holds at $a=0$ and is
violated at the $n^{-1/2}$ rate when $a\neq0$.

The construction of the DRDiD estimator follows \citet{sant2020doubly}, with
the nuisance estimation procedure matched to that of our proposed estimator
as closely as possible. The main difference is that, because DRDiD is based
on an orthogonal moment, we estimate its nuisance functions using the
MSE-optimal bandwidth of order $n^{-1/5}$. Its confidence intervals are
constructed using the analytic standard error. The results are reported in
Table \ref{tab:drdid}. At $a=0$, DRDiD is slightly more efficient than
DiD-SC, but as $|a|$ increases, it exhibits substantial bias and
deteriorating coverage, while DiD-SC remains robust.

\begin{table}[ht!]
\centering
\small
\setlength{\tabcolsep}{3.5pt}
\begin{tabular}{cccccccccc}
\toprule
& \multicolumn{4}{c}{$n=2000$}
& & \multicolumn{4}{c}{$n=4000$} \\
\cmidrule(lr){2-5}\cmidrule(lr){7-10}
$a$ & Bias & SD & CI coverage & CI length
& & Bias & SD & CI coverage & CI length \\
\midrule
\multicolumn{10}{c}{\textit{Panel A: DRDiD}} \\
$-2.0$ & -0.0725 & 0.0408 & 0.572 & 0.1605
& & -0.0542 & 0.0287 & 0.506 & 0.1101 \\
$-1.5$ & -0.0536 & 0.0398 & 0.730 & 0.1562
& & -0.0410 & 0.0281 & 0.666 & 0.1084 \\
$-1.0$ & -0.0348 & 0.0389 & 0.862 & 0.1528
& & -0.0278 & 0.0277 & 0.836 & 0.1070 \\
$-0.5$ & -0.0160 & 0.0383 & 0.930 & 0.1501
& & -0.0146 & 0.0273 & 0.918 & 0.1059 \\
$0.0$  & 0.0029 & 0.0379 & 0.948 & 0.1484
& & -0.0014 & 0.0271 & 0.958 & 0.1051 \\
$0.5$  & 0.0217 & 0.0377 & 0.904 & 0.1476
& & 0.0119 & 0.0269 & 0.926 & 0.1046 \\
$1.0$  & 0.0405 & 0.0377 & 0.814 & 0.1477
& & 0.0251 & 0.0268 & 0.822 & 0.1044 \\
$1.5$  & 0.0594 & 0.0380 & 0.654 & 0.1488
& & 0.0383 & 0.0268 & 0.698 & 0.1046 \\
$2.0$  & 0.0782 & 0.0385 & 0.460 & 0.1507
& & 0.0515 & 0.0270 & 0.510 & 0.1051 \\
\hline

\addlinespace
\multicolumn{10}{c}{\textit{Panel B: DiD-SC}} \\
$-2.0$ & 0.0055 & 0.0418 & 0.964 & 0.1749
& & -0.0004 & 0.0285 & 0.954 & 0.1139 \\
$-1.5$ & 0.0049 & 0.0404 & 0.962 & 0.1670
& & -0.0006 & 0.0280 & 0.952 & 0.1118 \\
$-1.0$ & 0.0043 & 0.0395 & 0.962 & 0.1614
& & -0.0008 & 0.0278 & 0.950 & 0.1105 \\
$-0.5$ & 0.0036 & 0.0391 & 0.952 & 0.1582
& & -0.0010 & 0.0277 & 0.952 & 0.1099 \\
$0.0$  & 0.0030 & 0.0392 & 0.950 & 0.1577
& & -0.0012 & 0.0278 & 0.952 & 0.1099 \\
$0.5$  & 0.0024 & 0.0398 & 0.950 & 0.1595
& & -0.0014 & 0.0281 & 0.948 & 0.1106 \\
$1.0$  & 0.0018 & 0.0408 & 0.956 & 0.1640
& & -0.0017 & 0.0285 & 0.952 & 0.1122 \\
$1.5$  & 0.0012 & 0.0423 & 0.956 & 0.1706
& & -0.0019 & 0.0292 & 0.952 & 0.1144 \\
$2.0$  & 0.0006 & 0.0441 & 0.958 & 0.1797
& & -0.0021 & 0.0300 & 0.956 & 0.1172 \\
\bottomrule
\end{tabular}
\caption{Comparison with the DRDiD estimator}
\label{tab:drdid}
\caption*{\footnotesize
The table compares DRDiD with DiD-SC under a family of DGPs that satisfies Assumption \ref{asm:SC} but locally violates Assumption \ref{asm:PT} at the $a/\sqrt{n}$ rate. Bias is the Monte Carlo mean of the estimator minus the true ATT. SD is the empirical standard deviation. The confidence intervals (CI) for DiD-SC and DRDiD are constructed using a multiplier bootstrap with 500 replications and the analytic standard error, respectively. CI length is the mean interval length across Monte Carlo replications. Each configuration is based on 500 Monte Carlo replications.}
\end{table}%

We next compare DiD-SC with conventional synthetic control (SC) and
synthetic difference-in-differences (SDiD) \citep{arkhangelsky2021synthetic}%
. We start from DGP3 again and perturb it to strengthen the dependence on $X$
while maintaining both identification assumptions. Let $q(x)$ denote the
function $\sin(2\pi x)+0.5\sin(4\pi x)$, centered and standardized over the
treated group's covariate distribution. For a parameter $a$, we perturb the
weights in \eqref{eqn:weight-construction-sim} by adding to the exponent a
term proportional to $a(g-3.5)q(x)$. We then shift the untreated change
common to all donors by a term proportional to $aq(x)$, and reconstruct the
treated-group untreated conditional means using the perturbed weights.
Consequently, both Assumptions \ref{asm:PT} and \ref{asm:SC} continue to
hold for every $a$. The case $a=0$ coincides with DGP3, whereas larger
values of $a$ make the identifying relationships depend more strongly and
nonlinearly on $X$.

Because both SC and SDiD operate at the aggregate level and do not flexibly
incorporate individual-level covariates, we aggregate the simulated data to
the state level and omit the covariate information. Table \ref%
{tab:aggregate-sc} reports the results. We focus on the bias and standard
deviation of the point estimators and do not compare inference, since
standard inferential procedures for SC and SDiD typically rely on a
sufficiently rich time dimension or donor pool.\footnote{%
Conventional SC inference can be based on permutation tests across donor
units. With only four donors in our design, there are five possible
assignments of the treated unit, so the smallest attainable permutation $p$%
-value is $1/5=0.2$, making conventional permutation-based inference
uninformative in our setting.}

The results show that when $a=0$, all three methods perform similarly, with
small biases and comparable standard deviations. As $a$ increases, however,
both SC and SDiD exhibit increasingly large biases and greater variability,
whereas DiD-SC remains nearly unbiased with little change in its standard
deviation. These results illustrate the information loss from relying only
on state-level average outcomes when the identifying relationships depend
strongly on individual-level covariates.

\begin{table}[ht!]
\centering
\small
\setlength{\tabcolsep}{3.5pt}
\begin{tabular}{cccccccccccccc}
\toprule
& \multicolumn{6}{c}{$n=2000$}
& & \multicolumn{6}{c}{$n=4000$} \\
\cmidrule(lr){2-7}\cmidrule(lr){9-14}
& \multicolumn{2}{c}{SC}
& \multicolumn{2}{c}{SDiD}
& \multicolumn{2}{c}{DiD-SC}
& & \multicolumn{2}{c}{SC}
& \multicolumn{2}{c}{SDiD}
& \multicolumn{2}{c}{DiD-SC} \\
\cmidrule(lr){2-3}\cmidrule(lr){4-5}\cmidrule(lr){6-7}
\cmidrule(lr){9-10}\cmidrule(lr){11-12}\cmidrule(lr){13-14}
$a$
& Bias & SD
& Bias & SD
& Bias & SD
& & Bias & SD
& Bias & SD
& Bias & SD \\
\midrule
0.0
& 0.0056 & 0.0475
& 0.0065 & 0.0444
& 0.0030 & 0.0393
& & 0.0051 & 0.0358
& 0.0071 & 0.0334
& -0.0012 & 0.0277 \\
0.5
& 0.0107 & 0.0544
& 0.0121 & 0.0568
& 0.0025 & 0.0396
& & 0.0096 & 0.0405
& 0.0126 & 0.0429
& -0.0013 & 0.0279 \\
1.0
& 0.0167 & 0.0636
& 0.0183 & 0.0685
& 0.0020 & 0.0404
& & 0.0150 & 0.0467
& 0.0183 & 0.0506
& -0.0013 & 0.0281 \\
1.5
& 0.0247 & 0.0742
& 0.0241 & 0.0806
& 0.0014 & 0.0416
& & 0.0223 & 0.0540
& 0.0235 & 0.0587
& -0.0014 & 0.0286 \\
\bottomrule
\end{tabular}
\caption{Comparison with synthetic control estimators based on aggregate-level data}
\label{tab:aggregate-sc}
\caption*{\footnotesize
The table compares conventional synthetic control (SC), synthetic difference-in-differences (SDiD), and DiD-SC across a family of DGPs with varying strength of the dependence of the synthetic control weights on $X$, indexed by $a$. Bias is the Monte Carlo mean of the estimator minus the true ATT. SD is the empirical standard deviation of the
estimator. Each configuration is based on 500 Monte Carlo replications.}
\end{table}%

\section{Conclusion}

\label{sec:conclusion}

This paper introduces a doubly robust methodology that integrates DiD and
synthetic control approaches, enabling identification, consistent
semiparametric estimation, and valid bootstrap inference for the ATT under
either parallel trends or synthetic control. The method applies to panel
data, repeated cross-sectional data, and staggered treatment designs, and is
particularly well suited to micro-level datasets in which rich within-group
information can be exploited. The proposed procedure is especially useful
when the credibility of parallel trends is uncertain but a stable synthetic
control representation remains plausible, providing additional protection
against parallel trends violations.

\section*{Acknowledgment}

The authors contributed equally to this work and are listed in alphabetical
order. We thank Co-Editor Bo Honor\'e, an Associate Editor, and two
anonymous referees for their constructive comments and suggestions, which
substantially improved the article. We gratefully acknowledge Yawei Wang for
outstanding research assistance. Haitian Xie acknowledges financial support
from the National Natural Science Foundation of China (No. 72403008, No.
72495123, No. 72572003, No. 72621002). Part of this work was completed while
Yuhang Zhang was a Ph.D. student at Peking University and a postdoctoral
researcher at The Chinese University of Hong Kong.

\bigskip

\renewcommand{\baselinestretch}{1.0} \selectfont

\bibliographystyle{chicago}
\bibliography{DiD_SC}

\begin{thebibliography}{}

\bibitem[\protect\citeauthoryear{Abadie}{Abadie}{2005}]{abadie2005semiparametric}
Abadie, A. (2005).
\newblock Semiparametric difference-in-differences estimators.
\newblock {\em Review of Economic Studies\/}~{\em 72\/}(1), 1--19.

\bibitem[\protect\citeauthoryear{Abadie}{Abadie}{2021}]{abadie2021using}
Abadie, A. (2021).
\newblock Using synthetic controls: Feasibility, data requirements, and
  methodological aspects.
\newblock {\em Journal of Economic Literature\/}~{\em 59\/}(2), 391--425.

\bibitem[\protect\citeauthoryear{Abadie, Diamond, and Hainmueller}{Abadie
  et~al.}{2010}]{abadie2010synthetic}
Abadie, A., A.~Diamond, and J.~Hainmueller (2010).
\newblock Synthetic control methods for comparative case studies: Estimating
  the effect of {C}alifornia’s tobacco control program.
\newblock {\em Journal of the American Statistical Association\/}~{\em
  105\/}(490), 493--505.

\bibitem[\protect\citeauthoryear{Abadie, Diamond, and Hainmueller}{Abadie
  et~al.}{2015}]{abadie2015comparative}
Abadie, A., A.~Diamond, and J.~Hainmueller (2015).
\newblock Comparative politics and the synthetic control method.
\newblock {\em American Journal of Political Science\/}~{\em 59\/}(2),
  495--510.

\bibitem[\protect\citeauthoryear{Abadie and Gardeazabal}{Abadie and
  Gardeazabal}{2003}]{abadie2003economic}
Abadie, A. and J.~Gardeazabal (2003).
\newblock The economic costs of conflict: A case study of the {B}asque
  {C}ountry.
\newblock {\em American Economic Review\/}~{\em 93\/}(1), 113--132.

\bibitem[\protect\citeauthoryear{Acemoglu and Angrist}{Acemoglu and
  Angrist}{2001}]{acemoglu2001consequences}
Acemoglu, D. and J.~D. Angrist (2001).
\newblock Consequences of employment protection? {T}he case of the {A}mericans
  with {D}isabilities {A}ct.
\newblock {\em Journal of Political Economy\/}~{\em 109\/}(5), 915--957.

\bibitem[\protect\citeauthoryear{{Alaska Department of Labor and Workforce
  Development}}{{Alaska Department of Labor and Workforce
  Development}}{2004}]{AKDOLWD2004Population}
{Alaska Department of Labor and Workforce Development} (2004, January).
\newblock Labor department estimates {A}laska's 2003 population.
\newblock News Release No. 04-33.
\newblock Research and Analysis, Demographics Unit.

\bibitem[\protect\citeauthoryear{Angrist and Pischke}{Angrist and
  Pischke}{2010}]{AngristPischke2010_credibility}
Angrist, J.~D. and J.-S. Pischke (2010).
\newblock The credibility revolution in empirical economics: How better
  research design is taking the con out of econometrics.
\newblock {\em Journal of Economic Perspectives\/}~{\em 24\/}(2), 3--30.

\bibitem[\protect\citeauthoryear{Arkhangelsky, Athey, Hirshberg, Imbens, and
  Wager}{Arkhangelsky et~al.}{2021}]{arkhangelsky2021synthetic}
Arkhangelsky, D., S.~Athey, D.~A. Hirshberg, G.~W. Imbens, and S.~Wager (2021).
\newblock Synthetic difference-in-differences.
\newblock {\em American Economic Review\/}~{\em 111\/}(12), 4088--4118.

\bibitem[\protect\citeauthoryear{Arkhangelsky and Imbens}{Arkhangelsky and
  Imbens}{2022}]{arkhangelsky2022doubly}
Arkhangelsky, D. and G.~W. Imbens (2022).
\newblock Doubly robust identification for causal panel data models.
\newblock {\em Econometrics Journal\/}~{\em 25\/}(3), 649--674.

\bibitem[\protect\citeauthoryear{Arkhangelsky and Samkov}{Arkhangelsky and
  Samkov}{2024}]{arkhangelsky2024sequential}
Arkhangelsky, D. and A.~Samkov (2024).
\newblock Sequential synthetic difference in differences.
\newblock {\em arXiv preprint arXiv:2404.00164\/}.

\bibitem[\protect\citeauthoryear{Athey, Bayati, Doudchenko, Imbens, and
  Khosravi}{Athey et~al.}{2021}]{athey2021matrix}
Athey, S., M.~Bayati, N.~Doudchenko, G.~Imbens, and K.~Khosravi (2021).
\newblock Matrix completion methods for causal panel data models.
\newblock {\em Journal of the American Statistical Association\/}~{\em
  116\/}(536), 1716--1730.

\bibitem[\protect\citeauthoryear{Athey and Imbens}{Athey and
  Imbens}{2006}]{athey2006identification}
Athey, S. and G.~W. Imbens (2006).
\newblock Identification and inference in nonlinear difference-in-differences
  models.
\newblock {\em Econometrica\/}~{\em 74\/}(2), 431--497.

\bibitem[\protect\citeauthoryear{Athey and Imbens}{Athey and
  Imbens}{2017}]{athey2017state}
Athey, S. and G.~W. Imbens (2017).
\newblock The state of applied econometrics: Causality and policy evaluation.
\newblock {\em Journal of Economic Perspectives\/}~{\em 31\/}(2), 3--32.

\bibitem[\protect\citeauthoryear{Baker, Callaway, Cunningham, Goodman-Bacon,
  and Sant'Anna}{Baker et~al.}{2026}]{baker2026practitioner}
Baker, A., B.~Callaway, S.~Cunningham, A.~Goodman-Bacon, and P.~H.~C. Sant'Anna
  (2026).
\newblock Difference-in-differences designs: A practitioner's guide.
\newblock {\em Journal of Economic Literature\/}~{\em 64\/}(2), 498--557.

\bibitem[\protect\citeauthoryear{Ben-Michael, Feller, and
  Rothstein}{Ben-Michael et~al.}{2021}]{ben2021augmented}
Ben-Michael, E., A.~Feller, and J.~Rothstein (2021).
\newblock The augmented synthetic control method.
\newblock {\em Journal of the American Statistical Association\/}~{\em
  116\/}(536), 1789--1803.

\bibitem[\protect\citeauthoryear{Ben-Michael, Feller, and
  Rothstein}{Ben-Michael et~al.}{2022}]{ben2022synthetic}
Ben-Michael, E., A.~Feller, and J.~Rothstein (2022).
\newblock Synthetic controls with staggered adoption.
\newblock {\em Journal of the Royal Statistical Society Series B: Statistical
  Methodology\/}~{\em 84\/}(2), 351--381.

\bibitem[\protect\citeauthoryear{Callaway, Dyal, Sant'Anna, and
  Tsyawo}{Callaway et~al.}{2025}]{callaway2025beyond}
Callaway, B., D.~Dyal, P.~H.~C. Sant'Anna, and E.~S. Tsyawo (2025).
\newblock Beyond parallel trends: An identification-strategy-robust approach to
  causal inference with panel data.
\newblock {\em arXiv preprint arXiv:2511.21977\/}.

\bibitem[\protect\citeauthoryear{Callaway and Sant’Anna}{Callaway and
  Sant’Anna}{2021}]{callaway2021difference}
Callaway, B. and P.~H. Sant’Anna (2021).
\newblock Difference-in-differences with multiple time periods.
\newblock {\em Journal of Econometrics\/}~{\em 225\/}(2), 200--230.

\bibitem[\protect\citeauthoryear{Calonico, Cattaneo, and Farrell}{Calonico
  et~al.}{2019}]{calonico2019nprobust}
Calonico, S., M.~D. Cattaneo, and M.~H. Farrell (2019).
\newblock nprobust: Nonparametric kernel-based estimation and robust
  bias-corrected inference.
\newblock {\em Journal of Statistical Software\/}~{\em 91}, 1--33.

\bibitem[\protect\citeauthoryear{Card and Krueger}{Card and
  Krueger}{1994}]{card1994minimum}
Card, D. and A.~B. Krueger (1994).
\newblock Minimum wages and employment: A case study of the fast-food industry
  in {N}ew {J}ersey and {P}ennsylvania.
\newblock {\em American Economic Review\/}~{\em 84\/}(4), 772--793.

\bibitem[\protect\citeauthoryear{Chen and Feng}{Chen and
  Feng}{2023}]{chen2023group}
Chen, S. and J.~Feng (2023).
\newblock Group-heterogeneous changes-in-changes and distributional synthetic
  controls.
\newblock {\em arXiv preprint arXiv:2307.15313\/}.

\bibitem[\protect\citeauthoryear{Chen, Linton, and Van~Keilegom}{Chen
  et~al.}{2003}]{chen2003estimation}
Chen, X., O.~Linton, and I.~Van~Keilegom (2003).
\newblock Estimation of semiparametric models when the criterion function is
  not smooth.
\newblock {\em Econometrica\/}~{\em 71\/}(5), 1591--1608.

\bibitem[\protect\citeauthoryear{Chen, Sant'Anna, and Xie}{Chen
  et~al.}{2025}]{chen2025efficient}
Chen, X., P.~H.~C. Sant'Anna, and H.~Xie (2025).
\newblock Efficient difference-in-differences and event study estimators.
\newblock arXiv:2506.17729.

\bibitem[\protect\citeauthoryear{Chen}{Chen}{2020}]{chen2020distributional}
Chen, Y.-T. (2020).
\newblock A distributional synthetic control method for policy evaluation.
\newblock {\em Journal of Applied Econometrics\/}~{\em 35\/}(5), 505--525.

\bibitem[\protect\citeauthoryear{Cheng and Huang}{Cheng and
  Huang}{2010}]{cheng2010bootstrap}
Cheng, G. and J.~Z. Huang (2010).
\newblock Bootstrap consistency for general semiparametric m-estimation.
\newblock {\em Annals of Statistics\/}~{\em 38\/}(5), 2884--2915.

\bibitem[\protect\citeauthoryear{Chernozhukov, Chetverikov, Demirer, Duflo,
  Hansen, Newey, and Robins}{Chernozhukov et~al.}{2018}]{chernozhukov2018DML}
Chernozhukov, V., D.~Chetverikov, M.~Demirer, E.~Duflo, C.~Hansen, W.~Newey,
  and J.~Robins (2018).
\newblock Double/debiased machine learning for treatment and structural
  parameters.
\newblock {\em Econometrics Journal\/}~{\em 21\/}(1), C1--C68.

\bibitem[\protect\citeauthoryear{Cohen and Einav}{Cohen and
  Einav}{2003}]{cohen2003effects}
Cohen, A. and L.~Einav (2003).
\newblock The effects of mandatory seat belt laws on driving behavior and
  traffic fatalities.
\newblock {\em Review of Economics and Statistics\/}~{\em 85\/}(4), 828--843.

\bibitem[\protect\citeauthoryear{Corak}{Corak}{2001}]{corak2001death}
Corak, M. (2001).
\newblock Death and divorce: The long-term consequences of parental loss on
  adolescents.
\newblock {\em Journal of Labor Economics\/}~{\em 19\/}(3), 682--715.

\bibitem[\protect\citeauthoryear{Ding and Li}{Ding and
  Li}{2019}]{ding2019bracketing}
Ding, P. and F.~Li (2019).
\newblock A bracketing relationship between difference-in-differences and
  lagged-dependent-variable adjustment.
\newblock {\em Political Analysis\/}~{\em 27\/}(4), 605--615.

\bibitem[\protect\citeauthoryear{Dube}{Dube}{2019}]{dube2019minimum}
Dube, A. (2019).
\newblock Minimum wages and the distribution of family incomes.
\newblock {\em American Economic Journal: Applied Economics\/}~{\em 11\/}(4),
  268--304.

\bibitem[\protect\citeauthoryear{Fan and Gijbels}{Fan and
  Gijbels}{1996}]{fan1996local}
Fan, J. and I.~Gijbels (1996).
\newblock {\em Local Polynomial Modelling and Its Applications}, Volume~66 of
  {\em Monographs on Statistics and Applied Probability}.
\newblock London: Chapman \& Hall.

\bibitem[\protect\citeauthoryear{Finkelstein}{Finkelstein}{2002}]{finkelstein2002effect}
Finkelstein, A. (2002).
\newblock The effect of tax subsidies to employer-provided supplementary health
  insurance: {E}vidence from {C}anada.
\newblock {\em Journal of Public Economics\/}~{\em 84\/}(3), 305--339.

\bibitem[\protect\citeauthoryear{Ghanem, Sant'Anna, and W{\"u}thrich}{Ghanem
  et~al.}{2026}]{ghanem2026pretrends}
Ghanem, D., P.~H.~C. Sant'Anna, and K.~W{\"u}thrich (2026).
\newblock When should pre-trends be parallel?
\newblock {\em AEA Papers and Proceedings\/}~{\em 116}, 64--69.

\bibitem[\protect\citeauthoryear{Goldsmith-Pinkham}{Goldsmith-Pinkham}{2024}]{goldsmith2024tracking}
Goldsmith-Pinkham, P. (2024).
\newblock Tracking the credibility revolution across fields.
\newblock {\em arXiv preprint arXiv:2405.20604\/}.

\bibitem[\protect\citeauthoryear{Gunsilius}{Gunsilius}{2023}]{gunsilius2023distributional}
Gunsilius, F.~F. (2023).
\newblock Distributional synthetic controls.
\newblock {\em Econometrica\/}~{\em 91\/}(3), 1105--1117.

\bibitem[\protect\citeauthoryear{Hall}{Hall}{1988a}]{hall1988symmetric}
Hall, P. (1988a).
\newblock On symmetric bootstrap confidence intervals.
\newblock {\em Journal of the Royal Statistical Society Series B: Statistical
  Methodology\/}~{\em 50\/}(1), 35--45.

\bibitem[\protect\citeauthoryear{Hall}{Hall}{1988b}]{hall1988theoretical}
Hall, P. (1988b).
\newblock Theoretical comparison of bootstrap confidence intervals.
\newblock {\em The Annals of Statistics\/}~{\em 16\/}(3), 927--953.

\bibitem[\protect\citeauthoryear{Heckman, Ichimura, and Todd}{Heckman
  et~al.}{1997}]{heckman1997matching}
Heckman, J.~J., H.~Ichimura, and P.~E. Todd (1997).
\newblock Matching as an econometric evaluation estimator: Evidence from
  evaluating a job training programme.
\newblock {\em Review of Economic Studies\/}~{\em 64\/}(4), 605--654.

\bibitem[\protect\citeauthoryear{Holm}{Holm}{1979}]{holm1979simple}
Holm, S. (1979).
\newblock A simple sequentially rejective multiple test procedure.
\newblock {\em Scandinavian Journal of Statistics\/}~{\em 6\/}(2), 65--70.

\bibitem[\protect\citeauthoryear{Hsiao, Steve~Ching, and Ki~Wan}{Hsiao
  et~al.}{2012}]{hsiao2012panel}
Hsiao, C., H.~Steve~Ching, and S.~Ki~Wan (2012).
\newblock A panel data approach for program evaluation: Measuring the benefits
  of political and economic integration of {H}ong {K}ong with {M}ainland
  {C}hina.
\newblock {\em Journal of Applied Econometrics\/}~{\em 27\/}(5), 705--740.

\bibitem[\protect\citeauthoryear{Kellogg, Mogstad, Pouliot, and
  Torgovitsky}{Kellogg et~al.}{2021}]{kellogg2021combining}
Kellogg, M., M.~Mogstad, G.~A. Pouliot, and A.~Torgovitsky (2021).
\newblock Combining matching and synthetic control to tradeoff biases from
  extrapolation and interpolation.
\newblock {\em Journal of the American Statistical Association\/}~{\em
  116\/}(536), 1804--1816.

\bibitem[\protect\citeauthoryear{Kennedy}{Kennedy}{2024}]{kennedy2024semiparametric}
Kennedy, E.~H. (2024).
\newblock Semiparametric doubly robust targeted double machine learning: A
  review.
\newblock In E.~Laber, B.~Chakraborty, E.~Moodie, T.~Cai, and M.~van~der Laan
  (Eds.), {\em Handbook of Statistical Methods for Precision Medicine},
  Chapter~10, pp.\  215--232. Taylor \& Francis.

\bibitem[\protect\citeauthoryear{Kennedy, Balakrishnan, and G’Sell}{Kennedy
  et~al.}{2020}]{kennedy2020sharp}
Kennedy, E.~H., S.~Balakrishnan, and M.~G’Sell (2020).
\newblock {Sharp instruments for classifying compliers and generalizing causal
  effects}.
\newblock {\em Annals of Statistics\/}~{\em 48\/}(4), 2008 -- 2030.

\bibitem[\protect\citeauthoryear{Kong, Linton, and Xia}{Kong
  et~al.}{2010}]{kong2010uniform}
Kong, E., O.~Linton, and Y.~Xia (2010).
\newblock Uniform {B}ahadur representation for local polynomial estimates of
  m-regression and its application to the additive model.
\newblock {\em Econometric Theory\/}~{\em 26\/}(5), 1529--1564.

\bibitem[\protect\citeauthoryear{Lewbel, Choi, and Zhou}{Lewbel
  et~al.}{2023}]{lewbel2023over}
Lewbel, A., J.~Y. Choi, and Z.~Zhou (2023).
\newblock Over-identified doubly robust identification and estimation.
\newblock {\em Journal of Econometrics\/}~{\em 235\/}(1), 25--42.

\bibitem[\protect\citeauthoryear{Masry}{Masry}{1996}]{masry1996multivariate}
Masry, E. (1996).
\newblock Multivariate local polynomial regression for time series: Uniform
  strong consistency and rates.
\newblock {\em Journal of Time Series Analysis\/}~{\em 17\/}(6), 571--599.

\bibitem[\protect\citeauthoryear{McFadden}{McFadden}{1979}]{McFadden1979QuantitativeMethods}
McFadden, D. (1979).
\newblock Quantitative methods for analysing travel behaviour of individuals:
  Some recent developments.
\newblock In D.~A. Hensher and P.~R. Stopher (Eds.), {\em Behavioural Travel
  Modelling}, pp.\  279--318. London: Croom Helm.

\bibitem[\protect\citeauthoryear{Meyer, Viscusi, and Durbin}{Meyer
  et~al.}{1995}]{meyer1995workers}
Meyer, B.~D., W.~K. Viscusi, and D.~L. Durbin (1995).
\newblock Workers' compensation and injury duration: Evidence from a natural
  experiment.
\newblock {\em American Economic Review\/}~{\em 85\/}(3), 322--340.

\bibitem[\protect\citeauthoryear{Minton and Wheaton}{Minton and
  Wheaton}{2023}]{MintonWheaton2023Migration}
Minton, R. and B.~Wheaton (2023).
\newblock Minimum wages and internal migration.
\newblock Working paper.

\bibitem[\protect\citeauthoryear{Newey}{Newey}{1994a}]{newey1994asymptotic}
Newey, W.~K. (1994a).
\newblock The asymptotic variance of semiparametric estimators.
\newblock {\em Econometrica\/}~{\em 62\/}(6), 1349--1382.

\bibitem[\protect\citeauthoryear{Newey}{Newey}{1994b}]{newey1994kernel}
Newey, W.~K. (1994b).
\newblock Kernel estimation of partial means and a general variance estimator.
\newblock {\em Econometric Theory\/}~{\em 10\/}(2), 233--253.

\bibitem[\protect\citeauthoryear{Robbins, Saunders, and Kilmer}{Robbins
  et~al.}{2017}]{robbins2017framework}
Robbins, M.~W., J.~Saunders, and B.~Kilmer (2017).
\newblock A framework for synthetic control methods with high-dimensional,
  micro-level data: Evaluating a neighborhood-specific crime intervention.
\newblock {\em Journal of the American Statistical Association\/}~{\em
  112\/}(517), 109--126.

\bibitem[\protect\citeauthoryear{Roth, Sant'Anna, Bilinski, and Poe}{Roth
  et~al.}{2023}]{roth2023whats}
Roth, J., P.~H.~C. Sant'Anna, A.~Bilinski, and J.~Poe (2023).
\newblock What's trending in difference-in-differences? {A} synthesis of the
  recent econometrics literature.
\newblock {\em Journal of Econometrics\/}~{\em 235\/}(2), 2218--2244.

\bibitem[\protect\citeauthoryear{Sant'Anna and Xu}{Sant'Anna and
  Xu}{2026}]{santanna2026compositional}
Sant'Anna, P. H.~C. and Q.~Xu (2026).
\newblock Difference-in-differences with compositional changes.
\newblock {\em Journal of Econometrics\/}~{\em 253}, 106147.

\bibitem[\protect\citeauthoryear{Sant’Anna and Zhao}{Sant’Anna and
  Zhao}{2020}]{sant2020doubly}
Sant’Anna, P.~H. and J.~Zhao (2020).
\newblock Doubly robust difference-in-differences estimators.
\newblock {\em Journal of Econometrics\/}~{\em 219\/}(1), 101--122.

\bibitem[\protect\citeauthoryear{Shen, Ding, Sekhon, and Yu}{Shen
  et~al.}{2023}]{shen2023same}
Shen, D., P.~Ding, J.~Sekhon, and B.~Yu (2023).
\newblock Same root different leaves: Time series and cross-sectional methods
  in panel data.
\newblock {\em Econometrica\/}~{\em 91\/}(6), 2125--2154.

\bibitem[\protect\citeauthoryear{Shen and Zhang}{Shen and
  Zhang}{2016}]{shen2016distributional}
Shen, S. and X.~Zhang (2016).
\newblock Distributional tests for regression discontinuity: Theory and
  empirical examples.
\newblock {\em Review of Economics and Statistics\/}~{\em 98\/}(4), 685--700.

\bibitem[\protect\citeauthoryear{van~der Vaart}{van~der
  Vaart}{1998}]{vanderVaart1998}
van~der Vaart, A.~W. (1998).
\newblock {\em Asymptotic Statistics}.
\newblock Cambridge: Cambridge University Press.

\end{thebibliography}

\pagebreak \appendix

\numberwithin{equation}{section} \numberwithin{lemma}{section} %
\numberwithin{table}{section}

\renewcommand{\thepage}{A.\arabic{page}} \setcounter{page}{1}

\section{Appendix: proofs of the main theoretical results}

\label{sec:main-proofs}

This appendix presents the proofs of the theoretical results in Sections \ref%
{sec:setup} and \ref{sec:asymptotics}. Proofs of the results in the other
sections are provided in the Supplementary Appendix. \vspace{0.5cm}

\begin{proof}[Proof of Theorem \protect\ref{thm:identification}]
Part (i): When $m_{\Delta }$ is correctly specified, we have, under
Assumption \ref{asm:PT}: 
\begin{eqnarray*}
&&\mathbb{E}\left[ \mathcal{G}_{g}\left( \Delta Y-m_{\Delta }\left( X\right)
\right) |X\right] \\
&=&\mathbb{E}\left[ \Delta Y-m_{\Delta }\left( X\right) |X,\mathcal{G}_{g}=1%
\right] \mathbb{P} \left( \mathcal{G}_{g}=1|X\right) \\
&=&\left( \mathbb{E}\left[ Y_{\mathcal{T}}-Y_{\mathcal{T}-1}|X,\mathcal{G}%
_{g}=1\right] -m_{\Delta }\left( X\right) \right) \mathbb{P} \left( \mathcal{%
G}_{g}=1|X\right) \\
&=&\left( \mathbb{E}\left[ Y_{\mathcal{T}}-Y_{\mathcal{T}-1}|X,\mathcal{G}%
_{g}=1\right] -\mathbb{E}\left[ Y_{\mathcal{T}}-Y_{\mathcal{T}-1}|X, G \neq 1%
\right] \right) \mathbb{P} \left( \mathcal{G}_{g}=1|X\right) \\
&=&0,
\end{eqnarray*}%
for $g=2,\ldots ,N_{G}+1,$ and so%
\begin{equation*}
\mathbb{E}\left[ \phi \left( S;m_{\Delta },\tilde{p},\tilde{w};\pi
_{1}\right) \right] =\mathbb{E}\left[ \phi \left( S;m_{\Delta },p,\tilde{w}%
;\pi _{1}\right) \right] .
\end{equation*}

When $p$ is correctly specified and $\sum_{g=2}^{N_{G}+1}\tilde{w}_{g}\left(
X\right) =1$, we have%
\begin{equation}
\mathbb{E}\left[ \left. \mathcal{G}_{1}-\sum_{g=2}^{N_{G}+1}\tilde{w}%
_{g}\left( X\right) \frac{p_{1}\left( X\right) }{p_{g}\left( X\right) }%
\mathcal{G}_{g}\right\vert X\right] =0,  \label{conditional mean zero(p)}
\end{equation}%
and so 
\begin{equation*}
\mathbb{E}\left[ \phi \left( S;\tilde{m}_{\Delta },p,\tilde{w};\pi
_{1}\right) \right] =\mathbb{E}\left[ \phi \left( S;m_{\Delta },p,\tilde{w}%
;\pi _{1}\right) \right] .
\end{equation*}

It then suffices to show that 
\begin{equation*}
\mathbb{E}\left[ \phi \left( S;m_{\Delta },p,\tilde{w};\pi _{1}\right) %
\right] =\theta .
\end{equation*}%
We have 
\begin{eqnarray*}
&&\mathbb{E}\left[ \phi \left( S;m_{\Delta },p,\tilde{w};\pi _{1}\right) %
\right] \\
&=&\mathbb{E}\left[ \frac{\mathcal{G}_{1}}{\pi _{1}}\left( \Delta
Y-m_{\Delta }\left( X\right) \right) \right] =\mathbb{E}\left[ \frac{%
\mathcal{G}_{1}}{\pi _{1}}\Delta Y\right] -\mathbb{E}\left[ \frac{\mathcal{G}%
_{1}}{\pi _{1}}m_{\Delta }\left( X\right) \right] \\
&=&\mathbb{E}\left[ \Delta Y|G=1\right] -\mathbb{E}\left[ m_{\Delta }\left(
X\right) |G=1\right] \\
&=&\mathbb{E}\left[ Y_{\mathcal{T}}\left( 1\right) -Y_{\mathcal{T}-1}\left(
0\right) |G=1\right] -\mathbb{E}\Big[ \underset{\equiv m_{\Delta }\left(
X\right) }{\underbrace{\mathbb{E}\left( Y_{\mathcal{T}}\left( 0\right) -Y_{%
\mathcal{T}-1}\left( 0\right) |X,G\neq 1\right) }}|G=1\Big] \\
&=&\mathbb{E}\left[ Y_{\mathcal{T}}\left( 1\right) -Y_{\mathcal{T}-1}\left(
0\right) |G=1\right] -\mathbb{E}\left[ \mathbb{E}\left( Y_{\mathcal{T}%
}\left( 0\right) -Y_{\mathcal{T}-1}\left( 0\right) |X,G=1\right) |G=1\right]
\\
&=&\mathbb{E}\left[ Y_{\mathcal{T}}\left( 1\right) -Y_{\mathcal{T}-1}\left(
0\right) |G=1\right] -\mathbb{E}\left[ Y_{\mathcal{T}}\left( 0\right) -Y_{%
\mathcal{T}-1}\left( 0\right) |G=1\right] \\
&=&\mathbb{E}\left[ Y_{\mathcal{T}}\left( 1\right) -Y_{\mathcal{T}}\left(
0\right) |G=1\right] =\theta ,
\end{eqnarray*}%
as desired, where the fourth equality uses Assumption \ref{asm:PT}.

Part (ii): When $p$ is correctly specified and $\sum_{g=2}^{N_{G}+1}w_{g}%
\left( X\right) =1$, we use (\ref{conditional mean zero(p)}) to obtain:%
\begin{eqnarray*}
&&\mathbb{E}\left[ \phi \left( S;\tilde{m}_{\Delta },p,w;\pi _{1}\right) %
\right] \\
&=&\frac{1}{\pi _{1}}\mathbb{E}\left[ \left( \mathcal{G}_{1}-%
\sum_{g=2}^{N_{G}+1}w_{g}\left( X\right) \frac{p_{1}\left( X\right) }{%
p_{g}\left( X\right) }\mathcal{G}_{g}\right) \Delta Y\right] \\
&=&\frac{1}{\pi _{1}}\mathbb{E}\left( \mathcal{G}_{1}\Delta Y\right) -\frac{1%
}{\pi _{1}}\mathbb{E}\sum_{g=2}^{N_{G}+1}w_{g}\left( X\right) \frac{%
p_{1}\left( X\right) }{p_{g}\left( X\right) }\mathbb{E}\left( \mathcal{G}%
_{g}\Delta Y|X\right) \\
&=&\frac{1}{\pi _{1}}\mathbb{E}\left( \mathcal{G}_{1}\Delta Y\right) -\frac{1%
}{\pi _{1}}\mathbb{E}\sum_{g=2}^{N_{G}+1}w_{g}\left( X\right) p_{1}\left(
X\right) \mathbb{E}\left( Y_{\mathcal{T}}\left( 0\right) -Y_{\mathcal{T}%
-1}\left( 0\right) |X,G=g\right) \\
&=&\frac{1}{\pi _{1}}\mathbb{E}\left( \mathcal{G}_{1}\left[ Y_{\mathcal{T}%
}\left( 1\right) -Y_{\mathcal{T}-1}\left( 0\right) \right] \right) -\frac{1}{%
\pi _{1}}\mathbb{E}\left[ p_{1}\left( X\right) \mathbb{E}\left( Y_{\mathcal{T%
}}\left( 0\right) -Y_{\mathcal{T}-1}\left( 0\right) |X,G=1\right) \right] \\
&=&\frac{1}{\pi _{1}}\mathbb{E}\left( \mathcal{G}_{1}\left[ Y_{\mathcal{T}%
}\left( 1\right) -Y_{\mathcal{T}-1}\left( 0\right) \right] \right) -\frac{1}{%
\pi _{1}}\mathbb{E}\left[ \mathbb{E}\left[ \mathcal{G}_{1}(Y_{\mathcal{T}%
}\left( 0\right) -Y_{\mathcal{T}-1}\left( 0\right) )|X\right] \right] \\
&=&\frac{1}{\pi _{1}}\mathbb{E}\left( \mathcal{G}_{1}\left[ Y_{\mathcal{T}%
}\left( 1\right) -Y_{\mathcal{T}-1}\left( 0\right) \right] \right) -\frac{1}{%
\pi _{1}}\mathbb{E}\left[ \mathcal{G}_{1}\left[ Y_{\mathcal{T}}\left(
0\right) -Y_{\mathcal{T}-1}\left( 0\right) \right] \right] \\
&=&\frac{1}{\pi _{1}}\mathbb{E}\left( \mathcal{G}_{1}\left[ Y_{\mathcal{T}%
}\left( 1\right) -Y_{\mathcal{T}}\left( 0\right) \right] \right) =\mathbb{E}%
\left( \left[ Y_{\mathcal{T}}\left( 1\right) -Y_{\mathcal{T}}\left( 0\right) %
\right] |G=1\right) =\theta ,
\end{eqnarray*}%
where the fourth equality uses Assumption \ref{asm:SC}.
\end{proof}

\vspace{0.5cm}

\begin{proof}[Proof of Theorem \protect\ref{thm:asymp-dist-1}]
In the proof, we use the notation $w$ rather than $\omega $ to represent the
pseudo-true weights. For simplicity of the exposition, we assume that the
nuisance estimators $\hat{m}_{\Delta }$, $\hat{p}$, and $\hat{w}$ are
constructed from another independent sample, which is essentially achieved
by the cross-fitting method.\footnote{%
See, for example, \citet{kennedy2020sharp} and %
\citet{kennedy2024semiparametric} for similar proofs related to double
machine learning estimators.} The estimator will be simplified to 
\begin{equation*}
\hat{\theta}=\frac{1}{n}\sum_{i=1}^{n}\phi (S_{i};\hat{m}_{\Delta },\hat{p},%
\hat{w};\hat{\pi}_{1}).
\end{equation*}%
Define the following infeasible estimator 
\begin{equation*}
\tilde{\theta}=\frac{1}{n}\sum_{i=1}^{n}\phi (S_{i};m_{\Delta },p,w;\hat{\pi}%
_{1}),
\end{equation*}%
where the nuisance estimators $\hat{m}_{\Delta },\hat{p},$ and $\hat{w}$ are
replaced by the true $m_{\Delta }$ and $p$ and the pseudo-true $w$, but $%
\hat{\pi}_{1}$ remains estimated. Our goal is to establish the first-order
equivalence between $\hat{\theta}$ and $\tilde{\theta}$, where it suffices
to only examine the numerators while omitting the common denominator $\hat{%
\pi}_{1}$. The difference between the numerators of $\phi (S;\hat{m}_{\Delta
},\hat{p},\hat{w};\hat{\pi}_{1})$ and $\phi (S;m_{\Delta },p,w;\hat{\pi}_{1})
$ can be decomposed into the following five terms 
\begin{equation*}
\phi (S;\hat{m}_{\Delta },\hat{p},\hat{w};\hat{\pi}_{1})\hat{\pi}_{1}-\phi
(S;m_{\Delta },p,w;\hat{\pi}_{1})\hat{\pi}_{1}=E_{1}+E_{2}+E_{3}+E_{4}+E_{5},
\end{equation*}%
where 
\begin{align}
E_{1}(S;\hat{m}_{\Delta },\hat{p},\hat{w})& \equiv
\sum_{g=2}^{N_{G}+1}E_{1g}(S;\hat{m}_{\Delta },\hat{p},\hat{w})\equiv
-\sum_{g=2}^{N_{G}+1}\hat{w}_{g}(X)\left( \frac{\hat{p}_{1}(X)}{\hat{p}%
_{g}(X)}-\frac{p_{1}(X)}{p_{g}(X)}\right) \mathcal{G}_{g}(\Delta Y-m_{\Delta
}(X)),  \label{eqn:E1} \\
E_{2}(S;\hat{m}_{\Delta },\hat{p},\hat{w})& \equiv
\sum_{g=2}^{N_{G}+1}E_{2g}(S;\hat{m}_{\Delta },\hat{p},\hat{w})\equiv
-\sum_{g=2}^{N_{G}+1}(\hat{w}_{g}(X)-w_{g}(X))\frac{p_{1}(X)}{p_{g}(X)}%
\mathcal{G}_{g}(\Delta Y-m_{\Delta }(X)),  \label{eqn:E2} \\
E_{3}(S;\hat{m}_{\Delta },\hat{p},\hat{w})& \equiv -\left( \mathcal{G}%
_{1}-\sum_{g=2}^{N_{G}+1}w_{g}(X)\mathcal{G}_{g}\frac{p_{1}(X)}{p_{g}(X)}%
\right) (\hat{m}_{\Delta }(X)-m_{\Delta }(X)),  \notag \\
E_{4}(S;\hat{m}_{\Delta },\hat{p},\hat{w})& \equiv
\sum_{g=2}^{N_{G}+1}E_{4g}(S;\hat{m}_{\Delta },\hat{p},\hat{w})\equiv
\sum_{g=2}^{N_{G}+1}\mathcal{G}_{g}\frac{p_{1}(X)}{p_{g}(X)}(\hat{w}%
_{g}(X)-w_{g}(X))(\hat{m}_{\Delta }(X)-m_{\Delta }(X)),  \notag \\
E_{5}(S;\hat{m}_{\Delta },\hat{p},\hat{w})& \equiv
\sum_{g=2}^{N_{G}+1}E_{5g}(S;\hat{m}_{\Delta },\hat{p},\hat{w})\equiv
\sum_{g=2}^{N_{G}+1}\mathcal{G}_{g}\left( \frac{\hat{p}_{1}(X)}{\hat{p}%
_{g}(X)}-\frac{p_{1}(X)}{p_{g}(X)}\right) \hat{w}_{g}(X)(\hat{m}_{\Delta
}(X)-m_{\Delta }(X)).  \notag
\end{align}

Both $E_{1}$ and $E_{2}$ have zero mean because $\mathbb{E}[\mathcal{G}%
_{g}(\Delta Y-m_{\Delta }(X))|X]=0,g\geq 2,$ under Assumption \ref{asm:PT}.
The term $E_{3}$ has a zero mean because 
\begin{equation*}
\mathbb{E}\left[ \sum_{g=2}^{N_{G}+1}w_{g}(X)\mathcal{G}_{g}\frac{p_{1}(X)}{%
p_{g}(X)}\big|X\right] =p_{1}(X)\sum_{g=2}^{N_{G}+1}w_{g}(X)=\mathbb{E}[%
\mathcal{G}_{1}|X].
\end{equation*}%
The term $E_{4}$ also has a zero mean because 
\begin{align*}
& \sum_{g=2}^{N_{G}+1}\mathbb{E}\left[ \mathcal{G}_{g}\frac{p_{1}(X)}{%
p_{g}(X)}(\hat{w}_{g}(X)-w_{g}(X))(\hat{m}_{\Delta }(X)-m_{\Delta }(X))|X,%
\hat{m}_{\Delta },\hat{w}_{g}\right] \\
& =p_{1}(X)(\hat{m}_{\Delta }(X)-m_{\Delta }(X))\left( \sum_{g=2}^{N_{G}+1}(%
\hat{w}_{g}(X)-w_{g}(X))\right) \\
& =p_{1}(X)(\hat{m}_{\Delta }(X)-m_{\Delta }(X))(1-1)=0,
\end{align*}%
using $\sum_{g=2}^{N_{G}+1}\hat{w}_{g}(X)=1$ and $%
\sum_{g=2}^{N_{G}+1}w_{g}(X)=1.$ In the above, the expectation is taken with
respect to the distribution of $\mathcal{G}_{g}$, conditioned on $X$, while
holding the functions $\hat{m}_{\Delta }$ and $\hat{w}_{g}$ fixed so that
the expectation is not taken with respect to the estimation errors in these
two functions. We adopt the same convention in the rest of the proofs.

To summarize this decomposition, the first four terms are sample averages
of terms with zero mean and shrinking variance, and the fifth
term is a product of the first-stage estimation errors. Based on this
decomposition, we can write the difference between $\hat{\theta}$ and $%
\tilde{\theta}$ as follows: 
\begin{align*}
& \hat{\theta}-\tilde{\theta}=\frac{1}{\hat{\pi}_{1}}\left(
\sum_{g=2}^{N_{G}+1}\sum_{j=1,2,5}\hat{E}_{jg}+\hat{E}_{3}+\hat{E}%
_{4}\right) , \\
& \text{ with }\hat{E}_{jg}=\frac{1}{n}\sum_{i=1}^{n}E_{jg}(S_{i};\hat{m}%
_{\Delta },\hat{p},\hat{w}),\hat{E}_{j}=\frac{1}{n}\sum_{i=1}^{n}E_{j}(S_{i};%
\hat{m}_{\Delta },\hat{p},\hat{w}),\ j=3,4.
\end{align*}%
The goal is to show that all these terms are of order $o_{p}(n^{-1/2})$.
Since $N_{G}$ is finite, for $j=1,2,5$ it suffices for us to focus on a
single group $g$. We first examine $\hat{E}_{1g}$. Notice that, conditional
on $(\hat{m}_{\Delta },\hat{p},\hat{w})$, the sequence $(E_{1g}(S_{i};\hat{m}%
_{\Delta },\hat{p},\hat{w}))_{i=1}^{n}$ is iid because $(\hat{m}_{\Delta },%
\hat{p},\hat{w})%
\indep%
(S_{i})_{i=1}^{n}$. Hence, for $i\neq i^{\prime }$, we have 
\begin{align*}
& \mathbb{E}\left[ E_{1g}(S_{i};\hat{m}_{\Delta },\hat{p},\hat{w}%
)E_{1g}(S_{i^{\prime }};\hat{m}_{\Delta },\hat{p},\hat{w})|\hat{m}_{\Delta },%
\hat{p},\hat{w}\right] \\
& =\mathbb{E}\left[ E_{1g}(S_{i};\hat{m}_{\Delta },\hat{p},\hat{w})|\hat{m}%
_{\Delta },\hat{p},\hat{w}\right] \mathbb{E}\left[ E_{1g}(S_{i^{\prime }};%
\hat{m}_{\Delta },\hat{p},\hat{w})|\hat{m}_{\Delta },\hat{p},\hat{w}\right]
=0,
\end{align*}%
where the last line follows from the fact that $\mathbb{E}\left[
E_{1g}(S_{i};\hat{m}_{\Delta },\hat{p},\hat{w})|\hat{m}_{\Delta },\hat{p},%
\hat{w}\right] =0$. Then the conditional second moment of $\hat{E}_{1g}$
will only consist of the diagonal terms: 
\begin{align*}
\mathbb{E}\left[ (\hat{E}_{1g})^{2}|\hat{m}_{\Delta },\hat{p},\hat{w}\right]
& =\frac{1}{n^{2}}\sum_{i,i^{\prime }=1}^{n}\mathbb{E}\left[ E_{1g}(S_{i};%
\hat{m}_{\Delta },\hat{p},\hat{w})E_{1g}(S_{i^{\prime }};\hat{m}_{\Delta },%
\hat{p},\hat{w})|\hat{m}_{\Delta },\hat{p},\hat{w}\right] \\
& =\frac{1}{n^{2}}\sum_{i=1}^{n}\mathbb{E}\left[ E_{1g}(S_{i};\hat{m}%
_{\Delta },\hat{p},\hat{w})^{2}|\hat{m}_{\Delta },\hat{p},\hat{w}\right] \\
& =\frac{1}{n}\mathbb{E}\left[ E_{1g}(S_{i};\hat{m}_{\Delta },\hat{p},\hat{w}%
)^{2}|\hat{m}_{\Delta },\hat{p},\hat{w}\right] .
\end{align*}%
The conditional second moment of $E_{1g}(S_{i};\hat{m}_{\Delta },\hat{p},%
\hat{w})$ is bounded as 
\begin{align*}
& \mathbb{E}\left[ E_{1g}(S_{i};\hat{m}_{\Delta },\hat{p},\hat{w})^{2}|\hat{m%
}_{\Delta },\hat{p},\hat{w}\right] \\
& =\mathbb{E}\left[ p_{g}(X)\left( \frac{\hat{p}_{1}(X)}{\hat{p}_{g}(X)}-%
\frac{p_{1}(X)}{p_{g}(X)}\right) ^{2}\hat{w}_{g}(X)^{2}\mathbb{E}\left[
(\Delta Y-m_{\Delta }(X))^{2}|X,G=g\right] \Big|\hat{m}_{\Delta },\hat{p},%
\hat{w}\right] \\
& \leq C\lVert \hat{w}_{g}\rVert _{\infty }^{2}\left\Vert \hat{p}_{1}/\hat{p}%
_{g}-p_{1}/p_{g}\right\Vert _{L_{2}(F_{X})}^{2},
\end{align*}%
under the assumption that $\mathbb{E}\left[ (\Delta Y-m_{\Delta
}(X))^{2}|X,G=g\right] $ is bounded. By Markov's inequality, we have for any 
$t>0$, 
\begin{align*}
& \mathbb{P}\left( \frac{\sqrt{n}|\hat{E}_{1g}|}{\lVert \hat{w}_{g}\rVert
_{\infty }\left\Vert \hat{p}_{1}/\hat{p}_{g}-p_{1}/p_{g}\right\Vert
_{L_{2}(F_{X})}}>t\right) \\
& =\,\mathbb{E}\left[ \mathbb{P}\left( \frac{\sqrt{n}|\hat{E}_{1g}|}{\lVert 
\hat{w}_{g}\rVert _{\infty }\left\Vert \hat{p}_{1}/\hat{p}%
_{g}-p_{1}/p_{g}\right\Vert _{L_{2}(F_{X})}}>t\Bigg|\hat{m}_{\Delta },\hat{p}%
,\hat{w}\right) \right] \\
& \leq \,\mathbb{E}\left[ \frac{\mathbb{E}\left[ n|\hat{E}_{1g}|^{2}\text{ }|%
\text{ }\hat{m}_{\Delta },\hat{p},\hat{w}\right] }{\lVert \hat{w}_{g}\rVert
_{\infty }^{2}\left\Vert \hat{p}_{1}/\hat{p}_{g}-p_{1}/p_{g}\right\Vert
_{L_{2}(F_{X})}^{2}t^{2}}\right] \leq C/t^{2}.
\end{align*}%
This proves that 
\begin{equation*}
\hat{E}_{1g}=O_{p}\left( \frac{1}{\sqrt{n}}\lVert \hat{w}_{g}\rVert _{\infty
}\left( \left\Vert \hat{p}_{1}/\hat{p}_{g}-p_{1}/p_{g}\right\Vert
_{L_{2}(F_{X})}\right) \right) =o_{p}(n^{-1/2}).
\end{equation*}%
Similarly, we can show that $\hat{E}_{2g}=o_{p}(n^{-1/2})$ under the
additional condition that $\left\Vert \hat{w}_{g}-w_{g}\right\Vert
_{L_{2}(F_{X})}=o_{p}(1)$. For the term $\hat{E}_{3}$, notice that the
conditional second moment of $E_{3}$ is bounded as 
\begin{align*}
& \mathbb{E}\left[ E_{3}(S_{i};\hat{m}_{\Delta },\hat{p},\hat{w})^{2}|\hat{m}%
_{\Delta },\hat{p},\hat{w}\right] \\
& =\mathbb{E}\left[ \left( \mathcal{G}_{1}-\sum_{g=2}^{N_{G}+1}w_{g}(X)%
\mathcal{G}_{g}\frac{p_{1}(X)}{p_{g}(X)}\right) ^{2}(\hat{m}_{\Delta
}(X)-m_{\Delta }(X))^{2}\Big|\hat{m}_{\Delta },\hat{p},\hat{w}\right] \\
& \leq \,C\left\Vert \hat{m}_{\Delta }-m_{\Delta }\right\Vert
_{L_{2}(F_{X})}^{2},
\end{align*}%
under the assumption that $1/p_{g}$ and $w_{g}$ are bounded functions. Then,
following the same steps as above, we can show that 
\begin{equation*}
\hat{E}_{3}=O_{p}\left( \frac{1}{\sqrt{n}}\left( \left\Vert \hat{m}_{\Delta
}-m_{\Delta }\right\Vert _{L_{2}(F_{X})}\right) \right) =o_{p}(n^{-1/2}).
\end{equation*}%
The same reasoning applies to $\hat{E}_{4}$,  except that the weight difference is bounded in the sup
norm because two estimation errors now appear in the same second moment: 
\begin{equation*}
\hat{E}_{4}=O_{p}\left( \max_{g}\lVert \hat{w}_{g}-w_{g}\rVert _{\infty
}\left\Vert \hat{m}_{\Delta }-m_{\Delta }\right\Vert _{L_{2}(F_{X})}/\sqrt{n}%
\right) =o_{p}(n^{-1/2}),
\end{equation*}%
where $\max_{g}\lVert \hat{w}_{g}-w_{g}\rVert _{\infty }=O_{p}(1)$ because $%
\lVert \hat{w}_{g}\rVert _{\infty }=O_{p}(1)$ and $w_{g}$ is bounded. For
the term $\hat{E}_{5g}$, by the triangle inequality, we have 
\begin{align*}
\mathbb{E}\left[ |\hat{E}_{5g}||\hat{m}_{\Delta },\hat{p},\hat{w}\right] &
\leq \frac{1}{n}\sum_{i=1}^{n}\mathbb{E}\left[ |E_{5g}(S_{i};\hat{m}_{\Delta
},\hat{p},\hat{w})||\hat{m}_{\Delta },\hat{p},\hat{w}\right] \\
& =\mathbb{E}\left[ |E_{5g}(S_{i};\hat{m}_{\Delta },\hat{p},\hat{w})||\hat{m}%
_{\Delta },\hat{p},\hat{w}\right] .
\end{align*}%
We can proceed to prove that 
\begin{align*}
& \mathbb{E}\left[ |E_{5g}(S_{i};\hat{m}_{\Delta },\hat{p},\hat{w})||\hat{m}%
_{\Delta },\hat{p},\hat{w}\right] \\
\leq \,& \mathbb{E}\left[ \left\vert \hat{w}_{g}(X)\right\vert \left\vert {(}%
\hat{p}_{1}(X)/\hat{p}_{g}(X)-p_{1}(X)/p_{g}(X))(\hat{m}_{\Delta
}(X)-m_{\Delta }(X))\right\vert |\hat{m}_{\Delta },\hat{p},\hat{w}\right] \\
\leq \,& \lVert \hat{w}_{g}\rVert _{\infty }\left\Vert \hat{p}_{1}/\hat{p}%
_{g}-p_{1}/p_{g}\right\Vert _{L_{2}(F_{X})}\left\Vert \hat{m}_{\Delta
}-m_{\Delta }\right\Vert _{L_{2}(F_{X})},
\end{align*}%
based on the Cauchy--Schwarz inequality. Using Markov's inequality again, we
obtain that 
\begin{equation*}
\hat{E}_{5g}=O_{p}\left( \lVert \hat{w}_{g}\rVert _{\infty }\left\Vert \hat{p%
}_{1}/\hat{p}_{g}-p_{1}/p_{g}\right\Vert _{L_{2}(F_{X})}\left\Vert \hat{m}%
_{\Delta }-m_{\Delta }\right\Vert _{L_{2}(F_{X})}\right) =o_{p}(n^{-1/2}).
\end{equation*}%
This proves that $\hat{\theta}-\tilde{\theta}=o_{p}(n^{-1/2})$. Therefore, $%
\hat{\theta}$ has the same asymptotic distribution as $\tilde{\theta}$. Now 
\begin{align}
& \sqrt{n}(\tilde{\theta}-\theta )  \notag \\
=& \frac{1}{\hat{\pi}_{1}\sqrt{n}}\sum_{i=1}^{n}\left( \mathcal{G}%
_{1i}(\Delta Y_{i}-m_{\Delta }(X_{i}))-\sum_{g=2}^{N_{G}+1}w_{g}(X_{i})%
\mathcal{G}_{gi}\frac{p_{1}(X_{i})}{p_{g}(X_{i})}(\Delta Y_{i}-m_{\Delta
}(X_{i}))-\theta \mathcal{G}_{1i}\right)  \notag \\
=& \frac{1}{\pi _{1}\sqrt{n}}\sum_{i=1}^{n}\left( \mathcal{G}_{1i}(\Delta
Y_{i}-m_{\Delta }(X_{i}))-\sum_{g=2}^{N_{G}+1}w_{g}(X_{i})\mathcal{G}_{gi}%
\frac{p_{1}(X_{i})}{p_{g}(X_{i})}(\Delta Y_{i}-m_{\Delta }(X_{i}))-\theta 
\mathcal{G}_{1i}\right) +o_{p}(1)  \notag \\
\equiv & \frac{1}{\sqrt{n}}\sum_{i=1}^{n}\psi _{\mathrm{PT}}(S_{i})+o_{p}(1),
\label{eqn: theta_tilde_minus_theta}
\end{align}%
where $\psi _{\mathrm{PT}}$ is defined by 
\begin{eqnarray}
\psi _{\mathrm{PT}}(S) &=&\frac{1}{\pi _{1}}\left[ \mathcal{G}_{1}(\Delta
Y-m_{\Delta }(X)-\theta )-\sum_{g=2}^{N_{G}+1}w_{g}(X)\mathcal{G}_{g}\frac{%
p_{1}(X)}{p_{g}(X)}(\Delta Y-m_{\Delta }(X))\right]  \notag \\
&=&\phi \left( S;m_{\Delta },p,w;\pi _{1}\right) -\theta \mathcal{G}_{1}/\pi
_{1}.  \label{eqn:psi_PT}
\end{eqnarray}%
By Theorem \ref{thm:identification}, $\mathbb{E}[\phi (S;m_{\Delta },p,w;\pi
_{1})]=\theta $, and $\mathbb{E}[\mathcal{G}_{1}]=\pi _{1}$, so $\mathbb{E}%
[\psi _{\mathrm{PT}}(S)]=0$. Then, $\sqrt{n}(\hat{\theta}-\theta )\overset{d}%
{\rightarrow }N(0,V_{\mathrm{PT}})$ with the asymptotic variance given by $%
V_{\mathrm{PT}}=\mathbb{E}\left[ \psi _{\mathrm{PT}}(S)^{2}\right] $. We can
represent $V_{\mathrm{PT}} $ as 
\begin{align}
V_{\mathrm{PT}}& =\mathbb{E}\left[ \psi _{\mathrm{PT}}(S)^{2}\right]
\label{eqn:V1-calculation} \\
& =\frac{1}{\pi _{1}^{2}}\left( \mathbb{E}\left[ \mathcal{G}_{1}(\Delta
Y-m_{\Delta }(X)-\theta )^{2}+\sum_{g=2}^{N_{G}+1}w_{g}(X)^{2}\mathcal{G}_{g}%
\frac{p_{1}(X)^{2}}{p_{g}(X)^{2}}(\Delta Y-m_{\Delta }(X))^{2}\right] \right)
\notag \\
& =\frac{1}{\pi _{1}^{2}}\Bigg(\mathbb{E}\Bigg[\mathcal{G}_{1}(\Delta
Y-m_{1,\Delta }(X))^{2}+\mathcal{G}_{1}(m_{1,\Delta }(X)-m_{\Delta
}(X)-\theta )^{2}  \notag \\
& \quad +\sum_{g=2}^{N_{G}+1}w_{g}(X)^{2}\mathcal{G}_{g}\frac{p_{1}(X)^{2}}{%
p_{g}(X)^{2}}(\Delta Y-m_{\Delta }(X))^{2}\Bigg]\Bigg),  \notag
\end{align}%
where in the second line the cross-product terms disappear because $\mathcal{%
G}_{g}\mathcal{G}_{g^{\prime }}=0$ for $g\neq g^{\prime }$ with $g,g^{\prime
}\geq 1$, and in the third line we use the following fact: 
\begin{align*}
& \mathbb{E}\left[ \mathcal{G}_{1}(m_{1,\Delta }(X)-m_{\Delta }(X)-\theta
)(\Delta Y-m_{1,\Delta }(X))\right] \\
=\,& \mathbb{E}\left[ \mathbb{E}\left[ p_{1}(X)(m_{1,\Delta }(X)-m_{\Delta
}(X)-\theta )(\Delta Y-m_{1,\Delta }(X))|X,G=1\right] \right] =0,
\end{align*}%
as $\mathbb{E}\left[ (\Delta Y-m_{1,\Delta }(X))|X,G=1\right] =0.$ Under
condition (i) of Theorem \ref{thm:asymp-dist-1}, the asymptotic variance is
finite. Thus, $\hat{\theta}$ is asymptotically normal with the given
asymptotic variance, completing the proof.
\end{proof}

\vspace{0.5cm}

\begin{proof}[Proof of Theorem \protect\ref{thm:asymp-dist-2}]
We follow the steps in the proof of Theorem \ref{thm:asymp-dist-1}. We have
the following decomposition: 
\begin{equation*}
\phi (S;\hat{m}_{\Delta },\hat{p},\hat{w};\hat{\pi}_{1})\hat{\pi}_{1}-\phi
(S;m_{\Delta },p,w;\hat{\pi}_{1})\hat{\pi}%
_{1}=E_{1}+E_{2}+E_{3}+E_{4}+E_{5}+E_{6},
\end{equation*}%
where 
\begin{align}
E_{1}(S;\hat{m}_{\Delta },\hat{p},\hat{w})& \equiv
-\sum_{g=2}^{N_{G}+1}w_{g}(X)\left( \hat{r}_{1,g}(X)-r_{1,g}(X)\right) 
\mathcal{G}_{g}(\Delta Y-m_{\Delta }(X)),  \notag \\
E_{2}(S;\hat{m}_{\Delta },\hat{p},\hat{w})& \equiv -\sum_{g=2}^{N_{G}+1}(%
\hat{w}_{g}(X)-w_{g}(X))r_{1,g}(X)\mathcal{G}_{g}(\Delta Y-m_{\Delta }(X)), 
\notag \\
E_{3}(S;\hat{m}_{\Delta },\hat{p},\hat{w})& \equiv -\left( \mathcal{G}%
_{1}-\sum_{g=2}^{N_{G}+1}w_{g}(X)\mathcal{G}_{g}r_{1,g}(X)\right) (\hat{m}%
_{\Delta }(X)-m_{\Delta }(X)),  \label{eqn:E1_E6-decomposition} \\
E_{4}(S;\hat{m}_{\Delta },\hat{p},\hat{w})& \equiv \sum_{g=2}^{N_{G}+1}%
\mathcal{G}_{g}r_{1,g}(X)(\hat{w}_{g}(X)-w_{g}(X))(\hat{m}_{\Delta
}(X)-m_{\Delta }(X)),  \notag \\
E_{5}(S;\hat{m}_{\Delta },\hat{p},\hat{w})& \equiv \sum_{g=2}^{N_{G}+1}%
\mathcal{G}_{g}\left( \hat{r}_{1,g}(X)-r_{1,g}(X)\right) \hat{w}_{g}(X)(\hat{%
m}_{\Delta }(X)-m_{\Delta }(X)),  \notag \\
E_{6}(S;\hat{m}_{\Delta },\hat{p},\hat{w})& \equiv -\sum_{g=2}^{N_{G}+1}(%
\hat{w}_{g}(X)-w_{g}(X))\left( \hat{r}_{1,g}(X)-r_{1,g}(X)\right) \mathcal{G}%
_{g}(\Delta Y-m_{\Delta }(X)).  \notag
\end{align}

The terms $E_{2},E_{3},E_{4},$ and $E_{5}$ are the same as those in the
proof of Theorem \ref{thm:asymp-dist-1}, and $E_{1}+E_{6}$ is the same as $%
E_{1}$ defined in the proof of Theorem \ref{thm:asymp-dist-1}. As in the
proof of Theorem \ref{thm:asymp-dist-1}, for $j=1,2,4,5,6,$ we can define $%
E_{jg}(S;\hat{m}_{\Delta },\hat{p},\hat{w})$ as the summand in $E_{j}(S;\hat{%
m}_{\Delta },\hat{p},\hat{w})$ so that $E_{j}(S;\hat{m}_{\Delta },\hat{p},%
\hat{w})=\sum_{g=2}^{N_{G}+1}E_{jg}(S;\hat{m}_{\Delta },\hat{p},\hat{w})$.

For $j=1,\ldots ,6$, denote $\hat{E}_{j}\equiv
n^{-1}\sum_{i=1}^{n}E_{j}(S_{i};\hat{m}_{\Delta },\hat{p},\hat{w})$, and for 
$j=1,2,4,5,6,$ denote $\hat{E}_{jg}\equiv n^{-1}\sum_{i=1}^{n}E_{jg}(S_{i};%
\hat{m}_{\Delta },\hat{p},\hat{w}).$ Following the same arguments as in the
proof of Theorem \ref{thm:asymp-dist-1}, we have $\hat{E}%
_{j}=o_{p}(n^{-1/2}),j=3,4,5,6$. In particular, to show that $\hat{E}_6 =
o_p(n^{-1/2})$, we use the fact that $\bigl\|\mathbb{E}[\mathcal{G}_g(\Delta
Y - m_{\Delta}(X)) |X]\bigr\|_\infty = O(1)$, which follows from condition
(i) of Theorem \ref{thm:asymp-dist-1} and the Cauchy--Schwarz inequality.
Hence, for each $j=3,4,5,6,$ $\hat{E}_{j}$ is asymptotically negligible and
does not contribute to the asymptotic variance of $\hat{\theta}.$

Unlike the case of $\hat{E}_{j}$ for $j=3,4,5,6$, $\hat{E}_{1}$ and $\hat{E}%
_{2}$ do contribute to the asymptotic variance of $\hat{\theta}$ when
Assumption \ref{asm:PT} is not imposed. We now consider each of $\hat{E}_{1}$
and $\hat{E}_{2}$ in turn. Let $\tilde{E}_{1g}$ denote the analog of $\hat{E}%
_{1g}$ constructed using the population distribution instead of the
empirical distribution, that is, 
\begin{align*}
\tilde{E}_{1g}& \equiv \mathbb{E}[E_{1g}(S;\hat{m}_{\Delta },\hat{p},\hat{w}%
)|\hat{r}_{1,g}] \\
& =-\int_{\mathcal{X}}w_{g}(x)p_{g}(x)(\hat{r}_{1,g}(x)-r_{1,g}(x))(m_{g,%
\Delta }(x)-m_{\Delta }(x))f_{X}(x)dx,
\end{align*}%
where $m_{g,\Delta }(x)\equiv \mathbb{E}[\Delta Y|G=g,X=x]$. For $i\neq
i^{\prime }$, we have 
\begin{equation*}
\mathbb{E}\left[ (E_{1g}(S_{i};\hat{m}_{\Delta },\hat{p},\hat{w})-\tilde{E}%
_{1g})(E_{1g}(S_{i^{\prime }};\hat{m}_{\Delta },\hat{p},\hat{w})-\tilde{E}%
_{1g})|\hat{r}_{1,g}\right] =0.
\end{equation*}%
Thus, 
\begin{align*}
\mathbb{E}\left[ (\hat{E}_{1g}-\tilde{E}_{1g})^{2}|\hat{r}_{1,g}\right] & =%
\frac{1}{n}\mathbb{E}\left[ (E_{1g}(S_{i};\hat{m}_{\Delta },\hat{p},\hat{w})-%
\tilde{E}_{1g})^{2}|\hat{r}_{1,g}\right] \\
& \leq C\lVert \hat{r}_{1,g}-r_{1,g}\rVert _{L_{2}(F_{X})}^{2}/n.
\end{align*}%
This implies that $\hat{E}_{1g}-\tilde{E}_{1g}=o_{p}(n^{-1/2})$. Therefore,
we only need to analyze the asymptotic behavior of $\tilde{E}_{1g}$. By the
uniform Bahadur representation, we have 
\begin{align*}
\tilde{E}_{1g}& =-\frac{1}{nh^d}\sum_{i=1}^{n}\int_{\mathcal{X}%
}w_{g}(x)p_{g}(x)K^{\text{eq}}_{\bar{s}}\left( \frac{X_{i}-x}{h}\right) 
\frac{\mathcal{G}_{1i}-r_{1,g}(X_{i})\mathcal{G}_{gi}}{p_{g}(X_{i})} \\
& \quad \quad \times (m_{g,\Delta }(x)-m_{\Delta }(x))dx+o_{p}(n^{-1/2}) \\
& =-\frac{1}{n}\sum_{i=1}^{n}\int w_{g}(X_{i}-uh)p_{g}(X_{i}-uh)K^{\text{eq}%
}_{\bar{s}}(u)\frac{\mathcal{G}_{1i}-r_{1,g}(X_{i})\mathcal{G}_{gi}}{%
p_{g}(X_{i})} \\
& \quad \quad \times (m_{g,\Delta }(X_{i}-uh)-m_{\Delta
}(X_{i}-uh))du+o_{p}(n^{-1/2}) \\
& =-\frac{1}{n}\sum_{i=1}^{n}w_{g}(X_{i})(\mathcal{G}_{1i}-r_{1,g}(X_{i})%
\mathcal{G}_{gi})(m_{g,\Delta }(X_{i})-m_{\Delta }(X_{i}))+o_{p}(n^{-1/2}),
\end{align*}%
where a change of variables $uh=X_{i}-x$ is applied to obtain the second
equality, and the third equality follows from the standard Taylor expansion
together with the assumption that $w_{g}$, $p_{g}$, and $m_{g,\Delta }$ are $%
(\bar{s}+1)$-times continuously differentiable,\footnote{%
Notice that the derivatives are bounded because $X$ has compact support.
Also, $m_{g,\Delta }=m_{g,\mathcal{T}}-m_{g,\mathcal{T}-1}$ and is $(\bar{s}%
+1)$-times continuously differentiable by assumption.} the properties of the
kernel $K^{\text{eq}}_{\bar{s}}$, and $h=o\left(n^{-\frac{1}{2(\bar{s}+1)}%
}\right)$.

Define $\hat{E}_{2g}$ and $\tilde{E}_{2g}$ analogously, and we can show that 
$\hat{E}_{2g}-\tilde{E}_{2g}=o_{p}(n^{-1/2})$. More specifically, for $%
\tilde{E}_{2g}$ and $\tilde{E}_{2}\equiv \sum_{g=2}^{N_{G}+1}\tilde{E}_{2g}$%
, we have 
\begin{align}
\tilde{E}_{2g}& =-\int_{\mathcal{X}}(\hat{w}_{g}(x)-w_{g}(x))p_{1}(x)(m_{g,%
\Delta }(x)-m_{\Delta }(x))f_{X}(x)dx,  \notag \\
\tilde{E}_{2}& =-\int_{\mathcal{X}}%
\begin{bmatrix}
\hat{\bm{w}}_{0}(x)-\bm{w}_{0}(x) \\ 
-\mathbf{1}_{N_{G}-1}^{\prime }(\hat{\bm{w}}_{0}(x)-\bm{w}_{0}(x))%
\end{bmatrix}%
^{\prime }p_{1}(x)v_{\Delta }(x)f_{X}(x)dx,  \label{eqn:E2_tilde}
\end{align}%
where 
\begin{equation*}
v_{\Delta }(x)\equiv \left( 
\begin{array}{c}
m_{2,\Delta }(x)-m_{\Delta }(x) \\ 
\vdots  \\ 
m_{N_{G}+1,\Delta }(x)-m_{\Delta }(x)%
\end{array}%
\right) .
\end{equation*}%
Notice that in (\ref{eqn:E2_tilde}), we no longer have the constant term
because of the following cancellation: 
\begin{equation*}
\hat{w}_{N_{G}+1}(x)-w_{N_{G}+1}(x)=1-\sum_{g=2}^{N_{G}}\hat{w}_{g}(x)-(1-%
\sum_{g=2}^{N_{G}}w_{g}(x))=-\mathbf{1}_{N_{G}-1}^{\prime }(\hat{\bm{w}}_{0}(x)-%
\bm{w}_{0}(x)).
\end{equation*}%
We linearize the difference $\hat{\bm{w}}_{0}(x)-\bm{w}_{0}(x)$ via their
defining equations:%
\begin{equation*}
M(x)^{\prime }M(x)\bm{w}_{0}(x)=M(x)^{\prime }m_{1}(x)\quad \text{and}\quad 
\hat{M}(x)^{\prime }\hat{M}(x)\hat{\bm{w}}_{0}(x)=\hat{M}(x)^{\prime }\hat{m}%
_{1}(x).
\end{equation*}%
Using these two equations, we obtain the exact identity 
\begin{equation*}
M(x)^{\prime }M(x)\left( \hat{\bm{w}}_{0}(x)-\bm{w}_{0}(x)\right) =\hat{M}%
(x)^{\prime }\hat{m}_{1}(x)-M(x)^{\prime }m_{1}(x)-\left[ \hat{M}(x)^{\prime
}\hat{M}(x)-M(x)^{\prime }M(x)\right] \hat{\bm{w}}_{0}(x).
\end{equation*}%
Define $D_{M}(x)\equiv \hat{M}(x)-M(x)$ and $D_{m_{1}}(x)\equiv \hat{m}%
_{1}(x)-m_{1}(x).$ Substituting $\hat{M}=M+D_{M}$, $\hat{m}%
_{1}=m_{1}+D_{m_{1}}$, and $\hat{\bm{w}}_{0}=\bm{w}_{0}+(\hat{\bm{w}}_{0}-%
\bm{w}_{0})$ and collecting terms yields 
\begin{eqnarray}
\hat{\bm{w}}_{0}(x)-\bm{w}_{0}(x) &=&(M{(x)}^{\prime }M(x))^{-1}M{(x)}%
^{\prime }D_{m_{1}}(x)+(M{(x)}^{\prime }M(x))^{-1}{D_{M}(x)}^{\prime
}m_{1}(x)  \notag \\
&&-(M{(x)}^{\prime }M(x))^{-1}[{D_{M}(x)}^{\prime }M(x)+M{(x)}^{\prime
}D_{M}(x)]{\bm{w}_{0}(x)}+R(x),  \label{eqn:w0_hat-w0}
\end{eqnarray}%
where 
\begin{align*}
R(x)\equiv (M{(x)}^{\prime }M(x))^{-1}\Big[& D_{M}(x)^{\prime
}D_{m_{1}}(x)-D_{M}(x)^{\prime }D_{M}(x)\hat{\bm{w}}_{0}(x) \\
& -\big(D_{M}(x)^{\prime }M(x)+M(x)^{\prime }D_{M}(x)\big)\big(\hat{\bm{w}}%
_{0}(x)-\bm{w}_{0}(x)\big)\Big].
\end{align*}%
We show that the contribution of $R$ to (\ref{eqn:E2_tilde}) is $%
o_{p}(n^{-1/2})$. Let $|\cdot |$ denote the Euclidean norm of a vector and
the Frobenius norm of a matrix, and let $\Vert \cdot \Vert _{L_{2}(F_{X})}$
denote the $L_{2}(F_{X})$ norm of $|\cdot |$. By condition (iii) of Theorem %
\ref{thm:asymp-dist-2}, the eigenvalues of $M(x)^{\prime }M(x)$ are bounded
below by $c$, so $|(M(x)^{\prime }M(x))^{-1}|$ is bounded uniformly over $%
x\in \mathcal{X}$; and since the functions $m_{g,t}$ are continuous on the
compact set $\mathcal{X}$, $M(x)$ and $v_{\Delta }(x)$ are uniformly
bounded. Together with $p_{1}\leq 1$, this gives: there exists a constant $C$
such that, for all $x\in \mathcal{X}$, 
\begin{align*}
\left\vert p_{1}(x){v_{\Delta }(x)}^{\prime }%
\begin{bmatrix}
R(x) \\ 
-\mathbf{1}_{N_{G}-1}^{\prime }R(x)%
\end{bmatrix}%
\right\vert & \leq C\Big\{|D_{M}(x)||D_{m_{1}}(x)|+|D_{M}(x)|^{2}|\hat{\bm{w}%
}_{0}(x)| \\
& \qquad +|D_{M}(x)||\hat{\bm{w}}_{0}(x)-\bm{w}_{0}(x)|\Big\}.
\end{align*}%
Integrating against $f_{X}$ and applying the Cauchy--Schwarz inequality, the
contribution of $R$ to (\ref{eqn:E2_tilde}) is bounded in absolute value by 
\begin{equation*}
C\left\{ \Vert D_{M}\Vert _{L_{2}(F_{X})}\Vert D_{m_{1}}\Vert
_{L_{2}(F_{X})}+\sup_{x\in \mathcal{X}}|\hat{\bm{w}}_{0}(x)|\,\Vert
D_{M}\Vert _{L_{2}(F_{X})}^{2}+\Vert D_{M}\Vert _{L_{2}(F_{X})}\Vert 
\hat{\bm{w}}_{0}-\bm{w}_{0}\Vert _{L_{2}(F_{X})}\right\} =o_{p}(n^{-1/2}),
\end{equation*}%
because the entries of $D_{M}$ and $D_{m_{1}}$ are linear combinations of $%
\hat{m}_{g,t}-m_{g,t}$, so that $\Vert D_{M}\Vert _{L_{2}(F_{X})}$ and $%
\Vert D_{m_{1}}\Vert _{L_{2}(F_{X})}$ are $o_{p}(n^{-1/4})$ by condition
(i) of Theorem \ref{thm:asymp-dist-2}; $\Vert \hat{\bm{w}}_{0}-\bm{w}%
_{0}\Vert _{L_{2}(F_{X})}=o_{p}(n^{-1/4})$ by condition (i) of Theorem %
\ref{thm:asymp-dist-2}; and $\sup_{x\in \mathcal{X}}|\hat{\bm{w}}%
_{0}(x)|=O_{p}(1)$ by condition (ii) of Theorem \ref{thm:asymp-dist-1}.
Combining (\ref{eqn:E2_tilde}) with (\ref{eqn:w0_hat-w0}), we know that up
to an error term of order $o_{p}(n^{-1/2})$, $\tilde{E}_{2}$ is equal to the
sum of the following three terms: 
\begin{align*}
& -\int_{\mathcal{X}}p_{1}(x){v_{\Delta }(x)}^{\prime }%
\begin{bmatrix}
(M(x)^{\prime }M(x))^{-1}M(x)^{\prime }D_{m_{1}}(x) \\ 
-\mathbf{1}_{N_{G}-1}^{\prime }(M(x)^{\prime }M(x))^{-1}M(x)^{\prime
}D_{m_{1}}(x)%
\end{bmatrix}%
f_{X}(x)\,dx, \\
& -\int_{\mathcal{X}}p_{1}(x){v_{\Delta }(x)}^{\prime }%
\begin{bmatrix}
(M(x)^{\prime }M(x))^{-1}{D_{M}(x)}^{\prime }m_{1}(x) \\ 
-\mathbf{1}_{N_{G}-1}^{\prime }(M(x)^{\prime }M(x))^{-1}{D_{M}(x)}^{\prime
}m_{1}(x)%
\end{bmatrix}%
f_{X}(x)\,dx, \\
& -\int_{\mathcal{X}}p_{1}(x){v_{\Delta }(x)}^{\prime }%
\begin{bmatrix}
-(M(x)^{\prime }M(x))^{-1}\big[{D_{M}(x)}^{\prime }M(x)+M(x)^{\prime
}D_{M}(x)\big]{\bm{w}_{0}(x)} \\ 
\mathbf{1}_{N_{G}-1}^{\prime }(M(x)^{\prime }M(x))^{-1}\big[{D_{M}(x)}%
^{\prime }M(x)+M(x)^{\prime }D_{M}(x)\big]{\bm{w}_{0}(x)}%
\end{bmatrix}%
f_{X}(x)\,dx.
\end{align*}%
Each of the three terms above is an integral of a smooth function of $x$
against a nuisance estimation error. We now convert integrals of this form
into sample averages, which yields the influence function of $\tilde{E}_{2}$%
. Using the uniform Bahadur representation of $\hat{m}_{g,t}$ together with
the previous change of variables techniques, we can show that%
\begin{equation}
\int_{\mathcal{X}}(\hat{m}_{g,t}(x)-m_{g,t}(x))\eta (x)f_{X}(x)dx=\frac{1}{n}%
\sum_{i=1}^{n}\left( \frac{\mathcal{G}_{gi}(Y_{ti}-m_{g,t}(X_{i}))}{%
p_{g}(X_{i})}\right) \eta (X_{i})+o_{p}(n^{-1/2}),  \label{eqn:m_hat-m}
\end{equation}%
for any $(\bar{s}+1)$-times continuously differentiable function $\eta (x)$,
applied componentwise when $\eta $ is vector- or matrix-valued.

We now introduce $A(S)$ and $a_{1}(S)$, the influence-function counterparts
of $D_{M}(x)$ and $D_{m_{1}}(x)$. They share the block structure of $M(x)$
and $m_{1}(x)$ in (\ref{eqn:M-m1-def}), with each conditional mean $%
m_{g,t}(x)$ replaced by the summand $\mathcal{G}%
_{g}(Y_{t}-m_{g,t}(X))/p_{g}(X)$ appearing in (\ref{eqn:m_hat-m}).
Accordingly, define $A(S)$ as the following $(\mathcal{T}-1)\times (N_{G}-1)$
matrix 
\begin{equation*}
A(S)\equiv 
\begin{pmatrix}
a_{-1,1}(S)^{\prime } \\ 
\vdots \\ 
a_{-1,\mathcal{T}-1}(S)^{\prime }%
\end{pmatrix}%
,
\end{equation*}%
where $a_{-1,t}(S)$ is defined as 
\begin{equation*}
a_{-1,t}(S)\equiv \left[ 
\begin{array}{c}
\frac{\mathcal{G}_{2}(Y_{t}-m_{2,t}(X))}{p_{2}(X)}-\frac{\mathcal{G}%
_{N_{G}+1}(Y_{t}-m_{N_{G}+1,t}(X))}{p_{N_{G}+1}(X)} \\ 
\vdots \\ 
\frac{\mathcal{G}_{N_{G}}(Y_{t}-m_{N_{G},t}(X))}{p_{N_{G}}(X)}-\frac{%
\mathcal{G}_{N_{G}+1}(Y_{t}-m_{N_{G}+1,t}(X))}{p_{N_{G}+1}(X)}%
\end{array}%
\right]
\end{equation*}%
for $1\leq t\leq \mathcal{T}-1.$ Similarly, we define 
\begin{equation*}
a_{1}(S)\equiv \left[ 
\begin{array}{c}
\frac{\mathcal{G}_{1}(Y_{1}-m_{1,1}(X))}{p_{1}(X)}-\frac{\mathcal{G}%
_{N_{G}+1}(Y_{1}-m_{N_{G}+1,1}(X))}{p_{N_{G}+1}(X)} \\ 
\vdots \\ 
\frac{\mathcal{G}_{1}(Y_{\mathcal{T}-1}-m_{1,\mathcal{T}-1}(X))}{p_{1}(X)}-%
\frac{\mathcal{G}_{N_{G}+1}(Y_{\mathcal{T}-1}-m_{N_{G}+1,\mathcal{T}-1}(X))}{%
p_{N_{G}+1}(X)}%
\end{array}%
\right] .
\end{equation*}%
Define $P(S)\equiv P_{1}(S)+P_{2}(S)+P_{3}(S)$, where 
\begin{align*}
P_{1}(S)& \equiv \left( M(X)^{\prime }M\left( X\right) \right)
^{-1}M(X)^{\prime }a_{1}(S), \\
P_{2}(S)& \equiv \left( M(X)^{\prime }M\left( X\right) \right)
^{-1}A(S)^{\prime }m_{1}(X), \\
P_{3}(S)& \equiv -\left( M(X)^{\prime }M\left( X\right) \right) ^{-1}\left[
A(S)^{\prime }M(X)+M(X)^{\prime }A(S)\right] {\bm{w}_{0}(X)}.
\end{align*}%
By (\ref{eqn:m_hat-m}), $P_{1}(S)$, $P_{2}(S)$, and $P_{3}(S)$ are the
influence functions of the three leading terms on the right-hand side of (%
\ref{eqn:w0_hat-w0}), in the same order. Summing the three contributions,
the influence function of $\hat{E}_{2}$ is 
\begin{equation*}
-p_{1}(X){v_{\Delta }(X)}^{\prime }%
\begin{bmatrix}
P(S) \\ 
-\mathbf{1}_{N_{G}-1}^{\prime }P(S)%
\end{bmatrix}%
.
\end{equation*}

Therefore, 
\begin{eqnarray*}
\sqrt{n}\left( \hat{\theta}-\tilde{\theta}\right) &=&\frac{1}{\sqrt{n}\hat{%
\pi}_{1}}\sum_{i=1}^{n}\left[ \phi (S_{i};\hat{m}_{\Delta },\hat{p},\hat{w};%
\hat{\pi}_{1})\hat{\pi}_{1}-\phi (S_{i};m_{\Delta },p,w;\hat{\pi}_{1})\hat{%
\pi}_{1}\right] \\
&=&\frac{1}{\sqrt{n}\pi _{1}}\sum_{i=1}^{n}\left[ \phi (S_{i};\hat{m}%
_{\Delta },\hat{p},\hat{w};\hat{\pi}_{1})\hat{\pi}_{1}-\phi (S_{i};m_{\Delta
},p,w;\hat{\pi}_{1})\hat{\pi}_{1}\right] \left( 1+o_{p}\left( 1\right)
\right) \\
&=&\frac{1}{\sqrt{n}}\sum_{i=1}^{n}\left[ \psi _{1}(S_{i})+\psi _{2}(S_{i})%
\right] +o_{p}\left( 1\right) ,
\end{eqnarray*}%
where 
\begin{eqnarray}
\psi _{1}\left( S\right) &=&-\frac{p_{1}(X)}{\pi _{1}}%
\sum_{g=2}^{N_{G}+1}w_{g}(X)\left( \frac{\mathcal{G}_{1}}{p_{1}(X)}-\frac{%
\mathcal{G}_{g}}{p_{g}(X)}\right) v_{\Delta ,g}(X)=-\frac{p_{1}(X)}{\pi _{1}}%
v_{\Delta }(X)^{\prime }Q(S),  \label{eqn:psi1} \\
\psi _{2}(S) &=&-\frac{p_{1}(X)}{\pi _{1}}{v_{\Delta }(X)}^{\prime }%
\begin{bmatrix}
P(S) \\ 
-\mathbf{1}_{N_{G}-1}^{\prime }P(S)%
\end{bmatrix}%
.  \label{eqn:psi2}
\end{eqnarray}
Here $v_{\Delta ,g}(X)$ denotes the entry of $v_{\Delta }(X) $ corresponding
to group $g$, and $Q(S)$ is the $N_{G}$-vector whose entry for group $g$ is $%
w_{g}(X)\left( \mathcal{G}_{1}/p_{1}(X)-\mathcal{G}_{g}/p_{g}(X)\right) $.
Both adjustments thus share the factor $p_{1}(X)v_{\Delta }(X)^{\prime }/\pi
_{1}$. By Bayes' rule, $p_{1}(x)/\pi _{1}=f_{X|G=1}(x)/f_{X}(x) $, where $%
f_{X|G=1}$ is the density of $X$ in the treated group. The adjustments are
therefore linear in the deviations from Assumption \ref{asm:PT}, with the
population covariate distribution reweighted to that of the treated group.

Using the above and (\ref{eqn: theta_tilde_minus_theta}), we have 
\begin{eqnarray*}
\sqrt{n}(\hat{\theta}-\theta ) &=&\sqrt{n}(\hat{\theta}-\tilde{\theta})+%
\sqrt{n}\left( \tilde{\theta}-\theta \right) \\
&=&\frac{1}{\sqrt{n}}\sum_{i=1}^{n}\left( \psi _{\mathrm{PT}}\left(
S_{i}\right) +\psi _{1}\left( S_{i}\right) +\psi _{2}\left( S_{i}\right)
\right) +o_{p}\left( 1\right) \\
&=&\frac{1}{\sqrt{n}}\sum_{i=1}^{n}\psi _{\mathrm{SC}}\left( S_{i}\right)
+o_{p}\left( 1\right) ,
\end{eqnarray*}%
where $\psi_{\mathrm{SC}}$ is defined by 
\begin{equation}
\psi _{\mathrm{SC}}(S)=\psi _{\mathrm{PT}}(S)+\psi _{1}(S)+\psi _{2}(S),
\label{eqn:psi_SC}
\end{equation}%
and $\psi _{\mathrm{PT}}(S)$ is defined in (\ref{eqn:psi_PT}). So the
influence function of $\hat{\theta}$ is $\psi _{\mathrm{SC}}(S)$, which
consists of three parts: $\psi _{\mathrm{PT}}(S)$ is the part that assumes
the nuisance functions are known, and $\psi _{1}\left( S\right) $ and $\psi
_{2}\left( S\right) $ are adjustments to the influence function accounting
for errors in estimating the nuisance functions. Both corrections have mean
zero: $\mathbb{E}[Q(S)|X]=0$ because $\mathbb{E}[\mathcal{G}%
_{g}/p_{g}(X)|X]=1$ for every $g$, and $\mathbb{E}[P(S)|X]=0$ because every
entry of $a_{1}(S)$ and $A(S)$ has conditional mean zero given $X$. The
asymptotic variance $V_{\mathrm{SC}}$ of $\sqrt{n}(\hat{\theta}-\theta )$ is 
$\mathbb{E}\left[ \psi _{\mathrm{SC}}(S)^{2}\right] .$
\end{proof}

\vspace{0.5cm}

\begin{proof}[Proof of Theorem \protect\ref{thm:bootstrap}]
Recall that $o_{p}(1)$ and $O_{p}(1)$ are defined under the joint
distribution of the bootstrap weights $\mathcal{W}_{n}$ and the sample $%
\mathcal{S}_{n},$ and $o_{p}^{\ast }(1)$ and $O_{p}^{\ast }(1)$ are defined
based on the conditional probability distribution of $\mathcal{W}_{n}$ given
the sample $\mathcal{S}_{n}.$ According to Lemma 3 in 
\cite{cheng2010bootstrap}%
, $o_{p}(1)$ is equivalent to $o_{p}^{\ast }(1)$, and $O_{p}(1)$ is
equivalent to $O_{p}^{\ast }(1)$. In addition, $o_{p}^{\ast }(1)\times
O_{p}^{\ast }(1)=o_{p}^{\ast }(1)$. Therefore, depending on the context, we
will use whichever of the two notions is more convenient or desired.

For Part (i), we want to derive an asymptotic representation for the
bootstrap estimator under the parallel trends condition. Define $\tilde{%
\theta}^{\ast }$ as the bootstrap estimator constructed using the true
values of the nuisance functions: 
\begin{equation*}
\tilde{\theta}^{\ast }\equiv \frac{1}{n}\sum_{i=1}^{n}W_{i}\phi
(S_{i};m_{\Delta },p,w;\hat{\pi}_{1}^{\ast }).
\end{equation*}%
Define the numerator of $\phi $ as $\phi _{\text{num}}(S)\equiv \phi
(S;m_{\Delta },p,w;\pi _{1})\pi _{1}$. The difference between $\tilde{\theta}%
^{\ast }$ and the true $\theta $ is 
\begin{align}
\tilde{\theta}^{\ast }-\theta & =\frac{\frac{1}{n}\sum_{i=1}^{n}W_{i}\phi _{%
\text{num}}(S_{i})}{\hat{\pi}_{1}^{\ast }}-\theta =\frac{1}{\hat{\pi}%
_{1}^{\ast }}\left( \frac{1}{n}\sum_{i=1}^{n}W_{i}\phi _{\text{num}}(S_{i})-%
\hat{\pi}_{1}^{\ast }\theta \right)  \notag \\
& =\frac{1}{\hat{\pi}_{1}^{\ast }}\left[ \frac{1}{n}\sum_{i=1}^{n}W_{i}\phi
_{\text{num}}(S_{i})-\left( \frac{1}{n}\sum_{i=1}^{n}W_{i}\mathcal{G}%
_{1i}\right) \theta \right] =\frac{1}{\hat{\pi}_{1}^{\ast }}\frac{1}{n}%
\sum_{i=1}^{n}W_{i}\left( \phi _{\text{num}}(S_{i})-\mathcal{G}_{1i}\theta
\right)  \notag \\
& =\frac{1}{\pi _{1}}\frac{1}{n}\sum_{i=1}^{n}W_{i}\left( \phi _{\text{num}%
}(S_{i})-\mathcal{G}_{1i}\theta \right) +\frac{1}{n}\sum_{i=1}^{n}W_{i}%
\left( \phi _{\text{num}}(S_{i})-\theta \mathcal{G}_{1i}\right) \left( \frac{%
1}{\hat{\pi}_{1}^{\ast }}-\frac{1}{\pi _{1}}\right)  \notag \\
& =\frac{1}{n}\sum_{i=1}^{n}W_{i}\psi _{\mathrm{PT}}(S_{i})+o_{p}(n^{-1/2}),
\label{eqn:theta_tilde_star-theta}
\end{align}%
where $\psi _{\mathrm{PT}}(S)=\phi (S;m_{\Delta },p,w;\pi _{1})-\theta 
\mathcal{G}_{1}/\pi _{1}$, as defined in (\ref{eqn:psi_PT}), and we have
used the fact that $\hat{\pi}_{1}^{\ast }=\pi _{1}+o_{p}(1)$ with $\pi _{1}$
being strictly positive (which is true under Assumption \ref{asm:overlap}).
Using the notations from the proof of Theorem \ref{thm:asymp-dist-1}, the
difference between $\hat{\theta}^{\ast }$ and $\tilde{\theta}^{\ast }$ is
equal to 
\begin{equation*}
\hat{\theta}^{\ast }-\tilde{\theta}^{\ast }=\frac{1}{\hat{\pi}_{1}^{\ast }}%
\left( \sum_{g=2}^{N_{G}+1}\sum_{j=1,2,5}\hat{E}_{jg}^{\ast }+\hat{E}%
_{3}^{\ast }+\hat{E}_{4}^{\ast }\right) ,
\end{equation*}%
where the superscript `$\ast $' denotes that the object is constructed using
the bootstrap weights and nuisance estimates, for example, 
\begin{equation}
\hat{E}_{jg}^{\ast }\equiv \frac{1}{n}\sum_{i=1}^{n}W_{i}E_{jg}(S_{i};\hat{m}%
_{\Delta }^{\ast },\hat{p}^{\ast },\hat{w}^{\ast }).
\label{eqn:E-star-definition}
\end{equation}%
Because the weights $W_{i}$ are iid and independent of the data $\mathcal{S}%
_{n}$ and the bootstrapped nuisance estimators satisfy the same requirements
as in Theorem \ref{thm:asymp-dist-1}, we can follow the same procedure as in
the proof of Theorem \ref{thm:asymp-dist-1} to show that $\hat{\theta}^{\ast
}-\tilde{\theta}^{\ast }=o_{p}(n^{-1/2})$. Combining this with (\ref%
{eqn:theta_tilde_star-theta}), we obtain that 
\begin{equation*}
\sqrt{n}(\hat{\theta}^{\ast }-\theta )=\frac{1}{\sqrt{n}}\sum_{i=1}^{n}W_{i}%
\psi _{\mathrm{PT}}(S_{i})+o_{p}(1).
\end{equation*}%
From the proof of Theorem \ref{thm:asymp-dist-1}, we also have 
\begin{equation*}
\sqrt{n}(\hat{\theta}-\theta )=\frac{1}{\sqrt{n}}\sum_{i=1}^{n}\psi _{%
\mathrm{PT}}(S_{i})+o_{p}(1).
\end{equation*}%
Subtracting the two representations and applying Lemma 3 in 
\cite{cheng2010bootstrap}%
, we obtain that 
\begin{equation*}
\sqrt{n}(\hat{\theta}^{\ast }-\hat{\theta})=\frac{1}{\sqrt{n}}%
\sum_{i=1}^{n}\left( W_{i}-1\right) \psi _{\mathrm{PT}}(S_{i})+o_{p}^{\ast
}(1).
\end{equation*}%
Taking $\delta $ to be the smaller of the exponents in Assumption \ref%
{asm:bootstrap-weights} and in the moment condition of Theorem \ref%
{thm:bootstrap}(i), both $\mathbb{E}[|W_{i}|^{2+\delta }]$ and $\mathbb{E}%
[|\psi _{\mathrm{PT}}(S_{i})|^{2+\delta }]$ are finite for the same $\delta $%
. Because we assume that $\mathbb{E}[|\psi _{\mathrm{PT}}(S_{i})|^{2+\delta
}]<\infty $ for some $\delta >0$, the strong law of large numbers gives,
almost surely, $n^{-1}\sum_{i=1}^{n}|\psi _{\mathrm{PT}}(S_{i})|^{2+\delta
}\rightarrow \mathbb{E}[|\psi _{\mathrm{PT}}(S_{i})|^{2+\delta }]$ and $%
n^{-1}\sum_{i=1}^{n}\psi _{\mathrm{PT}}(S_{i})^{2}\rightarrow \mathbb{E}%
[\psi _{\mathrm{PT}}(S_{i})^{2}]=V_{\mathrm{PT}}>0$. In particular, $%
\limsup_{n \rightarrow \infty }n^{-1}\sum_{i=1}^{n}|\psi _{\mathrm{PT}%
}(S_{i})|^{2+\delta }<\infty $ almost surely. By Lemma \ref{lm:clt} and the
independence between the bootstrap weights $\mathcal{W}_{n}$ and the sample
data $\mathcal{S}_{n}$, we have that 
\begin{equation*}
\left\vert \mathbb{P}^{\ast }\left( \sqrt{n}(\hat{\theta}^{\ast }-\hat{\theta%
})\leq a\big|\mathcal{S}_{n}\right) -\mathbb{P}\left( \sqrt{n}(\hat{\theta}%
-\theta )\leq a\right) \right\vert =o_{p}(1).
\end{equation*}%
The uniformity (over the real line) in the convergence follows from Lemma
2.11 of \cite{vanderVaart1998}.

For Part (ii), using the notations from the proof of Theorem \ref%
{thm:asymp-dist-2}, the difference between $\hat{\theta}^{\ast }$ and $%
\tilde{\theta}^{\ast }$ is equal to 
\begin{equation*}
\hat{\theta}^{\ast }-\tilde{\theta}^{\ast }=\frac{1}{\hat{\pi}_{1}^{\ast }}%
\left( \sum_{j=1,2,4,5,6}\hat{E}_{j}^{\ast }+\hat{E}_{3}^{\ast }\right) :=%
\frac{1}{\hat{\pi}_{1}^{\ast }}\left( \sum_{j=1,2,4,5,6}\sum_{g=2}^{N_{G}+1}%
\hat{E}_{jg}^{\ast }+\hat{E}_{3}^{\ast }\right) ,
\end{equation*}%
where the superscript `$\ast $' denotes that the object is constructed using
the bootstrap weights and nuisance estimates. Similar to (\ref%
{eqn:E-star-definition}), we define $\hat{E}_{j}^{\ast }$ and $\hat{E}%
_{jg}^{\ast }$ as the bootstrap versions of the quantities in (\ref%
{eqn:E1_E6-decomposition}) and the corresponding summands of these
quantities. Because the weights $W_{i}$ are iid and independent of the data $%
\mathcal{S}_{n}$, we can show that $\hat{E}_{j}^{\ast }=o_{p}(n^{-1/2})$ for 
$j=3,4,5,6$ by using arguments similar to those in the proof of Theorem \ref%
{thm:asymp-dist-1} and in Part (i) of this proof.

Similar to the proof of Theorem \ref{thm:asymp-dist-2}, we let $\tilde{E}%
_{1g}^{\ast }$ denote the analog of $\hat{E}_{1g}^{\ast }$ constructed using
the population distribution instead of the empirical distribution, i.e., 
\begin{align*}
\tilde{E}_{1g}^{\ast }& \equiv \mathbb{E}[W_{i}E_{1g}(S_{i};\hat{m}_{\Delta
}^{\ast },\hat{p}^{\ast },\hat{w}^{\ast })|\hat{r}_{1,g}^{\ast }] \\
& =\underbrace{\mathbb{E}[W_{i}|\hat{r}_{1,g}^{\ast }]}_{=1}\mathbb{E}%
[E_{1g}(S_{i};\hat{m}_{\Delta }^{\ast },\hat{p}^{\ast },\hat{w}^{\ast })|%
\hat{r}_{1,g}^{\ast }] \\
& =-\int_{\mathcal{X}}w_{g}(x)p_{g}(x)(\hat{r}_{1,g}^{\ast
}(x)-r_{1,g}(x))(m_{g,\Delta }(x)-m_{\Delta }(x))f_{X}(x)dx,
\end{align*}%
where $E_{1g}$ depends on the nuisance estimates only through $\hat{r}%
_{1,g}^{\ast }$ and, due to cross-fitting, $\hat{r}_{1,g}^{\ast }$ is
constructed using weights independent of $W_{i}$ and $S_{i}$. We can follow
the same procedure as in the proof of Theorem \ref{thm:asymp-dist-2} to show
that $\hat{E}_{1g}^{\ast }-\tilde{E}_{1g}^{\ast }=o_{p}(n^{-1/2})$. More
specifically, for $i\neq i^{\prime }$, we have 
\begin{equation*}
\mathbb{E}\left[ (W_{i}E_{1g}(S_{i};\hat{m}_{\Delta }^{\ast },\hat{p}^{\ast
},\hat{w}^{\ast })-\tilde{E}_{1g}^{\ast })(W_{i^{\prime
}}E_{1g}(S_{i^{\prime }};\hat{m}_{\Delta }^{\ast },\hat{p}^{\ast },\hat{w}%
^{\ast })-\tilde{E}_{1g}^{\ast })|\hat{r}_{1,g}^{\ast }\right] =0.
\end{equation*}%
Thus, 
\begin{align*}
\mathbb{E}\left[ (\hat{E}_{1g}^{\ast }-\tilde{E}_{1g}^{\ast })^{2}|\hat{r}%
_{1,g}^{\ast }\right] & =\frac{1}{n}\mathbb{E}\left[ (W_{i}E_{1g}(S_{i};\hat{%
m}_{\Delta }^{\ast },\hat{p}^{\ast },\hat{w}^{\ast })-\tilde{E}_{1g}^{\ast
})^{2}|\hat{r}_{1,g}^{\ast }\right] \\
& \leq C\lVert \hat{r}_{1,g}^{\ast }-r_{1,g}\rVert _{L_{2}(F_{X})}^{2}/n.
\end{align*}%
This implies that $\hat{E}_{1g}^{\ast }-\tilde{E}_{1g}^{\ast
}=o_{p}(n^{-1/2})$. Therefore, we only need to analyze the asymptotic
behavior of $\tilde{E}_{1g}^{\ast }$. By the uniform Bahadur representation
specified in Theorem \ref{thm:bootstrap}, we have 
\begin{align*}
\tilde{E}_{1g}^{\ast }& =-\frac{1}{nh^d}\sum_{i=1}^{n}\int_{\mathcal{X}%
}w_{g}(x)p_{g}(x)W_{i}K^{\text{eq}}_{\bar{s}}\left( \frac{X_{i}-x}{h}\right)
(\mathcal{G}_{1i}-r_{1,g}(X_{i})\mathcal{G}_{gi})/p_{g}(X_{i}) \\
& \quad \quad \times (m_{g,\Delta }(x)-m_{\Delta }(x))dx+o_{p}(n^{-1/2}) \\
& =-\frac{1}{n}\sum_{i=1}^{n}\int w_{g}(X_{i}-uh)p_{g}(X_{i}-uh)K^{\text{eq}%
}_{\bar{s}}(u)W_{i}(\mathcal{G}_{1i}-r_{1,g}(X_{i})\mathcal{G}%
_{gi})/p_{g}(X_{i}) \\
& \quad \quad \times (m_{g,\Delta }(X_{i}-uh)-m_{\Delta
}(X_{i}-uh))du+o_{p}(n^{-1/2}) \\
& =-\frac{1}{n}\sum_{i=1}^{n}W_{i}w_g(X_i)(\mathcal{G}_{1i}-r_{1,g}(X_i)%
\mathcal{G}_{gi}) (m_{g,\Delta}(X_i)-m_\Delta(X_i))+o_{p}(n^{-1/2}),
\end{align*}%
where a change of variables $uh=X_{i}-x$ is applied as before. Therefore, $%
\frac{1}{\pi_1}\sum_{g=2}^{N_G+1}\tilde{E}_{1g}^{\ast } = \frac{1}{n}%
\sum_{i=1}^n W_i\psi_1(S_i) +o_p(n^{-1/2}), $ where $\psi _{1} $ is the
adjustment term defined in (\ref{eqn:psi1}). Note that the bootstrap weights
can be arranged to the left of $\psi _{1}(S_{i})$ because the nuisance
functions $\hat{r}_{1,g}^{\ast }(x)-r_{1,g}(x)$ enter $\tilde{E}_{1g}^{\ast
} $ in a linear way. The bootstrap nuisance estimates can also be linearized
into $\tilde{E}_{2}^{\ast }$ according to (\ref{eqn:E2_tilde}) and (\ref%
{eqn:w0_hat-w0}). Then we can follow the same procedure to derive the
influence function for $\hat{E}_{2}^{\ast }$, which is $\frac{1}{\pi_1}\hat{E%
}_{2}^{\ast }=\frac{1}{n}\sum_{i=1}^{n}W_{i}\psi _{2}(S_{i})+o_{p}(n^{-1/2})$,
where $\psi _{2}$ is the adjustment term defined in (\ref{eqn:psi2}).
Therefore, we have 
\begin{equation*}
\hat{\theta}^{\ast }-\tilde{\theta}^{\ast }=\frac{1}{n}\sum_{i=1}^{n}W_{i}(%
\psi _{1}(S_{i})+\psi _{2}(S_{i}))+o_{p}(n^{-1/2}).
\end{equation*}%
Taking $\delta $ to be the smaller of the exponents in Assumption \ref%
{asm:bootstrap-weights} and in the moment condition of Theorem \ref%
{thm:bootstrap}(ii), both $\mathbb{E}[|W_{i}|^{2+\delta }]$ and $\mathbb{E}%
[|\psi _{\mathrm{SC}}(S_{i})|^{2+\delta }]$ are finite for the same $\delta $%
. The remainder of the proof follows the same steps as in Part (i) to apply
Lemma \ref{lm:clt}.
\end{proof}

\pagebreak \renewcommand{\thepage}{S.\arabic{page}} \setcounter{page}{1}

\begin{center}
{\LARGE \bfseries Supplementary Appendix}
\end{center}

\bigskip

\section{Proofs of other theoretical results}

This section presents and proves a lemma used in the main appendix and
provides proofs of the remaining theoretical results not proved in Appendix %
\ref{sec:main-proofs}.

\begin{lemma}
\label{lm:clt} Let $a_{ni},1\leq i\leq n,$ be a triangular sequence of
deterministic constants such that $\lim_{n\rightarrow \infty
}n^{-1}\sum_{i=1}^{n}a_{ni}^{2}=\alpha ^{2}>0$ and $\limsup_{n \rightarrow
\infty}n^{-1}\sum_{i=1}^{n}|a_{ni}|^{2+\delta }<\infty $ for some $\delta >0$. Let $X_{ni},1\leq
i\leq n,$ be an iid triangular sequence of random variables such that $%
\mathbb{E}[X_{ni}]=0$, $\mathbb{E}[X_{ni}^{2}]=1$, and $\sup_{n}\mathbb{E}%
[|X_{n1}|^{2+\delta }]<\infty $. Then $\frac{1}{\sqrt{n}%
}\sum_{i=1}^{n}a_{ni}X_{ni}$ converges in distribution to $N(0,\alpha ^{2})$%
, as $n\rightarrow \infty $.
\end{lemma}

\begin{proof}[Proof of Lemma \protect\ref{lm:clt}]
Define $\tilde{X}_{ni}\equiv a_{ni}X_{ni}$ and apply the Lindeberg--Feller
central limit theorem for triangular arrays to the sequence $\tilde{X}_{ni}$%
. It suffices for us to verify the Lyapunov condition, which implies the
Lindeberg condition. Notice that the $\tilde{X}_{ni}$'s are independent with
mean $\mathbb{E}[\tilde{X}_{ni}]=0$ and variance $\mathbb{E}[\tilde{X}%
_{ni}^{2}]=a_{ni}^{2}$. We verify the Lyapunov condition as follows: 
\begin{align*}
\frac{\sum_{i=1}^{n}\mathbb{E}[|\tilde{X}_{ni}|^{2+\delta }]}{\left(
\sum_{i=1}^{n}\mathbb{E}[|\tilde{X}_{ni}|^{2}]\right) ^{(2+\delta )/2}}& =%
\mathbb{E}[|X_{n1}|^{2+\delta }]\frac{\sum_{i=1}^{n}|a_{ni}|^{2+\delta }}{%
\left( \sum_{i=1}^{n}a_{ni}^{2}\right) ^{(2+\delta )/2}} \\
& =\mathbb{E}[|X_{n1}|^{2+\delta }]\frac{\frac{1}{n}%
\sum_{i=1}^{n}|a_{ni}|^{2+\delta }}{\left( \frac{1}{n}%
\sum_{i=1}^{n}a_{ni}^{2}\right) ^{(2+\delta )/2}}n^{-\delta /2}=o(1),
\end{align*}%
where we used the conditions that $n^{-1}\sum_{i=1}^{n}a_{ni}^{2}\rightarrow
\alpha ^{2}>0$, $\limsup_{n\rightarrow \infty
}n^{-1}\sum_{i=1}^{n}|a_{ni}|^{2+\delta }<\infty $, and $\sup_{n}\mathbb{E}%
[|X_{n1}|^{2+\delta }]<\infty $. Then, by the Lindeberg--Feller central
limit theorem, we have 
\begin{equation*}
\frac{\frac{1}{\sqrt{n}}\sum_{i=1}^{n}a_{ni}X_{ni}}{\sqrt{\frac{1}{n}%
\sum_{i=1}^{n}a_{ni}^{2}}}=\frac{\sum_{i=1}^{n}a_{ni}X_{ni}}{\sqrt{%
\sum_{i=1}^{n}a_{ni}^{2}}}\overset{d}{\rightarrow }N(0,1),
\end{equation*}%
which implies the desired convergence in distribution, in view of Slutsky's
theorem and the assumption that $n^{-1}\sum_{i=1}^{n}a_{ni}^{2}\rightarrow
\alpha ^{2}$.
\end{proof}

\vspace{0.5cm}

\begin{proof}[Proof of Corollary \protect\ref{thm:rc-identification}]
Under Assumption \ref{asm:time-invariance}, we have 
\begin{equation*}
\mathbb{E}[\phi _{\text{rc}}(S^{\text{rc}};\mu _{G\neq 1,\mathcal{T}},\mu
_{G\neq 1,\mathcal{T}-1},p,w;\pi _{1},\lambda _{\mathcal{T}},\lambda _{%
\mathcal{T}-1})]=\mathbb{E}[\phi (S;m_{\Delta },p,w;\pi _{1})],
\end{equation*}%
where, on the right-hand side, we replace $\Delta Y$ and $m_{\Delta }(X)$ in 
$\phi (S;m_{\Delta },p,w;\pi _{1})$ with%
\begin{equation*}
\left( \frac{\mathfrak{T}_{\mathcal{T}}}{\lambda _{\mathcal{T}}}-\frac{%
\mathfrak{T}_{\mathcal{T}-1}}{\lambda _{\mathcal{T}-1}}\right) Y\text{ and }%
\left( \frac{\mu _{G\neq 1,\mathcal{T}}(X)}{\lambda _{\mathcal{T}}}-\frac{%
\mu _{G\neq 1,\mathcal{T}-1}(X)}{\lambda _{\mathcal{T}-1}}\right) ,
\end{equation*}%
respectively. Then the claims of the corollary follow from Theorem \ref%
{thm:identification}.
\end{proof}

\vspace{0.5cm}

\begin{proof}[Proof of Theorem \protect\ref{thm:rc-PT}]
We follow similar steps as in the proof of Theorem \ref{thm:asymp-dist-1}
and assume that the nuisance estimators are constructed from another
independent sample. Define $\tilde{\theta}_{\text{rc}}$ as the infeasible
estimator constructed using the true nuisance parameters: 
\begin{equation*}
\tilde{\theta}_{\text{rc}}\equiv \frac{1}{n}\sum_{i=1}^{n}\phi _{\text{rc}%
}(S_{i}^{\text{rc}};\mu _{G\neq 1,\mathcal{T}},\mu _{G\neq 1,\mathcal{T}%
-1},p,w;\hat{\pi}_{1},\hat{\lambda}_{\mathcal{T}},\hat{\lambda}_{\mathcal{T}%
-1}).
\end{equation*}%
The difference between the numerators of $\hat{\theta}_{\text{rc}}$ and $%
\tilde{\theta}_{\text{rc}}$ parallels that between $\hat{\theta}$ and $%
\tilde{\theta}$. Specifically, this involves substituting $\Delta Y$ with $%
\left( \frac{\mathfrak{T}_{\mathcal{T}}}{\hat{\lambda}_{\mathcal{T}}}-\frac{%
\mathfrak{T}_{\mathcal{T}-1}}{\hat{\lambda}_{\mathcal{T}-1}}\right) Y$, $%
m_{\Delta }(X)$ with $\left( \frac{\mu _{G\neq 1,\mathcal{T}}(X)}{\hat{%
\lambda}_{\mathcal{T}}}-\frac{\mu _{G\neq 1,\mathcal{T}-1}(X)}{\hat{\lambda}%
_{\mathcal{T}-1}}\right) $, and $\hat{m}_{\Delta }(X)$ with $\left( \frac{%
\hat{\mu}_{G\neq 1,\mathcal{T}}(X)}{\hat{\lambda}_{\mathcal{T}}}-\frac{\hat{%
\mu}_{G\neq 1,\mathcal{T}-1}(X)}{\hat{\lambda}_{\mathcal{T}-1}}\right) $. In
particular, we obtain that 
\begin{align*}
& \phi _{\text{rc}}(S_{i}^{\text{rc}};\hat{\mu}_{G\neq 1,\mathcal{T}},\hat{%
\mu}_{G\neq 1,\mathcal{T}-1},\hat{p},\hat{w};\hat{\pi}_{1},\hat{\lambda}_{%
\mathcal{T}},\hat{\lambda}_{\mathcal{T}-1})\hat{\pi}_{1}-\phi _{\text{rc}%
}(S_{i}^{\text{rc}};\mu _{G\neq 1,\mathcal{T}},\mu _{G\neq 1,\mathcal{T}%
-1},p,w;\hat{\pi}_{1},\hat{\lambda}_{\mathcal{T}},\hat{\lambda}_{\mathcal{T}%
-1})\hat{\pi}_{1} \\
=& \sum_{t=\mathcal{T}-1}^{\mathcal{T}}(-1)^{\mathcal{T}-t}\left(
\sum_{g=2}^{N_{G}+1}(E_{\text{rc},1g,t}+E_{\text{rc},2g,t}+E_{\text{rc}%
,4g,t}+E_{\text{rc},5g,t})+E_{\text{rc},3,t}\right) ,
\end{align*}%
where 
\begin{align*}
E_{\text{rc},1g,t}(S^{\text{rc}};\hat{\mu}_{G\neq 1,t},\hat{p},\hat{w};\hat{%
\lambda}_{t})& \equiv -\hat{w}_{g}(X)\mathcal{G}_{g}\left( \frac{\hat{p}%
_{1}(X)}{\hat{p}_{g}(X)}-\frac{p_{1}(X)}{p_{g}(X)}\right) \left( \frac{%
\mathfrak{T}_{t}Y-\mu _{G\neq 1,t}(X)}{\hat{\lambda}_{t}}\right) , \\
E_{\text{rc},2g,t}(S^{\text{rc}};\hat{\mu}_{G\neq 1,t},\hat{p},\hat{w};\hat{%
\lambda}_{t})& \equiv -(\hat{w}_{g}(X)-w_{g}(X))\mathcal{G}_{g}\frac{p_{1}(X)%
}{p_{g}(X)}\left( \frac{\mathfrak{T}_{t}Y-\mu _{G\neq 1,t}(X)}{\hat{\lambda}%
_{t}}\right) , \\
E_{\text{rc},3,t}(S^{\text{rc}};\hat{\mu}_{G\neq 1,t},\hat{p},\hat{w};\hat{%
\lambda}_{t})& \equiv -\left( \mathcal{G}_{1}-\sum_{g=2}^{N_{G}+1}w_{g}(X)%
\mathcal{G}_{g}\frac{p_{1}(X)}{p_{g}(X)}\right) \left( \frac{\hat{\mu}%
_{G\neq 1,t}(X)-\mu _{G\neq 1,t}(X)}{\hat{\lambda}_{t}}\right) , \\
E_{\text{rc},4g,t}(S^{\text{rc}};\hat{\mu}_{G\neq 1,t},\hat{p},\hat{w};\hat{%
\lambda}_{t})& \equiv \mathcal{G}_{g}\frac{p_{1}(X)}{p_{g}(X)}(\hat{w}%
_{g}(X)-w_{g}(X))\left( \frac{\hat{\mu}_{G\neq 1,t}(X)-\mu _{G\neq 1,t}(X)}{%
\hat{\lambda}_{t}}\right) , \\
E_{\text{rc},5g,t}(S^{\text{rc}};\hat{\mu}_{G\neq 1,t},\hat{p},\hat{w};\hat{%
\lambda}_{t})& \equiv \mathcal{G}_{g}\left( \frac{\hat{p}_{1}(X)}{\hat{p}%
_{g}(X)}-\frac{p_{1}(X)}{p_{g}(X)}\right) \hat{w}_{g}(X)\left( \frac{\hat{\mu%
}_{G\neq 1,t}(X)-\mu _{G\neq 1,t}(X)}{\hat{\lambda}_{t}}\right) .
\end{align*}%
We show that the sample average of the right-hand side is $o_{p}(n^{-1/2})$.
Each $\hat{\lambda}_{t}$ in the denominator can be replaced with the true $%
\lambda _{t}$, introducing an error of only $o_{p}(n^{-1/2})$. This follows
from condition (i), the consistency of the nuisance estimators, and the fact
that $1/\hat{\lambda}_{t}=1/\lambda _{t}+O_{p}(n^{-1/2})$, which is a
consequence of the central limit theorem and the assumption that $\lambda
_{t}>0$. After this replacement, write $\Delta Y_{\text{rc}}\equiv (%
\mathfrak{T}_{\mathcal{T}}/\lambda _{\mathcal{T}}-\mathfrak{T}_{\mathcal{T}%
-1}/\lambda _{\mathcal{T}-1})Y$. Then $\sum_{t=\mathcal{T}-1}^{\mathcal{T}%
}(-1)^{\mathcal{T}-t}E_{\text{rc},jg,t}$ is the term $E_{jg}$ in the proof
of Theorem \ref{thm:asymp-dist-1} with $\Delta Y$ replaced by $\Delta Y_{%
\text{rc}}$ and $m_{\Delta }$ as defined in Section \ref{sec:rc}, and
similarly for $E_{\text{rc},3,t}$. The two properties used in that proof
continue to hold. First, $\mathbb{E}[(\Delta Y_{\text{rc}}-m_{\Delta
}(X))^{2}|X,G=g]$ is bounded by condition (i), since $(a-b)^{2}\leq
2a^{2}+2b^{2}$. Second, for $g\geq 2$, 
\begin{equation*}
\mathbb{E}[\mathcal{G}_{g}(\Delta Y_{\text{rc}}-m_{\Delta
}(X))|X]=p_{g}(X)\left( m_{g,\mathcal{T}}(X)-m_{g,\mathcal{T}%
-1}(X)-m_{\Delta }(X)\right) =0
\end{equation*}%
under Assumptions \ref{asm:time-invariance} and \ref{asm:PT}. For a fixed $t$%
, $E_{\text{rc},1g,t}$ and $E_{\text{rc},2g,t}$ do not have mean zero; only
their differences across the two periods do. The remaining steps then follow
directly from the proof of Theorem \ref{thm:asymp-dist-1}, yielding $\hat{%
\theta}_{\text{rc}}=\tilde{\theta}_{\text{rc}}+o_{p}(n^{-1/2})$. We apply
the delta method to obtain the influence function for $\hat{\theta}_{\text{rc%
}}$. Notice that $\tilde{\theta}_{\text{rc}}$ is a combination of five
sample averages: 
\begin{equation*}
\tilde{\theta}_{\text{rc}}=\frac{\bar{\varphi}_{\mathcal{T}}}{\hat{\pi}_{1}%
\hat{\lambda}_{\mathcal{T}}}-\frac{\bar{\varphi}_{\mathcal{T}-1}}{\hat{\pi}%
_{1}\hat{\lambda}_{\mathcal{T}-1}}\equiv h\left( \left[ \hat{\pi}_{1},\hat{%
\lambda}_{\mathcal{T}},\hat{\lambda}_{\mathcal{T}-1},\bar{\varphi}_{\mathcal{%
T}},\bar{\varphi}_{\mathcal{T}-1}\right] ^{\prime }\right) ,
\end{equation*}%
where $h\left( z\right) =h(\left[ z_{1},z_{2},z_{3},z_{4},z_{5}\right]
^{\prime })\equiv \frac{z_{4}}{z_{1}z_{2}}-\frac{z_{5}}{z_{1}z_{3}}$, and $%
\bar{\varphi}_{\mathcal{T}}$ and $\bar{\varphi}_{\mathcal{T}-1}$ denote the
respective sample averages of 
\begin{align*}
\varphi _{\mathcal{T}}& \equiv \left( \mathcal{G}_{1}-\sum_{g=2}^{N_{G}+1}%
\mathcal{G}_{g}w_{g}(X)\,\frac{p_{1}(X)}{p_{g}(X)}\right) \left( \mathfrak{T}%
_{\mathcal{T}}Y-\mu _{G\neq 1,\mathcal{T}}(X)\right) , \\
\varphi _{\mathcal{T}-1}& \equiv \left( \mathcal{G}_{1}-\sum_{g=2}^{N_{G}+1}%
\mathcal{G}_{g}w_{g}(X)\,\frac{p_{1}(X)}{p_{g}(X)}\right) \left( \mathfrak{T}%
_{\mathcal{T}-1}Y-\mu _{G\neq 1,\mathcal{T}-1}(X)\right) .
\end{align*}%
According to the central limit theorem, we have that 
\begin{align*}
& \sqrt{n}\left\{ 
\begin{bmatrix}
\hat{\pi}_{1},\hat{\lambda}_{\mathcal{T}},\hat{\lambda}_{\mathcal{T}-1},\bar{%
\varphi}_{\mathcal{T}},\bar{\varphi}_{\mathcal{T}-1}%
\end{bmatrix}%
^{\prime }-%
\begin{bmatrix}
\pi _{1},\lambda _{\mathcal{T}},\lambda _{\mathcal{T}-1},\mathbb{E}\left(
\varphi _{\mathcal{T}}\right) ,\mathbb{E}\left( \varphi _{\mathcal{T}%
-1}\right)%
\end{bmatrix}%
^{\prime }\right\} \\
\overset{d}{\rightarrow }& N\left( 0,\mathrm{var}%
\begin{bmatrix}
\left( \mathcal{G}_{1},\mathfrak{T}_{\mathcal{T}},\mathfrak{T}_{\mathcal{T}%
-1},\varphi _{\mathcal{T}},\varphi _{\mathcal{T}-1}\right) ^{\prime }%
\end{bmatrix}%
\right) .
\end{align*}%
Let $z_{0}=[\pi _{1},\lambda _{\mathcal{T}},\lambda _{\mathcal{T}-1},\mathbb{%
E}\left( \varphi _{\mathcal{T}}\right) ,\mathbb{E}\left( \varphi _{\mathcal{T%
}-1}\right) ]^{\prime }.$ We can calculate the gradient of the function $h$
as 
\begin{align*}
\frac{\partial h}{\partial z_{1}}|_{z=z_{0}}& =-\frac{1}{\pi _{1}}\left( 
\frac{\mathbb{E}[\varphi _{\mathcal{T}}]}{\pi _{1}\lambda _{\mathcal{T}}}-%
\frac{\mathbb{E}\left( \varphi _{\mathcal{T}-1}\right) }{\pi _{1}\lambda _{%
\mathcal{T}-1}}\right) =-\frac{\theta }{\pi _{1}},\quad \frac{\partial h}{%
\partial z_{2}}|_{z=z_{0}}=-\frac{1}{\lambda _{\mathcal{T}}}\left( \frac{%
\mathbb{E}[\varphi _{\mathcal{T}}]}{\pi _{1}\lambda _{\mathcal{T}}}\right) ,
\\
\frac{\partial h}{\partial z_{3}}|_{z=z_{0}}& =\frac{1}{\lambda _{\mathcal{T}%
-1}}\left( \frac{\mathbb{E}\left( \varphi _{\mathcal{T}-1}\right) }{\pi
_{1}\lambda _{\mathcal{T}-1}}\right) ,\quad \frac{\partial h}{\partial z_{4}}%
|_{z=z_{0}}=\frac{1}{\pi _{1}\lambda _{\mathcal{T}}},\quad \frac{\partial h}{%
\partial z_{5}}|_{z=z_{0}}=-\frac{1}{\pi _{1}\lambda _{\mathcal{T}-1}}.
\end{align*}%
Using the delta method, we have: $\sqrt{n}(\tilde{\theta}_{\text{rc}}-\theta
)\overset{d}{\rightarrow }N(0,V_{\mathrm{PT}}^{\text{rc}})$, where $V_{%
\mathrm{PT}}^{\text{rc}}$ is defined as 
\begin{align*}
V_{\mathrm{PT}}^{\text{rc}}& \equiv \mathrm{var}\left[ \left( \mathcal{G}%
_{1},\mathfrak{T}_{\mathcal{T}},\mathfrak{T}_{\mathcal{T}-1},\varphi _{%
\mathcal{T}},\varphi _{\mathcal{T}-1}\right) \left( \frac{\partial h}{%
\partial z}|_{z=z_{0}}\right) \right] \\
& =\mathrm{var}\left[ -\frac{\theta }{\pi _{1}}\mathcal{G}_{1}-\frac{1}{%
\lambda _{\mathcal{T}}}\left( \frac{\mathbb{E}\left( \varphi _{\mathcal{T}%
}\right) }{\pi _{1}\lambda _{\mathcal{T}}}\right) \mathfrak{T}_{\mathcal{T}}+%
\frac{1}{\lambda _{\mathcal{T}-1}}\left( \frac{\mathbb{E}\left( \varphi _{%
\mathcal{T}-1}\right) }{\pi _{1}\lambda _{\mathcal{T}-1}}\right) \mathfrak{T}%
_{\mathcal{T}-1}+\frac{\varphi _{\mathcal{T}}}{\pi _{1}\lambda _{\mathcal{T}}%
}-\frac{\varphi _{\mathcal{T}-1}}{\pi _{1}\lambda _{\mathcal{T}-1}}\right] .
\end{align*}%
It follows from Corollary \ref{thm:rc-identification} that 
\begin{equation*}
\frac{\mathbb{E}[\varphi _{\mathcal{T}}]}{\pi _{1}\lambda _{\mathcal{T}}}-%
\frac{\mathbb{E}[\varphi _{\mathcal{T}-1}]}{\pi _{1}\lambda _{\mathcal{T}-1}}%
=\theta .
\end{equation*}%
The influence function of $\tilde{\theta}_{\text{rc}}$, and hence $\hat{%
\theta}_{\text{rc}}$, obtained via the delta method, is then given by 
\begin{align}
\psi _{\mathrm{PT}}^{\text{rc}}(S^{\text{rc}})& =-\frac{\theta }{\pi _{1}}%
\mathcal{G}_{1}-\frac{\mathfrak{T}_{\mathcal{T}}-\lambda _{\mathcal{T}}}{%
\lambda _{\mathcal{T}}}\left( \frac{\mathbb{E}\left( \varphi _{\mathcal{T}%
-1}\right) }{\pi _{1}\lambda _{\mathcal{T}-1}}+\theta \right) +\frac{%
\mathfrak{T}_{\mathcal{T}-1}-\lambda _{\mathcal{T}-1}}{\lambda _{\mathcal{T}%
-1}}\left( \frac{\mathbb{E}\left( \varphi _{\mathcal{T}-1}\right) }{\pi
_{1}\lambda _{\mathcal{T}-1}}\right)  \notag \\
& +\frac{1}{\pi _{1}}\left( \mathcal{G}_{1}-\sum_{g=2}^{N_{G}+1}\mathcal{G}%
_{g}w_{g}(X)\frac{p_{1}(X)}{p_{g}(X)}\right) \left( \frac{\mathfrak{T}_{%
\mathcal{T}}Y-\mu _{G\neq 1,\mathcal{T}}(X)}{\lambda _{\mathcal{T}}}-\frac{%
\mathfrak{T}_{\mathcal{T}-1}Y-\mu _{G\neq 1,\mathcal{T}-1}(X)}{\lambda _{%
\mathcal{T}-1}}\right) .  \label{eqn:IF-theta-rc-PT}
\end{align}%
Obviously, $V_{\mathrm{PT}}^{\text{rc}}=\mathbb{E}[\psi _{\mathrm{PT}}^{%
\text{rc}}(S^{\text{rc}})^{2}].$ 
\end{proof}

\vspace{0.5cm}

\begin{proof}[Proof of Theorem \protect\ref{thm:rc-SC}]
Similar to the proof of Theorem \ref{thm:asymp-dist-2}, we can show that the
adjustment term in the influence function of $\hat{\theta}_{\text{rc}}$ will
be generated by the following terms: 
\begin{equation*}
\sum_{g=2}^{N_{G}+1}(\hat{E}_{\text{rc},1g,\mathcal{T}}-\hat{E}_{\text{rc}%
,1g,\mathcal{T}-1}+\hat{E}_{\text{rc},2g,\mathcal{T}}-\hat{E}_{\text{rc},2g,%
\mathcal{T}-1}),
\end{equation*}%
where $\hat{E}_{\text{rc},jg,t}\equiv \frac{1}{n}\sum_{i=1}^{n}E_{\text{rc}%
,jg,t}(S_{i}^{\text{rc}};\hat{\mu}_{G\neq 1,t},\hat{p},\hat{w};\hat{\lambda}%
_{t})$ and 
\begin{align*}
E_{\text{rc},1g,t}(S^{\text{rc}};\hat{\mu}_{G\neq 1,t},\hat{p},\hat{w};\hat{%
\lambda}_{t})& \equiv -w_{g}(X)\mathcal{G}_{g}\left( \hat{r}%
_{1,g}(X)-r_{1,g}(X)\right) \left( \frac{\mathfrak{T}_{t}Y-\mu _{G\neq
1,t}(X)}{\hat{\lambda}_{t}}\right) , \\
E_{\text{rc},2g,t}(S^{\text{rc}};\hat{\mu}_{G\neq 1,t},\hat{p},\hat{w};\hat{%
\lambda}_{t})& \equiv -(\hat{w}_{g}(X)-w_{g}(X))\mathcal{G}%
_{g}r_{1,g}(X)\left( \frac{\mathfrak{T}_{t}Y-\mu _{G\neq 1,t}(X)}{\hat{%
\lambda}_{t}}\right) .
\end{align*}%
First, similar to the proof of Theorem \ref{thm:rc-PT}, we can replace each $%
\hat{\lambda}_{t}$ in the denominator with the true $\lambda _{t}$,
introducing an error of only $o_{p}(n^{-1/2})$. Then we follow the proof of
Theorem \ref{thm:asymp-dist-2} to derive the influence functions for the
terms $\hat{E}_{\text{rc},jg,t}$. For $\hat{E}_{\text{rc},1g,t}$, we have 
\begin{equation*}
\hat{E}_{\text{rc},1g,t}=-\frac{1}{n}\sum_{i=1}^{n}w_{g}(X_{i})(\mathcal{G}%
_{1i}-r_{1,g}(X_{i})\mathcal{G}_{gi})\frac{\mu _{g,t}(X_{i})-\mu _{G\neq
1,t}(X_{i})}{\lambda _{t}}+o_{p}(n^{-1/2}),
\end{equation*}%
where $\mu _{g,t}(x)\equiv \mathbb{E}[\mathfrak{T}_{t}Y|G=g,X=x]$ and so the
influence function is 
\begin{equation*}
-w_{g}(X)(\mathcal{G}_{1}-r_{1,g}(X)\mathcal{G}_{g})\frac{\mu _{g,t}(X)-\mu
_{G\neq 1,t}(X)}{\lambda _{t}}.
\end{equation*}

Similarly, the influence function for $\sum_{g=2}^{N_{G}+1}\hat{E}_{\text{rc}%
,2g,t}$ is 
\begin{equation*}
-p_{1}(X)v_{\text{rc},\Delta ,t}(X)^{\prime }\left( 
\begin{array}{c}
P_{\text{rc}}(S^{\text{rc}}) \\ 
-\mathbf{1}_{N_{G}-1}^{\prime }P_{\text{rc}}(S^{\text{rc}})%
\end{array}%
\right) ,
\end{equation*}%
where 
\begin{align*}
v_{\text{rc},\Delta ,t}(X)& \equiv \frac{1}{\lambda _{t}}%
\begin{pmatrix}
\mu _{2,t}(X)-\mu _{G\neq 1,t}(X) \\ 
\vdots \\ 
\mu _{N_{G}+1,t}(X)-\mu _{G\neq 1,t}(X)%
\end{pmatrix}%
, \\
P_{\text{rc}}(S^{\text{rc}})& \equiv (M_{\text{rc}}(X)^{\prime }M_{\text{rc}%
}(X))^{-1}M_{\text{rc}}(X)^{\prime }a_{\text{rc},1}(S^{\text{rc}}) \\
& +(M_{\text{rc}}(X)^{\prime }M_{\text{rc}}(X))^{-1}A_{\text{rc}}(S^{\text{rc%
}})^{\prime }m_{\text{rc},1}(X) \\
& -\big(\lbrack M_{\text{rc}}(X)^{\prime }M_{\text{rc}}(X)]^{-1}\left[ A_{%
\text{rc}}(S^{\text{rc}})^{\prime }M_{\text{rc}}(X)+M_{\text{rc}}(X)^{\prime
}A_{\text{rc}}(S^{\text{rc}})\right] \\
& \quad \times (M_{\text{rc}}(X)^{\prime }M_{\text{rc}}(X))^{-1}M_{\text{rc}%
}(X)^{\prime }m_{\text{rc},1}(X)\big),
\end{align*}%
with $M_{\text{rc}}(X)$ and $m_{\text{rc},1}(X)$ as defined in Section \ref%
{sec:rc}, and $A_{\text{rc}}$ and $a_{\text{rc},1}$ defined as 
\begin{align*}
A_{\text{rc}}(S^{\text{rc}})& \equiv 
\begin{pmatrix}
a_{\text{rc},-1,1}(S^{\text{rc}})^{\prime } \\ 
\vdots \\ 
a_{\text{rc},-1,\mathcal{T}-1}(S^{\text{rc}})^{\prime }%
\end{pmatrix}%
, \\
a_{\text{rc},-1,t}(S^{\text{rc}})& \equiv 
\begin{pmatrix}
\frac{\mathcal{G}_{2}(\mathfrak{T}_{t}Y-\mu _{2,t}(X))}{p_{2}(X)}-\frac{%
\mathcal{G}_{N_{G}+1}(\mathfrak{T}_{t}Y-\mu _{N_{G}+1,t}(X))}{p_{N_{G}+1}(X)}
\\ 
\vdots \\ 
\frac{\mathcal{G}_{N_{G}}(\mathfrak{T}_{t}Y-\mu _{N_{G},t}(X))}{p_{N_{G}}(X)}%
-\frac{\mathcal{G}_{N_{G}+1}(\mathfrak{T}_{t}Y-\mu _{N_{G}+1,t}(X))}{%
p_{N_{G}+1}(X)}%
\end{pmatrix}%
, \\
a_{\text{rc},1}(S^{\text{rc}})& \equiv 
\begin{pmatrix}
\frac{\mathcal{G}_{1}(\mathfrak{T}_{1}Y-\mu _{1,1}(X))}{p_{1}(X)}-\frac{%
\mathcal{G}_{N_{G}+1}(\mathfrak{T}_{1}Y-\mu _{N_{G}+1,1}(X))}{p_{N_{G}+1}(X)}
\\ 
\vdots \\ 
\frac{\mathcal{G}_{1}(\mathfrak{T}_{\mathcal{T}-1}Y-\mu _{1,\mathcal{T}%
-1}(X))}{p_{1}(X)}-\frac{\mathcal{G}_{N_{G}+1}(\mathfrak{T}_{\mathcal{T}%
-1}Y-\mu _{N_{G}+1,\mathcal{T}-1}(X))}{p_{N_{G}+1}(X)}%
\end{pmatrix}%
.
\end{align*}%
Therefore, the influence function for $\hat{\theta}_{\text{rc}}$ is 
\begin{equation}
\psi _{\mathrm{SC}}^{\text{rc}}(S^{\text{rc}})=\psi _{\mathrm{PT}}^{\text{rc}%
}(S^{\text{rc}})+\psi _{1}^{\text{rc}}(S^{\text{rc}})+\psi _{2}^{\text{rc}%
}(S^{\text{rc}}),  \label{eqn:IF-theta-rc-SC}
\end{equation}%
where the two adjustment terms are 
\begin{align*}
\psi _{1}^{\text{rc}}(S^{\text{rc}})& \equiv -\frac{1}{\pi _{1}}%
\sum_{g=2}^{N_{G}+1}w_{g}(X)(\mathcal{G}_{1}-r_{1,g}(X)\mathcal{G}_{g}) \\
& \qquad \times \left( \frac{\mu _{g,\mathcal{T}}(X)-\mu _{G\neq 1,\mathcal{T%
}}(X)}{\lambda _{\mathcal{T}}}-\frac{\mu _{g,\mathcal{T}-1}(X)-\mu _{G\neq 1,%
\mathcal{T}-1}(X)}{\lambda _{\mathcal{T}-1}}\right) \\[0.04in]
& =-\frac{p_{1}(X)}{\pi _{1}}(v_{\text{rc},\Delta ,\mathcal{T}}(X)-v_{\text{%
rc},\Delta ,\mathcal{T}-1}(X))^{\prime }Q(S^{\text{rc}}), \\[0.04in]
\psi _{2}^{\text{rc}}(S^{\text{rc}})& \equiv -\frac{p_{1}(X)}{\pi _{1}}(v_{%
\text{rc},\Delta ,\mathcal{T}}(X)-v_{\text{rc},\Delta ,\mathcal{T}%
-1}(X))^{\prime }\left( 
\begin{array}{c}
P_{\text{rc}}(S^{\text{rc}}) \\ 
-\mathbf{1}_{N_{G}-1}^{\prime }P_{\text{rc}}(S^{\text{rc}})%
\end{array}%
\right) ,
\end{align*}%
and $Q(S^{\text{rc}})$ is the $N_{G}$-vector 
\begin{equation*}
Q(S^{\text{rc}})\equiv \left( w_{2}(X)\left( \frac{\mathcal{G}_{1}}{p_{1}(X)}%
-\frac{\mathcal{G}_{2}}{p_{2}(X)}\right) ,\ldots ,w_{N_{G}+1}(X)\left( \frac{%
\mathcal{G}_{1}}{p_{1}(X)}-\frac{\mathcal{G}_{N_{G}+1}}{p_{N_{G}+1}(X)}%
\right) \right) ^{\prime }.
\end{equation*}%
Finally, $V_{\mathrm{SC}}^{\text{rc}}=\mathbb{E}[\psi _{\mathrm{SC}}^{\text{%
rc}}(S^{\text{rc}})^{2}].$
\end{proof}

\vspace{0.5cm}

\begin{proof}[Proof of Corollary \protect\ref{cor:rc-bootstrap}]
The proof follows the same arguments as those for Theorem \ref{thm:bootstrap}%
. Specifically, we derive the asymptotic linear representations for $\sqrt{n}%
(\hat{\theta}_{\text{rc}}^{\ast }-\theta )$ and $\sqrt{n}(\hat{\theta}_{%
\text{rc}}-\theta )$, take their difference, and apply Lemma \ref{lm:clt}.
For brevity, the details are omitted.
\end{proof}

\vspace{0.5cm}

\begin{proof}[Proof of Corollary \protect\ref{cor:staggered}]
This result mirrors Theorem \ref{thm:identification}. To see this, simply
treat group $g$ as group 1 and the groups in $\mathcal{D}_{g,\bar{e}}$ as
groups $2, \dots, \tilde{N}_G+1$, where $\tilde{N}_G \equiv |\mathcal{D}_{g,%
\bar{e}}|$ plays the role of $N_G$ in Theorem \ref{thm:identification}.
Assumptions \ref{asm:NA-staggered}, \ref{asm:PT-staggered}, and \ref%
{asm:SC-staggered} correspond to Assumptions \ref{asm:NA}, \ref{asm:PT}, and %
\ref{asm:SC}, respectively. Hence, the conclusion is identical to Theorem %
\ref{thm:identification}.
\end{proof}

\section{Comparison between $V_{\mathrm{PT}}$ and $V_{\mathrm{pool}}$}

\label{sec:variance-comparison}

We compare the two asymptotic variances $V_{\mathrm{PT}}$ and $V_{\mathrm{%
pool}}$, defined in \eqref{eqn:V1} and \eqref{eqn:drdid-var}, respectively.
For illustration, suppose that there are no conditioning covariates $X$. The
first two terms in $V_{\mathrm{PT}}$ and $V_{\mathrm{pool}}$ are identical,
so the two variances differ only through the third term, which is the
weighted control-group component. Specifically, we can write 
\begin{align*}
V_{\mathrm{PT}} & = \Omega_{\mathrm{common}} + \Omega(\omega), \\
V_{\mathrm{pool}} & = \Omega_{\mathrm{common}} + \Omega(\omega_{\mathrm{pool}%
}),
\end{align*}
where 
\begin{align*}
\Omega_{\mathrm{common}} & = \mathbb{E}\left[ \mathcal{G}_{1}(\Delta
Y-m_{1,\Delta})^{2} +\mathcal{G}_{1}(m_{1,\Delta}-m_{\Delta}-\theta)^{2} %
\right]/\pi_{1}^{2}, \\
\Omega(\upsilon) & = \sum_{g=2}^{N_G+1} \upsilon_g^2 \frac{\sigma_g^2}{p_g},
\quad \text{for any weight vector } \upsilon=(\upsilon_2,\ldots,%
\upsilon_{N_G+1}).
\end{align*}
Here, $\sigma_g^2=\mathbb{E}[(\Delta Y-m_{\Delta})^2|G=g]$ denotes the
conditional variance of $\Delta Y$ in group $g$. The vector $\omega$ denotes
the probability limit of the estimated weights $\hat{w}$, as in Theorem \ref%
{thm:asymp-dist-1}. Finally, $\omega_{\mathrm{pool}}=(p_g/p_{-1},2\leq g\leq
N_G+1)$, with $p_{-1}=\sum_{g=2}^{N_G+1}p_g$, denotes the pooling weight
used in $V_{\mathrm{pool}}$.

In the homoskedastic case, where $\sigma_g=\sigma$ for all $g\geq 2$, we
have $\Omega(\upsilon)=\sigma^2\sum_{g=2}^{N_G+1}\upsilon_g^2/p_g$. By the
Cauchy--Schwarz inequality, for any weights satisfying $\sum_{g=2}^{N_G+1}%
\upsilon_g=1$, 
\begin{align*}
1 = \left(\sum_{g=2}^{N_G+1}\upsilon_g\right)^2 \leq \left(\sum_{g=2}^{N_G+1}%
\frac{\upsilon_g^2}{p_g}\right) \left(\sum_{g=2}^{N_G+1}p_g\right).
\end{align*}
Therefore, $\Omega(\upsilon)\geq \sigma^2/p_{-1}$, with equality if and only
if $\upsilon_g=p_g/p_{-1}$ for every $g\geq 2$. Hence, in the homoskedastic
case, the pooling weight $\omega_{\mathrm{pool}}$ minimizes the asymptotic
variance, and $V_{\mathrm{pool}}$ is weakly smaller than $V_{\mathrm{PT}}$.

In the heteroskedastic case, where the variances $\sigma_g^2$ differ across
control groups, the pooling weight no longer generally minimizes $%
\Omega(\upsilon)$. Instead, minimizing $\Omega(\upsilon)=\sum_{g=2}^{N_G+1}%
\upsilon_g^2\sigma_g^2/p_g$ subject to $\sum_{g=2}^{N_G+1}\upsilon_g=1$
yields the variance-minimizing weight 
\begin{equation*}
\upsilon_g^{*} = \frac{p_g/\sigma_g^2}{\sum_{j=2}^{N_G+1}p_j/\sigma_j^2},
\qquad g=2,\ldots,N_G+1.
\end{equation*}
Thus, groups with larger sample shares and smaller conditional variances
receive larger weights. Moreover, for any weight vector $\upsilon$
satisfying $\sum_{g=2}^{N_G+1}\upsilon_g=1$, we can write 
\begin{equation*}
\Omega(\upsilon) = \Omega(\upsilon^{*}) + \sum_{g=2}^{N_G+1}\frac{\sigma_g^2%
}{p_g}(\upsilon_g-\upsilon_g^{*})^2.
\end{equation*}
Therefore, in the heteroskedastic case, the comparison between $V_{\mathrm{PT%
}}$ and $V_{\mathrm{pool}}$ depends on whether $\omega$ or $\omega_{\mathrm{%
pool}}$ is closer to the variance-minimizing weight $\upsilon^{*}$, where
closeness is measured by the weighted quadratic distance $%
\sum_{g=2}^{N_G+1}(\sigma_g^2/p_g)(\upsilon_g-\upsilon_g^{*})^2$. In
particular, $V_{\mathrm{PT}}<V_{\mathrm{pool}}$ if and only if 
\begin{equation*}
\sum_{g=2}^{N_G+1}\frac{\sigma_g^2}{p_g} (\omega_g-\upsilon_g^{*})^2 <
\sum_{g=2}^{N_G+1}\frac{\sigma_g^2}{p_g} (\omega_{\mathrm{pool}%
,g}-\upsilon_g^{*})^2.
\end{equation*}
When $\omega$ places relatively more weight on control groups with lower
conditional variance, it may be closer to $\upsilon^{*}$ than the pooling
weights are, in which case $V_{\mathrm{PT}}$ can be smaller than $V_{\mathrm{%
pool}}$. This case can arise when the weight-estimation criterion assigns
larger weights to control groups that better approximate the treated group's
conditional mean dynamics and, at the same time, have more stable outcome
changes and hence smaller variances.

As a simple numerical illustration, consider the case with two control
groups of equal size, $p_2=p_3$, so that the pooling weight is $\omega_{%
\mathrm{pool},2}=1/2$. Let $\upsilon_2$ denote the generic weight on group
2. Suppose group 2 is noisier than group 3, with variance ratio $%
\sigma_2^2/\sigma_3^2 \geq 1$. Then the variance-minimizing weight is $%
\upsilon_2^{*}=\sigma_3^2/(\sigma_2^2+\sigma_3^2)$, and $\Omega(\upsilon)<%
\Omega(\omega_{\mathrm{pool}})$ if and only if $\upsilon_2$ is closer to $%
\upsilon_2^{*}$ than $1/2$ is. Table \ref{tab:variance-comparison} reports
the resulting ranges for several variance ratios.

\begin{table}[h]
\centering
\begin{tabular}{c c c}
\toprule
Variance ratio \(\sigma_2^2/\sigma_3^2\) 
& Variance-minimizing weight \(\upsilon_2^{*}\) 
& Range of \(\upsilon_2\) such that \(\Omega(\upsilon)<\Omega(\omega_{\mathrm{pool}})\) \\
\midrule
\(1.5\) & \(2/5\) & \((0.3,0.5)\) \\
\(2\) & \(1/3\) & \((1/6,1/2)\) \\
\(3\) & \(1/4\) & \((0,1/2)\) \\
\bottomrule
\end{tabular}
\caption{Variance comparison}
\label{tab:variance-comparison}
\caption*{\footnotesize The table reports, for the case with two control groups, the range of weights $\upsilon_2$ assigned to group 2 for which the resulting asymptotic variance is smaller than that obtained using the pooling weights.}
\end{table}%

The table shows that, as heteroskedasticity becomes more severe, the range
of weights that are more efficient than the pooling weights becomes
substantially wider. Therefore, in more heteroskedastic settings, it becomes
more likely for $V_{\mathrm{PT}}$ to be smaller than $V_{\mathrm{pool}}$.

\section{Additional empirical results}

\label{sec:additional-empirical}

\subsection{Bandwidth sensitivity}

\label{subsec:bandwidth-sensitivity}

We assess bandwidth sensitivity by jointly rescaling the two local-linear
nuisance bandwidths by $c/2.5$ relative to the baseline empirical
bandwidths, for $c\in\{1,1.5,2,2.5,3,3.5\}$, where $c=2.5$ corresponds to
the baseline specification. The results are reported in Table \ref%
{tab:empirical-bandwidth-sensitivity}. The point estimates are generally
stable for bandwidths near the baseline, with larger fluctuations appearing
only in some smaller subgroups under more extreme bandwidth choices.
Nevertheless, the confidence intervals overlap substantially and contain
zero for every subgroup and bandwidth considered, indicating that the
qualitative empirical conclusions are robust to the bandwidth choice.

\begin{table}[htbp!]
\centering
\small
\setlength{\tabcolsep}{2.8pt}
\begin{tabular}{cccccccccccc}
\toprule
& &
\multicolumn{2}{c}{$c=1$}
& \multicolumn{2}{c}{$c=1.5$}
& \multicolumn{2}{c}{$c=2$}
& \multicolumn{2}{c}{$c=3$}
& \multicolumn{2}{c}{$c=3.5$} \\
\cmidrule(lr){3-4}
\cmidrule(lr){5-6}
\cmidrule(lr){7-8}
\cmidrule(lr){9-10}
\cmidrule(lr){11-12}
Edu. & Child.
& Est. & 95\% CI
& Est. & 95\% CI
& Est. & 95\% CI
& Est. & 95\% CI
& Est. & 95\% CI \\
\midrule
0 & 0
& -0.74 & [-2.58, 1.11]
& -0.60 & [-2.28, 1.08]
& -0.55 & [-2.34, 1.25]
& -0.57 & [-2.17, 1.04]
& -0.52 & [-2.15, 1.11] \\

0 & 1
& 1.14 & [-3.98, 6.26]
& 0.59 & [-5.41, 6.59]
& 1.15 & [-1.18, 3.48]
& 1.04 & [-1.47, 3.55]
& 1.07 & [-1.41, 3.54] \\

0 & 2
& -0.11 & [-1.69, 1.46]
& -0.09 & [-1.69, 1.51]
& -0.07 & [-1.62, 1.47]
& -0.01 & [-1.53, 1.51]
& -0.14 & [-1.62, 1.34] \\

1 & 0
& 0.68 & [-1.49, 2.85]
& 0.59 & [-1.40, 2.58]
& 0.30 & [-1.60, 2.20]
& 0.20 & [-1.48, 1.87]
& 0.01 & [-1.73, 1.75] \\

1 & 1
& 0.95 & [-1.82, 3.72]
& 1.08 & [-1.78, 3.93]
& 0.65 & [-2.11, 3.41]
& 0.19 & [-2.56, 2.94]
& 0.05 & [-2.81, 2.91] \\

1 & 2
& -0.18 & [-2.01, 1.65]
& -0.28 & [-2.04, 1.49]
& -0.25 & [-1.97, 1.47]
& 0.15 & [-1.60, 1.90]
& 0.15 & [-1.59, 1.88] \\

2 & 0
& -0.50 & [-5.08, 4.08]
& 0.02 & [-4.32, 4.36]
& 0.46 & [-3.78, 4.71]
& 0.80 & [-3.27, 4.87]
& 0.76 & [-3.25, 4.77] \\

2 & 1
& -0.17 & [-3.81, 3.48]
& 0.04 & [-3.45, 3.53]
& -0.72 & [-4.44, 2.99]
& -0.49 & [-4.09, 3.11]
& 0.01 & [-3.51, 3.53] \\

2 & 2
& -1.38 & [-3.70, 0.94]
& -1.74 & [-4.01, 0.53]
& -1.44 & [-3.50, 0.62]
& -1.05 & [-3.21, 1.10]
& -1.23 & [-3.23, 0.77] \\
\bottomrule
\end{tabular}
\caption{Bandwidth sensitivity of the empirical estimates}
\label{tab:empirical-bandwidth-sensitivity}
\caption*{\footnotesize
The table reports sensitivity of the empirical estimates to the bandwidth choice. The two empirical undersmoothing bandwidths are jointly multiplied by $c/2.5$. The main specification, $c=2.5$, is omitted here and reported in Table~\ref{tab:empirical-main}. Entries report point estimates and 95\% multiplier bootstrap confidence intervals.}
\end{table}%

\subsection{Donor pool sensitivity}

\label{subsec:donor-sensitivity}

We further examine how the results depend on the choice of donor states by
considering three alternative donor pools: the four states receiving the
largest weights under the CDF-based specification of %
\citet{gunsilius2023distributional}{;} the four-state pool with the smallest
pre-treatment 95th-percentile condition number; and the three states
receiving the largest weights under the quantile-based specification used in
our main analysis, that is, the baseline pool with Utah, its lowest-weighted
donor, dropped. Table \ref{tab:empirical-donor-sensitivity} shows that the
main findings are remarkably similar across these alternatives. The
estimated effects retain the same sign in seven of the nine subgroups, with
the largest differences concentrated in a few smaller, more highly educated
subgroups. More importantly, all specifications continue to yield confidence
intervals that include zero. Thus, changing the donor pool affects some
point estimates but does not alter the substantive conclusion of the
empirical analysis.

\begin{table}[ht!]
\centering
\small
\setlength{\tabcolsep}{3.5pt}
\resizebox{\textwidth}{!}{
\begin{tabular}{cccccccc}
\toprule
& &
\multicolumn{2}{c}{CDF top 4}
& \multicolumn{2}{c}{Best-conditioned 4}
& \multicolumn{2}{c}{Quantile top 3} \\
\cmidrule(lr){3-4}\cmidrule(lr){5-6}\cmidrule(lr){7-8}
Education & Children
& Estimate & 95\% CI
& Estimate & 95\% CI
& Estimate & 95\% CI \\
\midrule
0 & 0 & -0.750 & [-2.456, 0.957] & -1.178 & [-2.779, 0.424] & -1.049 & [-2.621, 0.523] \\
0 & 1 &  0.593 & [-1.630, 2.816] &  1.445 & [-0.754, 3.644] &  1.170 & [-1.093, 3.433] \\
0 & 2 &  0.117 & [-1.388, 1.622] &  0.088 & [-1.536, 1.711] &  0.337 & [-1.196, 1.871] \\
1 & 0 &  0.623 & [-1.280, 2.526] &  0.135 & [-1.744, 2.014] &  0.313 & [-1.467, 2.093] \\
1 & 1 &  0.795 & [-1.631, 3.222] &  0.362 & [-1.985, 2.709] &  1.080 & [-1.688, 3.847] \\
1 & 2 &  0.823 & [-0.938, 2.584] &  0.558 & [-1.174, 2.291] & -0.027 & [-1.722, 1.669] \\
2 & 0 &  0.569 & [-3.267, 4.404] &  0.081 & [-3.951, 4.114] &  0.093 & [-3.825, 4.012] \\
2 & 1 &  1.199 & [-2.875, 5.273] &  1.170 & [-2.901, 5.241] & -0.453 & [-4.062, 3.156] \\
2 & 2 & -1.082 & [-3.285, 1.122] & -1.874 & [-4.117, 0.369] & -0.487 & [-2.863, 1.890] \\
\bottomrule
\end{tabular}
}
\caption{Sensitivity to the donor-state selection}
\label{tab:empirical-donor-sensitivity}
\caption*{\footnotesize
The table reports sensitivity of the empirical results to alternative donor-state selections. Entries report point estimates and 95\% multiplier bootstrap confidence intervals. The CDF top 4 are Michigan, Ohio, Maryland, and Virginia. The best-conditioned 4 are Maryland, Michigan, Utah, and Virginia. The quantile top 3 are Virginia, New Hampshire, and Maryland.}
\end{table}%

\subsection{Diagnostic tests for the identification assumptions}

\label{subsec:diagnostic-PT-SC}

We assess the plausibility of the identifying restrictions using the
1998--2002 pre-treatment data. For Assumption \ref{asm:PT}, we compare
Alaska's conditional outcome changes with those of each donor over the four
adjacent pre-treatment trends. For each donor--trend pair, we jointly test
equality across four age regions defined by the quartiles of the empirical
age distribution, yielding 16 tests within each education-children subgroup
and 144 tests in total. We apply the multiple-testing correction of %
\citet{holm1979simple} within each subgroup. Across the 144 tests, only one
rejects at the 5\% level after adjustment.

For Assumption \ref{asm:SC}, we estimate the conditional simplex weights
using 1998--2001 and assess their ability to reproduce Alaska's 2002
conditional mean in a placebo exercise. As in the previous diagnostic, we
average the resulting conditional-mean discrepancies within four age
quartiles and jointly test whether these four moments are zero. After
applying Holm's correction across the nine education-children subgroups,
none of the nine joint tests rejects at the 5\% level.

Overall, these diagnostics reveal little evidence of systematic
pre-treatment violations of either identifying restriction. As with any
pre-treatment diagnostic test, however, nonrejection should not be
interpreted as establishing that the corresponding restriction necessarily
continues to hold in the 2003 post-treatment period.

\section{Details of the simulation design}

\label{sec:simulation-design-details}

This appendix presents the construction of the data-generating processes of
Section \ref{subsec:DGPs-sim} in detail. Each subsection takes the output of
the previous one as its input.

\subsection{Calibration}

\label{subsec:sim-calibration}

The preliminary mean outcome path \eqref{eq:simulation-reference-mean} is
assembled in two stages: the calibrated component defined here, and the
rank-completion term of Section \ref{subsec:rank-completion}. The first
stage defines, for each control group $g$, 
\begin{align}
\tilde{m}_{g,t}(x)=\alpha _{g}+\delta _{t}+\lambda (g,x)^{\prime
}F_{t},\qquad \lambda (g,x)=\mathcal{B}_{g}(1,x,x^{2})^{\prime },\qquad
t=1,\ldots ,\mathcal{T}.  \label{eq:appE-calibrated-component}
\end{align}%
The four sets of constants $\alpha _{g}$, $\delta _{t}$, $F_{t}\in \mathbb{R}%
^{2\times 1}$ and $\mathcal{B}_{g}\in \mathbb{R}^{2\times 3}$ are calibrated
from the CPS extract used in Section \ref{sec:empirical}. We take $\alpha
_{g}$ and $\delta _{t}$ to be the estimated group and time effects, fit a
rank-two factor model to the outcome variation that remains once they are
removed, and regress the estimated loadings on $(1,x,x^{2})$ separately for
each donor.

\subsection{The rank-completion term}

\label{subsec:rank-completion}

The pre-treatment contrasts implied by \eqref{eq:appE-calibrated-component}
are the differences between each donor's calibrated component and the
reference donor's, 
\begin{align*}
[\tilde{M}(x)]_{t,g}=\tilde{m}_{g,t}(x)-\tilde{m}_{N_{G}+1,t}(x)=a_{g}+%
\left( \lambda (g,x)-\lambda (N_{G}+1,x)\right) ^{\prime }F_{t},
\end{align*}%
for $t\leq \mathcal{T}-1$ and $g=2,\ldots ,N_{G}$, in which the time effect $%
\delta _{t}$ cancels and $a_{g}=\alpha _{g}-\alpha _{N_{G}+1}$ is the
intercept contrast of donor $g$. Since $t$ runs over the pre-treatment
periods and $g$ over the non-reference donors, $\tilde{M}(x)$ is $(\mathcal{T%
}-1)\times (N_{G}-1)=5\times 3$, the dimensions of the donor-difference
matrix $M(x)$ of \eqref{eqn:M-m1-def}.

The intercept contrasts $a=(a_{g})_{g=2,\ldots ,N_{G}}$ and the loading
contrasts $d_{k}(x)=(\lambda _{k}(g,x)-\lambda _{k}(N_{G}+1,x))_{g=2,\ldots
,N_{G}}$, $k=1,2$, all lie in $\mathbb{R}^{3}$. The contrast matrix
factorizes as 
\begin{align*}
\tilde{M}(x)=\left[ \mathbf{1}\ \ F^{1}\ \ F^{2}\right] L(x),\qquad
L(x)=\left( a,\,d_{1}(x),\,d_{2}(x)\right) ^{\prime }\in \mathbb{R}^{3\times
3},
\end{align*}%
where $F^{k}\in \mathbb{R}^{5}$ is the pre-treatment path of the $k$th
factor and $\mathbf{1}$ is the vector of ones. The rows of $L(x)$ are the
intercept contrast and the two loading contrasts; its columns are indexed by
the three non-reference donors.

In the empirical calibration, $[\mathbf{1}\ F^{1}\ F^{2}]$ is well
conditioned but $L(x)$ is close to rank one at every covariate value. The
smallest singular value of $\tilde{M}(x)$ is therefore close to zero, and
condition (iii) of Theorem \ref{thm:asymp-dist-2} fails.

To restore the rank, let $H\in \mathbb{R}^{5\times 2}$ have orthonormal
columns spanning the orthogonal complement of $\func{span}\{\mathbf{1}%
,F^{1},F^{2}\}$ in $\mathbb{R}^{5}$, and let $C\in \mathbb{R}^{2\times 3}$
have orthonormal rows spanning $a^{\perp }$, the orthogonal complement of $a$
in $\mathbb{R}^{3}$. Define the rank-completion term $R_{g,t}$ by 
\begin{align}
R_{g,t}=c\sum_{k=1,2}C_{k,g}H_{t,k}\ \ (t\leq \mathcal{T}-1),\qquad R_{g,%
\mathcal{T}}=R_{g,\mathcal{T}-1},\qquad R_{N_{G}+1,t}\equiv 0,
\label{eq:appE-completion}
\end{align}%
with scale $c$ calibrated from the empirical singular values. Adding $R_{g,t}
$ to $\tilde{m}_{g,t}(x)$ gives the preliminary mean outcome path %
\eqref{eq:simulation-reference-mean}, and differencing against the reference
donor, for which $R_{N_{G}+1,t}\equiv 0$, carries the same decomposition to
the pre-treatment contrasts: 
\begin{align*}
m_{g,t}^{R}(x)& =\tilde{m}_{g,t}(x)+R_{g,t}, \\
M^{R}(x)& =\tilde{M}(x)+B,\qquad B_{t,g}=R_{g,t},
\end{align*}%
where $B$ has the same $5\times 3$ dimensions as $\tilde{M}(x)$ and $M^{R}(x)
$. The columns of $\tilde{M}(x)$ lie in $\func{span}\{\mathbf{1}%
,F^{1},F^{2}\}$ and those of $B $ in its orthogonal complement, so $\tilde{M}%
(x)^{\prime }B=0$ and 
\begin{align*}
M^{R}(x)^{\prime }M^{R}(x)=\tilde{M}(x)^{\prime }\tilde{M}%
(x)+c^{2}P_{a^{\perp }},\qquad P_{a^{\perp }}=I_{3}-aa^{\prime }/\left\Vert
a\right\Vert ^{2},
\end{align*}%
where $P_{a^{\perp }}$ is the orthogonal projector onto $a^{\perp }$. The
term $R_{g,t}$ therefore adds $c^{2}$ to $M^{R}(x)^{\prime }M^{R}(x)$ in the
two donor directions orthogonal to $a$, and zero along $a$ itself. The
calibrated intercept contrasts survive unchanged. Over the whole covariate
grid, the smallest singular value of $M^{R}(x)$ is now bounded away from
zero.

\subsection{Imposing the condition number and building the control paths}

\label{subsec:sim-conditioning}

The contrast matrix $M^{R}(x)$ is reconditioned by adjusting its singular
values: 
\begin{align}
M^{R}(x)& =U(x)\func{diag}\{s_{1}(x),s_{2}(x),s_{3}(x)\}V(x)^{\prime },
\label{eq:appE-svd} \\
M(x)& =U(x)\func{diag}\left\{ s_{1}(x),\ \max \left[ s_{2}(x),\frac{s_{1}(x)%
}{\kappa }\right] ,\ \frac{s_{1}(x)}{\kappa }\right\} V(x)^{\prime }.
\label{eq:appE-reset}
\end{align}%
The singular value decomposition \eqref{eq:appE-svd} is replaced by \eqref{%
eq:appE-reset}, which is \eqref{eq:simulation-conditioned-matrix} with $%
\kappa $ the target condition number. The singular vectors and the largest
singular value are retained, and both steps are carried out pointwise at
each of the $101$ covariate values. The inner maximum keeps the spectrum
ordered. The realized condition number of $M(x)$ equals $\kappa $ at every
grid point.

The control-group paths are built from the reference donor's preliminary
mean outcome path and the reconditioned contrasts: 
\begin{align}
m_{N_{G}+1,t}(x)& =m_{N_{G}+1,t}^{R}(x),\qquad t\leq \mathcal{T},
\label{eq:appE-cp-ref} \\
m_{g,t}(x)& =m_{N_{G}+1,t}^{R}(x)+[M(x)]_{t,g},\qquad g=2,\ldots ,N_{G},\
t\leq \mathcal{T}-1,  \label{eq:appE-cp-pre} \\
m_{g,\mathcal{T}}(x)& =m_{N_{G}+1,\mathcal{T}}^{R}(x)+[M(x)]_{\mathcal{T}%
-1,g}+3\varsigma _{g},\qquad g=2,\ldots ,N_{G},  \label{eq:appE-cp-post}
\end{align}%
where $\varsigma _{g}$ is donor $g$'s empirical change between the final
pre-treatment and post-treatment periods, measured relative to the reference
donor's change. By \eqref{eq:appE-cp-ref}, the reference donor's preliminary
and final mean outcome paths coincide. Differencing \eqref{eq:appE-cp-post}
against \eqref{eq:appE-cp-pre} and subtracting the corresponding change in %
\eqref{eq:appE-cp-ref} gives the departure from parallel trends among the
donors, 
\begin{align}
\left( m_{g,\mathcal{T}}(x)-m_{g,\mathcal{T}-1}(x)\right) -\left( m_{N_{G}+1,%
\mathcal{T}}(x)-m_{N_{G}+1,\mathcal{T} -1}(x)\right) =3\varsigma _{g},
\label{eq:appE-pt-wedge}
\end{align}%
which does not depend on $x$. It is the only source of violation of
Assumption \ref{asm:PT} among the donors; setting $\varsigma _{g}\equiv 0$
removes it, as the second and third designs in the next subsection do.

\subsection{Mean path for the treated group, the three designs, and the
outcomes}

\label{subsec:sim-three-designs}

All three designs build the treated group's untreated path from the donors.

In DGP1 the weights \eqref{eqn:weight-construction-sim} are applied to the
donor paths of \eqref{eq:appE-cp-ref}--\eqref{eq:appE-cp-post} in every
period: 
\begin{align*}
m_{1,t}(x)& =\sum_{g=2}^{N_{G}+1}w_{g}(x)m_{g,t}(x),\quad t=1,\ldots ,%
\mathcal{T}, \\
w_{g}(x)& =\frac{0.2}{N_{G}}+0.8\cdot \frac{\exp [(g-N_{G}-1)x]}{%
\sum_{g^{\prime }=2}^{N_{G}+1}\exp [(g^{\prime }-N_{G}-1)x]},\quad
g=2,\ldots ,N_{G}+1.
\end{align*}%
Assumption \ref{asm:SC} holds exactly by construction. The wedge \eqref{%
eq:appE-pt-wedge} breaks Assumption \ref{asm:PT}.

In DGP2 the donors' paths at $t=1,\ldots ,\mathcal{T}-1$ are left unchanged,
while every donor is given the reference donor's change between $\mathcal{T}%
-1$ and $\mathcal{T}$, which sets the wedge \eqref{eq:appE-pt-wedge} to
zero. The treated group's pre-treatment path is then the weighted donor
combination plus a perturbation $0.40\,u(x)$ over the pre-treatment periods,
where $u(x)\in \mathbb{R}^{5}$ is a unit vector in the left null space of $%
M(x)$, so that no affine combination of the donors can reproduce it. The
perturbation in the final pre-treatment period is carried unchanged into the
post-treatment period, so it drops out of $m_{1,\mathcal{T}}(x)-m_{1,%
\mathcal{T}-1}(x)$, which then equals the donors' common change. Assumption %
\ref{asm:PT} therefore holds exactly. Assumption \ref{asm:SC} does not.
Because $u(x)$ is orthogonal to the column space of $M(x)$, the
least-squares solution $\bm{w}_{0}(x)$ of $M(x)\bm{w}_{0}(x)=m_{1}(x)$ in %
\eqref{eqn:M-m1-def} remains the vector of designated weights $%
(w_{g}(x))_{g=2,\ldots ,N_{G}}$.

DGP3 removes the perturbation $0.40\,u(x)$ from DGP2. The treated group's
untreated path is then the weighted donor combination in every period, and
the donors keep the common final change imposed in DGP2. Both assumptions
hold.

In all three designs, untreated outcomes are generated from these means as $%
Y_{i,t}(0)=m_{G_{i},t}(X_{i})+\varepsilon _{i,t}$, with the covariate $X_{i}$%
, the group $G_{i}$ and the errors $\varepsilon _{i,t}$ all drawn as in
Section \ref{subsec:DGPs-sim}. The observed treated outcome in the final
period adds the treatment effect to the untreated outcome, $Y_{i,\mathcal{T}%
}(1)=Y_{i,\mathcal{T}}(0)+0.7\sin (2\pi X_{i})+1$, which yields the true ATT
reported there.

\end{document}